\documentclass[twoside,leqno,twocolumn]{article}

\usepackage[letterpaper]{geometry}

\usepackage{siamproceedings}

\usepackage[T1]{fontenc}
\usepackage{booktabs}
\usepackage{amsfonts}
\usepackage{graphicx}
\usepackage{epstopdf}
\usepackage{enumitem}
\usepackage{placeins}
\usepackage{dblfloatfix}
\usepackage{algpseudocode}
\usepackage{multirow}
\ifpdf
  \DeclareGraphicsExtensions{.eps,.pdf,.png,.jpg}
\else
  \DeclareGraphicsExtensions{.eps}
\fi

\usepackage{wrapfig}
\newsiamremark{remark}{Remark}
\newsiamremark{hypothesis}{Hypothesis}
\crefname{remark}{Remark}{Remarks}
\crefname{hypothesis}{Hypothesis}{Hypotheses}
\newsiamthm{claim}{Claim}

\usepackage{amsopn}

\usepackage[dvipsnames]{xcolor}
\usepackage{todonotes}
\usepackage{comment}

\usepackage{xspace}

\usepackage{makecell}

\usepackage{mathtools}
\usepackage{amsmath}
\usepackage{amssymb}  %
\usepackage{thmtools}
\usepackage{numprint}
\npdecimalsign{.}      %

\usepackage{pgfplots}
\pgfplotsset{compat=1.18}

\usepgfplotslibrary{colorbrewer}
\usepgfplotslibrary{groupplots}
\usepgfplotslibrary{statistics}
\usetikzlibrary{positioning}

\definecolor{my-dark-red}{RGB}{183, 28, 28}
\definecolor{my-red}{RGB}{244,67,54}
\definecolor{my-pink}{RGB}{233,30,99}
\definecolor{my-purple}{RGB}{156,39,176}
\definecolor{my-deep-purple}{RGB}{103,58,183}
\definecolor{my-indigo}{RGB}{63,81,181}
\definecolor{my-blue}{RGB}{33,150,243}
\definecolor{my-light-blue}{RGB}{3,169,244}
\definecolor{my-cyan}{RGB}{0,188,212}
\definecolor{my-teal}{RGB}{0,150,136}
\definecolor{my-green}{RGB}{76,175,80}
\definecolor{my-light-green}{RGB}{139,195,74}
\definecolor{my-lime}{RGB}{205,220,57}
\definecolor{my-yellow}{RGB}{255,235,59}
\definecolor{my-amber}{RGB}{255,193,7}
\definecolor{my-orange}{RGB}{255,152,0}
\definecolor{my-deep-orange}{RGB}{255,87,34}
\definecolor{my-brown}{RGB}{121,85,72}
\definecolor{my-grey}{RGB}{158,158,158}
\definecolor{my-blue-grey}{RGB}{96,125,139}
\definecolor{my-lipics-grey}{rgb}{0.6,0.6,0.61}

\newlength{\xdist}
\newlength{\ydist}
\usepackage{tikz}
\usetikzlibrary{external}
\usetikzlibrary{fit,backgrounds,positioning,arrows.meta, calc, decorations.pathmorphing}

\usepackage{amsmath}

\pgfplotsset{
  cycle list/Dark2-6,
  cycle multiindex* list={
    mark=*\nextlist
    Dark2-6\nextlist
    every mark/.append style={
      draw=.!60!black, fill=.!60!black, opacity=0.3, mark size=1pt
    }\\
    \nextlist
  },
}
\pgfplotscreateplotcyclelist{lineplotlist}{
  {Dark2-A,densely dashed,mark=o,mark options={solid}},
  {Dark2-B,mark=pentagon*},
  {Dark2-C,densely dashed,mark=triangle*,mark options={solid}},
  {Dark2-D,mark=diamond},
  {Dark2-E,densely dashed,mark=*,mark options={solid}},
  {Dark2-F,mark=triangle},
}

\newcommand{\mwisp}[0]{\textsc{Maximum Weight Independent Set}}

\newcommand{\MWIS}[0]{MWIS}

\newcommand{\MIS}[0]{{MIS}}
\newcommand{\maxis}[0]{maximal independent set}
\newcommand{\MWISP}[0]{\textsc{MWISP}}

\newcommand{\kamis}[0]{\textsf{KaMIS\_wB\&R}}

\newcommand{\dKaMinPar}[0]{\texttt{dKaMinPar}}
\newcommand{\kamping}[0]{\texttt{KaMPIng}}
\newcommand{\htwis}[0]{\textsf{HtWIS}}

\newcommand{\KaDisReduS}[0]{\textsf{DisReduS}}
\newcommand{\KaDisReduA}[0]{\textsf{DisReduA}}
\newcommand{\KaDisReduSA}[0]{\textsf{DisReduS/A}}
\newcommand{\sG}[0]{\textsf{GS}}
\newcommand{\aG}[0]{\textsf{GA}}
\newcommand{\sRG}[0]{\textsf{RGS}}
\newcommand{\aRG}[0]{\textsf{RGA}}
\newcommand{\sRnP}[0]{\textsf{RnPS}}
\newcommand{\aRnP}[0]{\textsf{RnPA}}
\newcommand{\msgqueue}[0]{\textsf{message-queue}}

\newcommand{\myRHG}[0]{\texttt{RHG}}
\newcommand{\myRHGBig}[0]{\texttt{RHG*}}
\newcommand{\myRGG}[0]{\texttt{RGG}}
\newcommand{\myGNM}[0]{\texttt{GNM}}

\algblock[OnReceive]{OnReceive}{EndOnReceive}
\algblockdefx[OnReceive]{OnReceive}{EndOnReceive}%
    [1]{\textbf{on receive}(#1)}%
    {}

\newcommand{\ie}[0]{i.e.,\xspace}

\newcommand{\etal}[0]{et al.\xspace}

\newcommand{\reductionRow}[3][l]{\textit{#2} & \makecell[#1]{#3}}
\newcommand{\reducedGraphRow}[2][l]{\reductionRow[#1]{Reduced Graph}{#2}}
\newcommand{\offsetRow}[2][l]{\reductionRow[#1]{Offset}{#2}}
\newcommand{\reconstructionRow}[2][l]{\reductionRow[#1]{Reconstruction}{#2}}

\newcommand{\mc}[3]{\multicolumn{#1}{#2}{#3}}

\Crefname{subsection}{Section}{Sections}

\DeclareMathOperator{\w}{\omega}
\DeclareMathOperator{\rank}{\ensuremath{rank}}

\DeclareMathOperator{\I}{\ensuremath{\mathcal{I}}}
\DeclareMathOperator{\G}{\ensuremath{\mathcal{G}}}
\DeclareMathOperator{\addedI}{\ensuremath{\tilde{\mathcal{I}}}}

\newcommand{\ghosts}[1]{\ensuremath{V^g_{#1}}}
\newcommand{\rghosts}[1]{\ensuremath{V^{g\prime}_{#1}}}
\newcommand{\cV}[0]{\ensuremath{\overline{V}}}
\newcommand{\N}[0]{\ensuremath{N}}

\newcommand{\set}[1]{\ensuremath{\{#1\}}}
\newcommand{\solw}[1]{\ensuremath{\alpha(#1)}}
\newcommand{\card}[1]{\ensuremath{|#1|}}
\newcommand{\recv}[0]{\ensuremath{\mathcal{R}}}
\newcommand{\pe}[1]{\ensuremath{#1}}

\newtheorem{@reduction}{Reduction}[section]

\newtheorem{@distributedreduction}{Reduction}[section]
\newenvironment{distributedreduction}[1][\unskip]{\begin{@distributedreduction}[#1]\noindent\newline\noindent\hspace*{-.18cm}}{\end{@distributedreduction}}

\newif\ifshowFull
\showFulltrue %

\begin{document}

\newcommand\relatedversion{}

\title{\Large Distributed Reductions for the Maximum Weight Independent Set Problem\relatedversion}
    \author{Jannick Borowitz\thanks{Karlsruhe Institute of Technology (\email{jannick.borowitz@kit.edu}, \email{schimek@kit.edu}, \url{https://kit.edu}).}
    \and Ernestine Gro\ss mann\thanks{University of Heidelberg (\email{e.grossmann@informatik.uni-heidelberg}, \url{https://uni-heidelberg.de}).}
    \and Matthias Schimek\footnotemark[2]}

\date{}

\maketitle

\fancyfoot[R]{\scriptsize{Copyright \textcopyright\ 2026 by SIAM\\
Unauthorized reproduction of this article is prohibited}}

\begin{abstract} 
 Finding maximum-weight independent sets in graphs is an important NP-hard optimization problem.
 Given a vertex-weighted graph $G$, the task is to find a subset of pairwise non-adjacent vertices of $G$ with maximum weight.
 Most recently published practical exact algorithms and heuristics for this problem use a variety of data-reduction rules to compute (near-)optimal solutions.
 Applying these rules results in an equivalent instance of reduced size. An optimal solution to the reduced instance can be easily used to construct an optimal solution for the original input.

 In this work, we present the first distributed-memory parallel reduction algorithms for this problem, targeting graphs beyond the scale of previous sequential approaches.
 Furthermore, we propose the first distributed reduce-and-greedy and reduce-and-peel algorithms for finding a maximum weight \hbox{independent set heuristically.}

 In our practical evaluation, our experiments on up to $1024$ processors demonstrate good scalability of our distributed reduce algorithms
 while maintaining good reduction impact.
 Our asynchronous reduce-and-peel approach achieves an average speedup of $33\times$ over a sequential state-of-the-art reduce-and-peel approach  on 36 real-world graphs
 with a solution quality close to the sequential algorithm.
 Our reduce-and-greedy algorithms even achieve average speedups of up to $50\times$ at the cost of a lower solution quality.
 Moreover, our distributed approach allows us to consider graphs with more than one billion vertices and 17 billion edges.

  The source code is available on Github\footnote{The latest version of the source-code is available on \protect\url{https://github.com/jabo17/kadisredu}.} and as a permanent record on Zenodo\footnote{Permanent records of the source code and the reproducibility repository are available at \protect\url{https://doi.org/10.5281/zenodo.17296045} and \protect\url{https://doi.org/10.5281/zenodo.17310080}, respectively.}.
\end{abstract}

\section{Introduction.}
\label{sec:introduction}
A {\mwisp} (\MWIS) is a subset of vertices of maximum weight, where no two vertices are adjacent.
The problem applies to many fields of interest, such as map-labeling~\cite{TemporalMapLaBarth2016, EvaluationOfLGemsa2014} or vehicle routing~\cite{NewInstancesFDong2021}.
As an NP-hard problem~\cite{SomeSimplifiedGarey1974}, finding an {\MWIS} can be challenging.
Depending on the application, a \emph{maximal} independent set of large weight,~\ie where no vertex can be added anymore, can be sufficient.
Most heuristic and exact solvers use \emph{reduction} algorithms as preprocessing or as a subroutine in between.
Reduction approaches test and apply data reduction rules to reduce the problem size while maintaining optimality.
To be more specific, the {\MWIS} can be reconstructed given an {\MWIS} for \hbox{the reduced graph.}

As instance sizes have grown significantly in recent years, we soon encounter the problem of instances not fitting into the memory of a single machine.
This further complicates the task of finding high-quality independent sets.
To overcome this problem, the graph may be represented in a distributed-memory machine where each process stores only a subgraph of the input graph.
It is also conceivable that the graph is already stored in distributed memory when computing an \MWIS{} as a building block of a more complex (distributed) application.
To find (near-)optimal solutions in this case, we develop distributed reduction techniques to apply data reduction rules in a distributed fashion.
Many sequential data reduction rules operate very locally, making them prime candidates \hbox{for distributed-memory applications.}

\paragraph{Contribution.}
\label{sec:contribution}
We present the first data reduction algorithms for the \MWIS{} problem in the distributed-memory setting.
Furthermore, we use these reduction algorithms to develop the first distributed reduce-and-greedy and reduce-and-peel solvers for this problem.
For our different algorithms, we present synchronous and asynchronous approaches, as well as a more powerful configuration with graph partitioning, which further reduces the size of graphs at the cost of \hbox{additional running time.}

In our practical evaluation on up to 1\,024 processors, we show that our algorithms yield substantial speedups over the sequential state-of-the-art algorithm while reaching a solution close to it.
Furthermore, we demonstrate that our approach is able to reduce massive graphs beyond the scale of the current \hbox{state-of-the-art (sequential) algorithm.}

\paragraph{Paper Outline.}%
After discussing related work (\Cref{sec:relatedWork}), we introduce basic notation and concepts in \Cref{sec:preliminaries}.
In \Cref{sec:distributedReductions}, we first introduce a distributed reduction model (\Cref{sec:distributedReductionModel}) which is then used to define specific distributed reduction rules (\Cref{sec:distributedReductionRules}).
We then describe two distributed reduction algorithms, using the previously defined rules to reduce the input graph in \Cref{sec:distReductionAlgos}.
\Cref{sec:distributedSolvers} briefly discusses how the reduction algorithms can be extended to complete \MWIS{} solvers.
In \Cref{sec:experiments}, we present the results of our experimental evaluation, followed by a brief conclusion in \Cref{sec:conclusion}.

\section{Related Work}
\label{sec:relatedWork}

Data reduction rules have become a very important tool in solving the \textsc{Maximum (Weight) Independent Set} problem. There is a long list of sequential exact algorithms that make use of data reduction rules~\cite{BranchAndReduAkiba2016,ThereAndBackFigiel2022,BoostingDataRGellne2021,TargetedBranchHespe2021,ExactlySolvingLamm2019,TargetedBranchLanged2024, ApplicationOfLiuJ2023,EfficientReducXiao2021}.
The most widely used exact solver for the {\MWISP} is probably the branch-and-reduce solver {\kamis} by Lamm~\etal~\cite{ExactlySolvingLamm2019}.
We use {\kamis} to find exact solutions for subproblems in our algorithms.
Additionally, many heuristics benefit from data reductions, and these techniques are incorporated in many recent works~\cite {AcceleratingLoDahlum2016,FindingNearOpGrossma2023,TowardsComputiGuJi2021,FindingNearOpLamm2017}. Next to these heuristics that directly apply data reduction rules, various local search techniques have been introduced for the problem~\cite{AMetaheuristicDong2022,AcceleratingReGrossma2024,AHybridIteratNoguei2018}. These methods themselves do not provide direct data reductions. However, in~\cite{AcceleratingReGrossma2024}, the authors show that for many instances, these approaches can also benefit from applying data reductions initially before starting the local search.
For a detailed overview of solvers and data reduction rules for the \textsc{Maximum Weight Independent Set} problem, we refer to a survey by Gro\ss mann~\etal\cite{AComprehensiveGrossma2024}.
For more general information on advances in data reductions, we refer the reader to the survey by Abu-Khzam~\etal~\cite{RecentAdvancesAbuKh2022}.

Although there is considerable work on data reduction rules for sequential approaches, there is a notable lack of research that applies these promising techniques to distributed algorithms for the \textsc{Maximum Weight Independent Set} problem.
This work represents a first step in this direction, where we utilize existing data reduction rules to apply them in a distributed setting.

\paragraph{Parallel and Distributed Approaches.}\label{sec:parallel-approaches}
A simple parallel randomized algorithm for finding a {\maxis} is due to Luby~\cite{ASimpleParallLuby1985}, which improves upon previous results~\cite{DBLP:conf/stoc/KarpW84}.
In each iteration of the algorithm, an independent subset of vertices $I$ is computed in parallel and added to the solution while $I$ with all its neighbors is removed from the graph.
Conflicts arising during the computation of $I$ are resolved by using the vertex degree and the vertex ID for breaking ties.
The algorithm requires $\mathcal{O}(\log{n})$ iterations in expectation, until the graph becomes empty.
Using $\mathcal{O}(m)$ processors, this results in an expected (poly-)logarithmic running time in the PRAM model.

Luby's algorithm is also the starting point for multiple distributed-memory parallel algorithms~\cite{8MaximalInde2000, DBLP:journals/dc/MetivierRSZ11, DistributedNeaWang2023} solving the {\maxis} problem.

For the weighted case, a distributed algorithm was proposed by Joo~\etal\cite{DistributedGreJooC2016}, which is specifically designed for wireless scheduling applications.

Next to these purely heuristic algorithms, Gro\ss mann~\etal\cite{AcceleratingReGrossma2024} contributed a shared-memory iterated local-search solver \textsf{CHILS} that consists of two phases.
In the first phase, $k$ different solutions are determined or improved with a baseline local-search algorithm on the input instance.
In the second phase, $k$ new solutions are computed on the \emph{difference-core}. It is an induced subgraph containing vertices $v\in V$ where there is a solution on the input instance that contains $v$, and there is also a solution that does not contain $v$.
Afterward, the solutions on the difference-core can be embedded back into the solutions on the input instance.
The $k$ solutions on the input instance and the difference-core can be computed in parallel, where all threads read from the graph, but do not modify it.
In the end, the best solution found is returned.

\subparagraph*{Parallel and Distributed Reductions.}
Although we are not aware of any distributed- or shared-memory parallel reduction-based algorithms for the {\MWIS} problem, two approaches exist for the unweighted problem.

Hespe~\etal\cite{ScalableKernelHespe2019} present a shared-memory algorithm.
Their key idea is to partition the input graph so that each processor can apply reductions at its vertex-disjoint block independently.
This idea is based on the observation that the data reduction rules by Akiba and Iwata~\cite{BranchAndReduAkiba2016} and Butenko~\etal\cite{FindingMaximumButenk2002} act very locally.
Thus, data reductions can be applied blockwise,~\ie each processor applies reductions within its block, given a shared graph representation.
Close to the border, reductions can only be applied with restrictions, as they cannot modify the neighborhood or vertices at other blocks without risking race conditions.
In addition to the local reductions, the authors also implement a parallel global reduction using parallel matching algorithms.
In their bachelor's thesis, George~\cite{DistributedKerGeorge2018} proposes a distributed-memory reduction algorithm.
Again, the input graph is partitioned into $p$ vertex-disjoint blocks where each of the $p$ processes owns exactly one block.
The key idea of their reduction algorithm is to reduce all degree one and two vertices exhaustively using the reduction rules by Chang et al.~\cite{ComputingANeaChang2017}.

While their approach works well in terms of running time for street networks, it does not scale for many other graph families.
Furthermore, experiments were only conducted on up to $64$ processors.

\section{Preliminaries.}
\label{sec:preliminaries}

We consider undirected, vertex-weighted graphs $G=(V,E,\w)$ with vertex set $V=\set{1,\,\ldots,\, n}$, edges $E\subseteq \binom{V}{2}$, and weight function $\w:V\to \mathbb{N}_{\geq 0}$. We extend $\w$ to sets by $\w(U)=\sum_{u\in U}{\w(u)}$ for $U\subseteq V$.
Let $\N(v)=\set{u\in V: \set{v,u}\in E}$ be the \emph{neighborhood} of a vertex $v\in V$ and $\N[v]=\N(v)\cup \set{v}$ its \emph{closed neighborhood}.
For a subset of vertices $U\subseteq V$, the definitions generalize to $\N(U)=\bigcup_{u\in U}{\N(u)}\setminus U$ and $\N[U]=\N(U)\cup U$ of $U$, respectively.
The degree of a vertex $v\in V$ is \hbox{defined as $\deg(v)= \card{\N(v)}$} and ${\Delta(G)= \max \set{\deg(v):v\in V}}$.
Often we consider subgraphs that are \emph{induced} on a subset of vertices $U\subseteq V$.
Given $G$ and $U$ the \emph{induced subgraph} is defined as ${G[U]= (U, E\cap \binom{U}{2})}$.
For simplicity, we use the common notation ${G-v = G[V\setminus \set {v}]}$ \hbox{and $G-U = G[V\setminus U]$.}
A subset $C\subseteq V$ is a \emph{clique} in $G$ if all vertices of $C$ are pairwise adjacent,~\ie $\set{u,v}\in E$ for all $u,v\in C$.

An \emph{independent set} $\I\subseteq V$ is a set that contains no adjacent vertices.
Further, $\I$ is maximal if no further vertex of $V\setminus \I$ can be added to $\I$.
It is a \emph{maximum independent set} (\MIS) if there is no independent set of $G$ that contains more vertices.
An independent set $\I$ is called a \emph{maximum weight independent set} (\MWIS) if there exists no independent set of larger weight.
We denote the weight of a maximum weight independent set as~$\solw{G}$.

\paragraph{Machine Model and Input Graph.}
\label{sec:machineModelAndInput}
We assume a \emph{distributed-memory} machine model consisting of $p$ processing elements (PEs) that allow single-ported, point-to-point communication.

The input graph is represented as a vertex-weighted, directed graph in the adjacency array format together with a 1D-vertex partition $(V_1,\ldots,V_p) = V$.
More precisely, we represent each undirected edge $\set{u,v}\in E$ by two directed edges, $(u,v)$ and $(v,u)$.
Usually, the partition is chosen such that the number of vertices or edges per PE are balanced.

Each PE $i$ stores a subgraph where PE $i$ \emph{owns} the vertices of the partition $V_i$.
Every vertex $u\in V_i$ has $\rank(u)=i$ and is called a \emph{local vertex} of PE $i$.
The subgraph at PE $i$ consists of all the edges $\set{u,v}\in E$ with vertex $u \in V_i$.
An edge $e = \set{u,v}$ is local, if $u$ and $v$ are local vertices.
Otherwise, $e$ is a \emph{cut}-edge.
A local vertex $u$ which is incident to a cut-edge $\set{u,v}$, is an \emph{interface vertex}.
It is replicated as \emph{ghost vertex} (or simply \emph{ghost}) at PE $\rank(v)$.
The set $\ghosts{i}\subseteq V$ are the ghosts of PE $i$.
We refer to the conjunction of ghosts and interface vertices of PE $i$ as border vertices of PE~$i$.

We often need to know weights of ghosts.
Therefore, we replicate the weights with the interface vertices where $\w_i(u)$ denotes the replicated weight of a ghost $u\in \ghosts{i}$.
Furthermore, we replicate local parts of the neighborhoods of a ghost $u\in\ghosts{i}$,~\ie $\N(u)\cap V_i$ is stored as the neighborhood of $u$.
Overall, this yields at PE $i$ a subgraph $G_i=(\cV_i, E_i, \w_i)$ with vertices $\cV_i= V_i\cup \ghosts{i}$ \hbox{and edges ${E_i=\set{\set{u,v}\in E: u\in \cV_i, v\in V_i}}$.}
Furthermore, for $u\in V_i$, we define the set of \emph{adjacent PEs of} $u$ as $\recv(u)=\set{\rank(v):v\in \N(u)\cap \ghosts{i}}$.

\section{Distributed Reductions.}\label{sec:distributedReductions}
Many sequential data reduction rules act very locally~\cite{ScalableKernelHespe2019}.
They can often decide whether a rule applies given only the neighborhood of a vertex.
Moreover, if such a rule applies, the modification to the graph is also usually restricted to that neighborhood.

Regarding parallelization, the property of locality is helpful as each PE $i$ can locally reduce (parts of) its subgraph $G_i$ of the overall input graph $G$.
In the following, we propose distributed reductions for $G_i$ at PE $i$ which are based on sequential counterparts.

After briefly discussing the handling of non-border vertices (\Cref{sec:reduction-non-border-vertices}), we introduce the overall reduction model (\Cref{def:distRedModel}) defining a set of allowed graph modifications which is not depending on any problem specific reductions and will also allow the modification of border vertices.
Then, in \Cref{sec:distributedReductionRules}, we describe the different distributed reductions for the MWIS problem.

\subsection{Reducing Non-Border Vertices.}
\label{sec:reduction-non-border-vertices}

Since the sequential data reduction rules we consider here all act locally, we can apply these rules on each $G_i$ independently, as long as they do not change the border. This leads us to the first distributed reduction, a meta-rule that allows PE $i$ to apply any sequential data reduction rule to $G_i$ if the rule acts sufficiently locally.
In general, the meta rule can be used to transfer (new) reduction rules to the distributed-memory model, which are not covered by the work out-of-the-box.
\begin{distributedreduction}[Meta Rule]
	\label{dred:MetaRule}
	Consider an application of a sequential reduction rule in $G_i$.
	Assume ghosts and interface vertices were not read to decide whether the reduction rule is applicable.
	Then, this reduction application is correct in $\G$.
\end{distributedreduction}
\begin{proof}
	If only local non-interface vertices are read to decide whether the reduction rule is applicable in $G_i$,  the subgraph considered is local and therefore not modified by any other PE.
	Thus, we are allowed to apply the sequential reduction rule in $\G$, yielding a reduced graph $\G'$.
  To represent $\G'$ in the distributed memory model, we only need to modify $G_i$ because the \hbox{reduction is local.}
\end{proof}

However, we also want to apply reductions to border vertices.
Reducing vertices closer to or at the border requires considering what other PEs might have reduced at the border.
Reducing border vertices is more challenging because information about the border is limited.
Furthermore, concurrent modifications at the border can lead to incorrect assumptions that affect the reduction decisions.
To that end, we propose a distributed reduction model in \Cref{sec:distributedReductionModel}, which specifies the extent to which distributed reductions can modify the border.
In \Cref{sec:distributedReductionRules}, we present distributed reduction rules that reduce the graph close to and at the border.

\subsection{Distributed Reduction Model.}
\label{sec:distributedReductionModel}
In our distributed reduction model, a distributed reduction rule is a rule that operates on the local graph $G_i$ at PE $i$.
For border vertices, we restrict the allowed modifications of such reduction applications.
For example, we do not allow the insertion of new cut edges or an increase in vertex weight.
These restrictions allow us to make assumptions about the modifications at the border by other PEs.
As a result, our distributed reductions can often exclude interface vertices or even propose to include an interface vertex.
Moreover, this can be done locally,~\ie we do not need to communicate between the reduction steps to approve the reduction decisions.
Communication is only required to agree on the proposed vertices for inclusion and to synchronize reduction progress, allowing potentially more reductions to be applied again.

\textit{Proposal to Include.}
An important reduction operation is the proposition to include an interface vertex into an MWIS.
We assume that an interface vertex $v$ can only be \emph{proposed to be included} if $\w_i(v)\geq \solw{G_i[N_i(v)]}$.
A proposal to include $v$ only succeeds if none of its ghost neighbors are proposed to be included or if $v$ wins a tie-breaking between the conflicting proposals.
More formally, if $v$ is proposed to be included, $v$ joins the set of \emph{solution proposals}~$\addedI_i$ and $N_i[v]$ is reduced from $G_i$,~\ie ${G_i'=G_i-N_i[v]}$.
Once the adjacent PEs finished reducing $N_i[v]$, PE $i$ can inquiry whether a ghost neighbor $u\in N_i(v)\cap \ghosts{i}$ was proposed to be included at PE $j=\rank(u)$,~\ie $u\in \addedI_j$.
If such a neighbor $u$ exists, we tie-break between these two vertices and choose $v$ if and only if $i<j$.
A PE can only reduce ghosts if a neighbor (interface vertex) is proposed to be included.

For sequential reductions, degree-one vertices are always reducible. However, in the distributed model, this is not the case directly.
To still be able to reduce these vertices, we introduce the move operation.

The following definition specifies the allowed modifications in $G_i$ by PE $i$.

\begin{definition}[Distributed Reduction Model]\label{def:distRedModel}
	PE $i$ can modify the subgraph $G_i$ arbitrarily, except that modifications for an interface vertex $v\in V_i$ are restricted to the following ones:
	\begin{itemize}
		\item $v$ can be reduced, if it is excluded or proposed to be included,
		\item the weight of $v$ can be decreased,
		\item and $N_i(v)$ can be changed if no ghost is added to the neighborhood.
	\end{itemize}
  Moreover, ghosts are only reduced if a neighbor is proposed for inclusion.
\end{definition}

After applying these changes, the (modified) global graph $\G'=(V',E',\w')$ is defined as the union of the modified local graphs $G'_1,\,\ldots,\,G'_p$.
More precisely, we have $V'=\bigcup_i V_i'\setminus (\bigcup_{j\neq i} \ghosts{j}\setminus \rghosts{j})$ %
with edges $E'=\bigcup_i E_i' \cap \binom{V'}{2}$ and $\w'(v)=\w_{\rank(v)}'(v)$ for $v\in V'$.
Note that the removal of the vertices $\left(\bigcup_{j\neq i} \ghosts{j}\setminus \rghosts{j}\right)$ from the union of all local vertices accounts for interface vertices $u$ being removed on a PE $j \neq \rank(u)$.
In the remainder of this work, we will refer to the current state of the modified graph as $\G$.%

In this section, we introduce different lemmata that are needed for the reductions. They describe key properties that follow from the distributed reduction model. We give an intuition for these here and the formal proofs 
\ifshowFull
in Appendix~\ref{appendix:sec:proofs}.
\fi

An important part of our reductions involves weight comparisons. Within the local graphs $G_i$, we know the exact weights of all vertices. Initially, each PE also knows the exact weights for ghost vertices, however, the weight of ghost vertices can be reduced. This means that each PE only knows an upper bound for the weights of ghost vertices.

\begin{lemma}[Upper Bound for Ghost Weights]\label{lem:weights}
	Let $v\in \ghosts{i}$.
	It follows that $\w_i(v)\geq \w_{\rank(v)}(v)$ or $v$ is globally reduced,~\ie $v\not\in V$.
\end{lemma}

Initially, the neighborhood of a local vertex in $G_i$ is identical to the neighborhood in $G$.
During the reduction applications, border vertices might be reduced by other PEs.
These reduced vertices may remain in $G_i$ until the reduction applications are communicated by the adjacent PEs.
Consequently, we always know all the remaining neighbors in $\G$ for a local vertex $v\in V_i$.
At worst $\N_i(v)$ contains additional globally reduced vertices.
\Cref{lem:neighborhood} summarizes this observation.

\begin{lemma}[Neighborhood]\label{lem:neighborhood}
	Let $v\in V$.
	If $v\in V_i$, it holds $N(v)\subseteq N_i(v)$ where the vertices of $N_i(v)\setminus N(v)$ are already reduced \hbox{by other PEs.}
	If $v\in \ghosts{i}$, it holds $N(v)\cap V_i \subseteq N_i(v)$.
\end{lemma}

\begin{figure*}[t!]
	\input{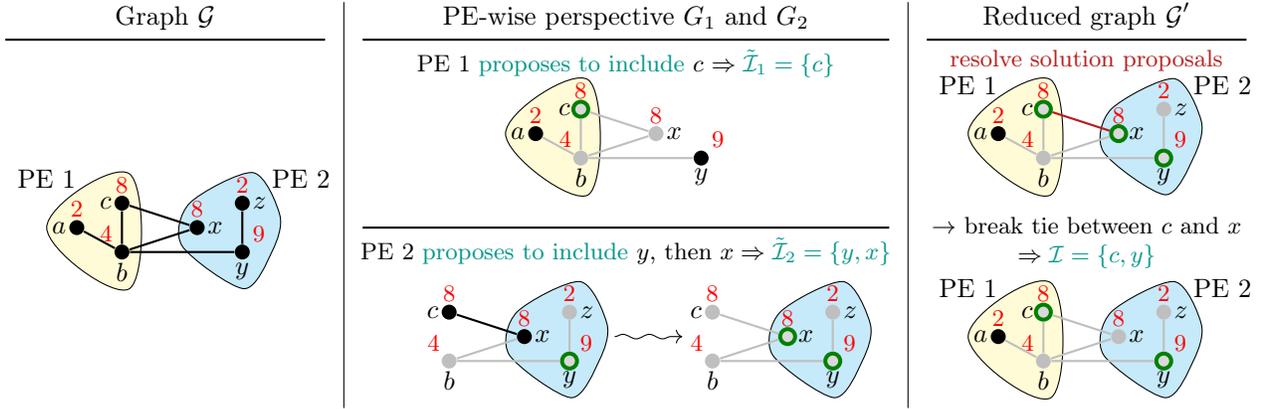}
	\caption{Graph $\G$ is distributed over PE $1$ and $2$ (left). PE $1$ proposes to include vertex $c$; PE $2$ proposes to include $y$  and then $x$ (middle).
		Note that at the time when $c$ and $x$ are proposed to be included, the only neighbor $b$ of $c$ in $G_i$ was already reduced by PE $2$.
		Thus, proposing to include $c$ and $x$ reduces globally an isolated cut-edge as stated in \Cref{lem:IncInterfaceVertices}.
    Since both have equal weight, we can simply use the PE ranks to break ties between $c$ and $x$ for an MWIS $\I$ of $\G$. This is described in Remark~\ref{rem:includeVertex} in more detail.}
	\label{fig:solutionConflict}
\end{figure*}

\Cref{lem:IncInterfaceVertices} addresses potential conflicts regarding an MWIS for $\G$ if the incident vertices of a cut-edge $\set{u,v}\in E$ are both proposed to be included.
Intuitively, such a conflict can only happen if $N_i[v]$ forms a clique in $G_i$ where $v\in V_i$ and $u\in N_i(v)$ have equal weight.
Moreover, $u$ must be the only remaining neighbor of $v$ in $\G$.
Thus, both proposed to be included vertices, $u$ and $v$ reduce only an isolated cut-edge $\set{u,v}$ in $\G$.
Since both have equal weight we can reconstruct the solution with either of them.
An example is \hbox{shown in \Cref{fig:solutionConflict}.}

\begin{lemma}[Include Alternatives]\label{lem:IncInterfaceVertices}
	Let $\G=(V,E,\w)$ be the global reduced graph.
	Let $G_i=(\cV_i, E_i, \w_i)$ the local graph at PE $i$ and $G_j=(\cV_j, E_j, \w_j)$ the local graph at PE $j\neq i$.
	Assume, it exists $v\in V_i$ and $u\in V_j$ which are now proposed to be included at PE $i$ and PE $j$, respectively.
	Further, assume $v\in N_j(u)$ at PE $j$ so that $v$ is reduced when $u$ is proposed to be included.
	Then, it holds $u\in N_i(v)$, $\w_i(v)=\w_j(u)$, $N_i(v)\cap V \subseteq \set{u}$, and $\N_j(u)\cap V \subseteq \set{v}$.
	If $u$ exists for the given $v$, then there is no local graph $G_k=(\cV_k, E_k, \w_k)$ at some PE $k\not\in\set{i,j}$ with a vertex $u'\in V_k\setminus \set{u}$ that reduces $v$ as well.
\end{lemma}

So far we introduced the modification to propose to include a vertex.
Now, we present our procedure to include vertices into an MWIS of $\G$.
This is used by all our distributed data reductions which aim \hbox{to include vertices.}

\begin{lemma}[Include Operation]\label{lem:includeOperation}%
	Let ${\G=(V,E,\w)}$ be the global reduced graph and let ${G_i=(V_i,E_i,\w_i)}$ be the local graph at PE $i$.
	Furthermore, let $v\in V_i$ be a vertex that was not yet reduced by any PE, \ie $v \in V$.
	If $v$ is an interface vertex, assume $\w_i(v)\geq \solw{G_i[\N_i(v)]}$ in addition.
	Assume, $v$ can be included into an MWIS of $\G$ so that $\G$ is reduced to $\G'=\G-N(v)$ with ${\solw{\G} = \solw{\G'}+\w_i(v)}$.

	Then, only PE $i$ needs to modify its local graph so that $\G'$ is obtained globally.
	The local reduced graph is given as ${G_i'=G_i-N_i[v]}$.
	Furthermore, add $v$ to the set ${\addedI_i'=\addedI_i\cup \set{v}}$ at PE $i$ if $v$ is an interface vertex.
	An MWIS of $\G$ can be reconstructed after exchanging the solution proposals of $N_i[v]$ with the adjacent PEs of $v$.
	Reconstruct the solution as follows:
	\begin{itemize}
		\item If $u$ exists with $j < i$, set ${\I_i = \I_i'}$ at PE $i$ and~${\I_j = \I_j' \cup \set{u}}$ at PE~$j$.
		\item Else if $u$ exists with $j > i$, set ${\I_i = \I_i'\cup \set{v}}$ at PE~$i$ and ${\I_j = \I_j'}$ at PE~$j$.
		\item Else, add $v$ to the local solution,~\ie $\I_i = \I_i'\cup \set{v}$.
	\end{itemize}
\end{lemma}

\begin{remark}[Include Operation]\label{rem:includeVertex}
	\Cref{lem:includeOperation} provides the include operation for a vertex $v\in V_i$ at PE $i$ if it is not yet reduced by any PE,~\ie $v\in V$.
	Instead of testing for an interface vertex $v$ whether $v\in V$, we propose to include $v$ without a test.
	This way we do not need extra communication for testing $v\in V$.
	Moreover, we do not violate optimality by proposing $v$ because $\Cref{lem:IncInterfaceVertices}$ applies.
	There is exactly one vertex that was proposed to be included and reduced $v$.
	It is a neighbor of $v$ in $G_i$, both have equal weights and other neighbors in the local graph are already reduced.
	Thus, $v$ is an alternative for this neighbor in an MWIS.
	With the tie-breaking, described in \Cref{lem:includeOperation}, the exactly one proposal is accepted for the MWIS.

\end{remark}

Next, we consider an interface vertex $u\in V_i$ and suppose it has a large weight.
If $v$ has a local neighbor $v\in N_i(u)$ that has not been reduced yet in $\G$, then $u$ has not been reduced by any PE either.
Our distributed reduction rules make use of this property to reduce a vertex $v$ which requires that $u$ is not reduced yet.
On the other hand, if $v$ was already reduced by another PE in $\G$, such a data reduction would only remove an already reduced vertex from $G_i$.
\begin{lemma}[Remaining Vertex]\label{lem:remainingVertex}
	Let $u,v\in V_i$ be adjacent with $\max\set{\w_i(x): x\in (N_i(u)\cap \ghosts{i})\setminus N_i(v)} + \w_i(v) \leq \w_i(u)$.
	Then it holds that $u\in V$ if $v\in V$.
\end{lemma}

\subsection{Distributed Reduction Rules.}\label{sec:distributedReductionRules}

In this section, we introduce the reduction rules that can reduce border vertices. 
\ifshowFull
To make this section more readable the proofs can all be found in the Appendix~\ref{appendix:sec:reductionProofs}.
\fi
The first rule in that regard is the \emph{Heavy Vertex} rule. It was first introduced by Lamm~\etal~\cite{ExactlySolvingLamm2019} for the sequential model.
The rule involves finding an MWIS for the induced neighborhood graph $G_i[N_i(v)]$.
\Cref{fig:exampleHeavyVertex} shows an example for \nameref{dred:HeavyVertex}.
We now show how it can also apply to interface vertices.

\begin{distributedreduction}[Distributed Heavy Vertex]\label{dred:HeavyVertex}
	Let $v\in V_i$ with $\w_i(v)\geq \solw{G_i[N_i(v)]}$.
	Then (propose to) include $v$ following Remark~\ref{rem:includeVertex}.
\end{distributedreduction}

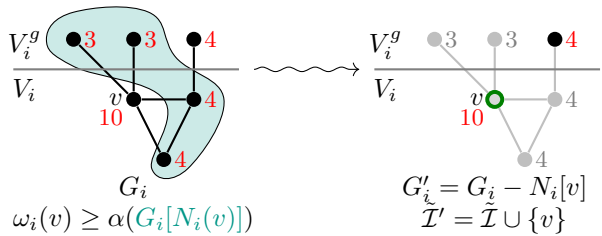
\begin{figure}[h]
	\centering
	\begin{tikzpicture}[
		minpoint/.style={inner sep=0pt, outer sep=0pt, draw=none},
		point/.style={circle, fill=black, inner sep=2pt},
		reduced/.style={circle, fill=gray!50, inner sep=2pt},
		included/.style={circle, draw=green!50!black, fill=gray!30, inner sep=2pt, line width=1.5pt},
		edge/.style={thick},
		reducededge/.style={edge, draw=gray!50},
		nodeweight/.style={font=\small, text=red, label distance=0.0mm},
		reducednodeweight/.style={font=\small, text=gray, label distance=0.0mm},
		pe2/.style={cyan!20,draw=black},
		pe1/.style={yellow!20,draw=black},
	]
	\pgfdeclarelayer{background}
	\pgfsetlayers{background,main}

	\begin{scope}
		\coordinate (origin) at (0,0) {};
		\node[point]  (a) at ($(origin) + (1.0\xdist, 0)$) {};
		\node[minpoint,label={[nodeweight]right:$3$}]   at (a) {};
		\node[point]  (b) at ($(origin) + (3.0\xdist, 0)$) {};
		\node[minpoint, label={[nodeweight]right:$3$}]   at (b) {};
		\node[point]  (c) at ($(origin) + (5.0\xdist, 0)$) {};
		\node[minpoint, label={[nodeweight]right:$4$}]   at (c) {};
		\node[point]  (v) at ($(origin) + (3.0\xdist, -2\xdist)$) {};
		\node[minpoint, label={left:$v$}, label={[nodeweight]200:$10$}]   at (v) {};
		\node[point]  (d) at ($(origin) + (5.0\xdist, -2\xdist)$) {};
		\node[minpoint, label={[nodeweight]right:$4$}]   at (d) {};
		\node[point]  (e) at ($(origin) + (4.0\xdist, -4\xdist)$) {};
		\node[minpoint, label={[nodeweight]right:$4$}]   at (e) {};

		\draw[edge] (a) -- (v);
		\draw[edge] (b) -- (v);
		\draw[edge] (c) -- (d);
		\draw[edge] (v) -- (d);
		\draw[edge] (v) -- (e);
		\draw[edge] (d) -- (e);

		\coordinate (l1) at ($(origin) +(-1\xdist, -\xdist)$);
		\coordinate (l2) at ($(origin) +(+6.5\xdist, -\xdist)$);
		\coordinate (l2') at ($(origin) +(+7\xdist, -\xdist)$);
		\draw[edge, gray] (l1) -- node[pos=0.05, above, color=black] {$V_i^g$} node[pos=0.05, below, color=black] {$V_i$}  (l2);
		\node[align=center] at ($(e) + (-\xdist,-1.5\xdist)$) {$G_i$\\ $\omega_i(v) \ge \alpha(\textcolor{my-teal}{G_i[N_i(v)]})$};
	\end{scope}
	\begin{scope}[on background layer, opacity=0.1]
		\coordinate (C0-0) at ($(b)+(+0.5\xdist,+0.9\xdist)$);
		\coordinate (C0-1) at ($(a)+(-0.0\xdist,+0.9\xdist)$);
		\coordinate (C0)   at ($(a)+(-0.9\xdist,0\xdist)$);
		\coordinate (C0-2) at ($(a)+(-0.0\xdist,-0.9\xdist)$);
		\coordinate (C0-5) at ($(a)!0.5!(v)$);
		\coordinate (C1)   at ($(b)!0.5!(v)$);
		\coordinate (C2)   at ($(v)!0.5!(d)$);
		\coordinate (C3)   at ($(e)+(-0.4\xdist,-0.3\xdist)$);
		\coordinate (C3-5) at ($(e)+(+0.8\xdist,-0.3\xdist)$);
		\coordinate (C4)   at ($(d)+(\xdist,0)$);
		\coordinate (C5)   at ($(d)!0.5!(c)$);
		\coordinate (C6)   at ($(v)!0.7!(c)$);
		\fill[my-teal!20, draw=black] plot[smooth cycle, tension=.7] coordinates {(C0-0) (C0-1) (C0) (C0-2)  (C2) (C3) (C3-5) (C4) (C6) };

	\end{scope}
	\begin{scope}[xshift=12\xdist]
		\coordinate (origin) at (0,0) {};
		\node[reduced]  (a) at ($(origin) + (1.0\xdist, 0)$) {};
		\node[minpoint,label={[reducednodeweight]right:$3$}]   at (a) {};
		\node[reduced]  (b) at ($(origin) + (3.0\xdist, 0)$) {};
		\node[minpoint, label={[reducednodeweight]right:$3$}]   at (b) {};
		\node[point]  (c) at ($(origin) + (5.0\xdist, 0)$) {};
		\node[minpoint, label={[nodeweight]right:$4$}]   at (c) {};
		\node[included]  (v) at ($(origin) + (3.0\xdist, -2\xdist)$) {};
		\node[minpoint, label={left:$v$}, label={[nodeweight]200:$10$}]   at (v) {};
		\node[reduced]  (d) at ($(origin) + (5.0\xdist, -2\xdist)$) {};
		\node[minpoint, label={[reducednodeweight]right:$4$}]   at (d) {};
		\node[reduced]  (e) at ($(origin) + (4.0\xdist, -4\xdist)$) {};
		\node[minpoint, label={[reducednodeweight]right:$4$}]   at (e) {};

		\draw[reducededge] (a) -- (v);
		\draw[reducededge] (b) -- (v);
		\draw[reducededge] (c) -- (d);
		\draw[reducededge] (v) -- (d);
		\draw[reducededge] (v) -- (e);
		\draw[reducededge] (d) -- (e);

		\coordinate (l1') at ($(origin) +(-1.5\xdist, -\xdist)$);
		\draw[->,decorate, decoration={snake,amplitude=0.3mm}] (l2') -- (l1');

		\coordinate (l1) at ($(origin) +(-1\xdist, -\xdist)$);
		\coordinate (l2) at ($(origin) +(+6.5\xdist, -\xdist)$);
		\draw[edge, gray] (l1) -- node[pos=0.05, above, color=black] {$V_i^g$} node[pos=0.05, below, color=black] {$V_i$}  (l2);
		\node[align=center] at ($(e) + (-\xdist,-1.5\xdist)$) {$G_i' = G_i - N_i[v]$\\ $\tilde{\mathcal{I}}' = \tilde{\mathcal{I}} \cup \{v\}$};

	\end{scope}

\end{tikzpicture}
	\caption{Example of \nameref{dred:HeavyVertex} in the perspective of PE $i$. The reduction {\color{my-teal}proposes to include} the interface vertex $v\in V_i$ in $G_i$ because ${\color{red}10}=\w_i(v)\geq\solw{G_i[N_i(v)]}={\color{red}10}$. The closed neighborhood $N_i[v]$ is reduced from $G_i$ which \hbox{results in $G_i'=G_i-N_i[v]$.}}
	\label{fig:exampleHeavyVertex}
\end{figure}
Finding an MWIS for the neighborhood graph is often inefficient and computational more efficient bounds are sufficient to apply \cref{dred:HeavyVertex}.
Therefore, we transfer \emph{Neighborhood Removal}~\cite{ExactlySolvingLamm2019}.

\begin{distributedreduction}[Dist. Neighborhood Removal]\label{dred:NeighborhoodRemoval}
	Let $v\in V_i$ with $\w_i(v)\geq\w_i(\N_i(v))$.
	Then (propose to) include $v$ following Remark~\ref{rem:includeVertex}.
\end{distributedreduction}

Next, we introduce \nameref{dred:SimplicialVertex} to reduce cliques.
A vertex $v$ is simplicial if its neighborhood forms a clique.
A clique in $G_i$ contains at most one ghost, since $E_i$ does not contain edges between ghosts.
Moreover, we cannot conclude for a ghost that it is a simplicial vertex in $\G$ because PE $i$ does not know its complete neighborhood in $\G$.

\begin{distributedreduction}[Distributed Simplicial Vertex]\label{dred:SimplicialVertex}
	Let $v\in V_i$ be a simplicial vertex with maximum weight in its neighborhood in $G_i$,~\ie it holds $\w_i(v)\geq \max \set{\w_i(u): u\in \N_i(v)}$.
	Then (propose to) include $v$ following Remark~\ref{rem:includeVertex}.
\end{distributedreduction}
\begin{proof}
	To prove correctness, we show that this reduction is a special case of \nameref{dred:HeavyVertex}.
	Let $v\in V_i$ be a simplicial vertex with $\w_i(v)\geq \max \set{\w_i(u): u\in \N_i(v)}$.
	Then, $\N_i(v)$ forms a clique in $G_i$.
	Thus, an {\MWIS} of the induced neighborhood graph of $G_i[\N_i(v)]$ can consist of at most one vertex.
	It follows ${\w_i(v) \geq \max \set{\w_i(u): u\in \N_i(v)} \geq \solw{G_i[\N_i(v)]}}$.
	Thus, we can apply \nameref{dred:HeavyVertex}.
  This yields the described reduced graph, offset and \hbox{proposed solution reconstruction.}
\end{proof}

The following distributed reduction adapts \emph{Simplicial Weight Transfer}~\cite{ExactlySolvingLamm2019}.
and applies to cliques in $G_i$ that do not contain a ghost vertex and at \hbox{least one non-interface vertex.}

\begin{distributedreduction}[Dist. Simp. Weight Transfer]\label{dred:SimplicialWeightTransfer}
	Let $v\in V_i\setminus N_i(\ghosts{i})$ be a simplicial vertex, let $S(v)\subseteq N_i(v)$ be the set of all simplicial vertices.
	Further, let $w_i(v)\geq w_i(u)$ for all $u\in S(v)$.
	If $\w_i(v) < \max \set{w_i(u): u \in \N_i(v)}$, fold $v$ into $N_i(v)$.
	\begin{tabular}{ll}
		\reducedGraphRow{$G_i'=G_i-X$ with                                                    \\ $X=\set{u\in \N_i[v]: \w_i(u)\leq \w_i(v)}$ \\ and set $w_i'(u)= w_i(u)-w_i(v)$ \\ for all $x\in N_i(v)\setminus X$.} \\
		\offsetRow{$\solw{\G} = \solw{\G'}+\w(v)$.}                                           \\
		\reconstructionRow{If $\I_i'\cap\N_i(v) = \emptyset$, then $\I_i =\I_i'\cup \set{v}$, \\ else $\I_i=\I_i'$.}
	\end{tabular}
\end{distributedreduction}

In the following, we transfer data reduction rules to the \nameref{def:distRedModel} to exclude local vertices.
These rules are capable of excluding vertices because in these cases they are not part of some {\MWIS}.
In the case of an interface vertex, we need to show that we can exclude it independently of whether potential ghost neighbors in $G_i$ have already been reduced in $\G$.

\begin{distributedreduction}[Distributed Basic Single-Edge]\label{dred:BasicSingleEdge}
	Let $u,v\in V_i$ be adjacent vertices with $\w_i(\N_i(u)\setminus \N_i(v)) \leq \w_i(u)$, then exclude $v$.

	\begin{tabular}{ll}
		\textit{Reduced Graph}  & $G_i'=G_i-v$             \\
		\textit{Offset}         & $\solw{\G} = \solw{\G'}$ \\
		\textit{Reconstruction} & $\I_i = \I_i'$
	\end{tabular}
\end{distributedreduction}

\begin{figure}
	\centering
	\begin{tikzpicture}[
		minpoint/.style={inner sep=0pt, outer sep=0pt, draw=none},
		point/.style={circle, fill=black, inner sep=2pt},
		reduced/.style={circle, fill=gray!50, inner sep=2pt},
		included/.style={circle, draw=green!50!black, fill=gray!30, inner sep=2pt, line width=1.5pt},
		edge/.style={thick},
		reducededge/.style={edge, draw=gray!50},
		nodeweight/.style={font=\small, text=red, label distance=0.0mm},
		reducednodeweight/.style={font=\small, text=gray, label distance=0.0mm},
		pe2/.style={cyan!20,draw=black},
		pe1/.style={yellow!20,draw=black},
	]
	\pgfdeclarelayer{background}
	\pgfsetlayers{background,main}

	\begin{scope}
		\coordinate (origin) at (0,0) {};
		\node[point]  (a) at ($(origin) + (2.0\xdist, 0)$) {};
		\node[minpoint,label={[nodeweight]right:$7$}]   at (a) {};
		\node[point]  (b) at ($(origin) + (4.0\xdist, 0)$) {};
		\node[minpoint, label={[nodeweight]right:$2$}]   at (b) {};
		\node[point]  (v) at ($(origin) + (1.0\xdist, -2\xdist)$) {};
		\node[minpoint, label={left:$v$}, label={[nodeweight]200:$7$}]   at (v) {};
		\node[point]  (u) at ($(origin) + (3.0\xdist, -2\xdist)$) {};
		\node[minpoint, label={right:$u$}, label={[nodeweight]300:$10$}]   at (u) {};
		\node[point]  (d) at ($(origin) + (2.0\xdist, -4\xdist)$) {};
		\node[minpoint, label={[nodeweight]right:$4$}]   at (d) {};

		\draw[edge] (a) -- (v);
		\draw[edge] (a) -- (u);
		\draw[edge] (b) -- (u);
		\draw[edge] (v) -- (u);
		\draw[edge] (v) -- (d);
		\draw[edge] (u) -- (d);

		\coordinate (l1) at ($(origin) +(-1\xdist, -\xdist)$);
		\coordinate (l2) at ($(origin) +(+5.5\xdist, -\xdist)$);
		\coordinate (l2') at ($(origin) +(+6\xdist, -\xdist)$);
		\draw[edge, gray] (l1) -- node[pos=0.05, above, color=black] {$V_i^g$} node[pos=0.05, below, color=black] {$V_i$}  (l2);
		\node[align=center] at ($(d) + (0.75\xdist,-1.5\xdist)$) {$G_i$\\ $\omega_i(u) \ge \omega_i(\textcolor{my-teal}{N_i(u) \setminus N_i(v)})$};
	\end{scope}
	\begin{scope}[on background layer, opacity=0.1]
		\coordinate (C0)    at ($(v)+(-0.1\xdist,+0.9\xdist)$);
		\coordinate (C1)    at ($(v)+(-1.2\xdist,0\xdist)$);
		\coordinate (C2)    at ($(v)+(-0.2\xdist,-1.0\xdist)$);
		\coordinate (C4)    at ($(u)!0.4!(a)$);
		\coordinate (C5)    at ($(b)+(0.25\xdist,-0.9\xdist)$);
		\coordinate (C6)    at ($(b)+(0.9\xdist,0\xdist)$);
		\coordinate (C7)    at ($(b)+(0\xdist,0.9\xdist)$);
		\coordinate (C8)    at ($(a)!0.3!(u)$);
		\fill[my-teal!20, draw=black] plot[smooth cycle, tension=.7] coordinates {(C0) (C1) (C2)  (C4) (C5) (C6) (C7) (C8)};

	\end{scope}
	\begin{scope}[xshift=12\xdist]
		\coordinate (origin) at (0,0) {};
		\node[point]  (a) at ($(origin) + (2.0\xdist, 0)$) {};
		\node[minpoint,label={[nodeweight]right:$7$}]   at (a) {};
		\node[point]  (b) at ($(origin) + (4.0\xdist, 0)$) {};
		\node[minpoint, label={[nodeweight]right:$2$}]   at (b) {};
		\node[reduced]  (v) at ($(origin) + (1.0\xdist, -2\xdist)$) {};
		\node[minpoint, label={left:$v$}, label={[reducednodeweight]200:$7$}]   at (v) {};
		\node[point]  (u) at ($(origin) + (3.0\xdist, -2\xdist)$) {};
		\node[minpoint, label={right:$u$}, label={[nodeweight]300:$10$}]   at (u) {};
		\node[point]  (d) at ($(origin) + (2.0\xdist, -4\xdist)$) {};
		\node[minpoint, label={[nodeweight]right:$4$}]   at (d) {};

		\draw[reducededge] (a) -- (v);
		\draw[edge] (a) -- (u);
		\draw[edge] (b) -- (u);
		\draw[reducededge] (v) -- (u);
		\draw[reducededge] (v) -- (d);
		\draw[edge] (u) -- (d);

		\coordinate (l1') at ($(origin) +(-1.5\xdist, -\xdist)$);
		\draw[->,decorate, decoration={snake,amplitude=0.3mm}] (l2') -- (l1');

		\coordinate (l1) at ($(origin) +(-1\xdist, -\xdist)$);
		\coordinate (l2) at ($(origin) +(+5.5\xdist, -\xdist)$);
		\coordinate (l2') at ($(origin) +(+7\xdist, -\xdist)$);
		\draw[edge, gray] (l1) -- node[pos=0.05, above, color=black] {$V_i^g$} node[pos=0.05, below, color=black] {$V_i$}  (l2);
		\node[align=center] at ($(d) + (0.75\xdist,-1.5\xdist)$) {$G_i' = G_i - v$};

	\end{scope}

\end{tikzpicture}
	\caption{Example of \nameref{dred:BasicSingleEdge} from the perspective of PE $i$. The reduction excludes the interface vertex $v\in V_i$ in $G_i$ because for a neighbor $u\in N_i(v)$ holds ${\color{red}10}=\w_i(u)\geq\w_i(N_i(u)\setminus N_i(v))={\color{red}9}$. This results in the \hbox{reduced graph $G_i'=G_i-v$.}}
	\label{fig:exampleBasicSingleEdge}
\end{figure}
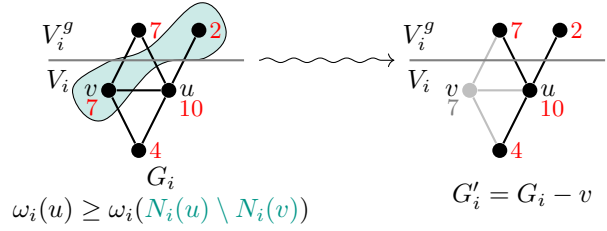

Next, we introduce \nameref{dred:ExtendedSingleEdge} to our \nameref{def:distRedModel}.
Again, we consider an edge in $G_i$ which is incident to two local vertices.
This time, we aim to reduce their \hbox{common local neighbors.}

\begin{distributedreduction}[Dist. Extended Single-Edge]\label{dred:ExtendedSingleEdge}
	Let $u,v\in V_i$ be adjacent vertices with $\w_i(\N_i(v)) - \w_i(u) \leq \w_i(v)$, then exclude $X=\N_i(v)\cap\N_i(u)\setminus \ghosts{i}$.

	\begin{tabular}{ll}
		\textit{Reduced Graph}  & $G_i'=G_i-X$           \\
		\textit{Offset}         & $\solw{G} = \solw{G'}$ \\
		\textit{Reconstruction} & $\I_i = \I_i'$
	\end{tabular}
\end{distributedreduction}

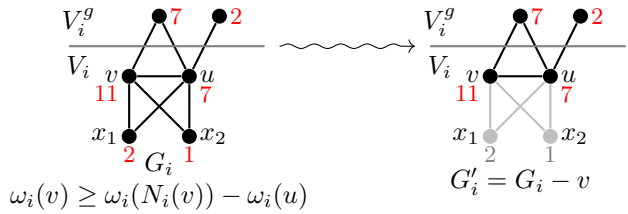
\begin{figure}
	\begin{tikzpicture}[
		minpoint/.style={inner sep=0pt, outer sep=0pt, draw=none},
		point/.style={circle, fill=black, inner sep=2pt},
		reduced/.style={circle, fill=gray!50, inner sep=2pt},
		included/.style={circle, draw=green!50!black, fill=gray!30, inner sep=2pt, line width=1.5pt},
		edge/.style={thick},
		reducededge/.style={edge, draw=gray!50},
		nodeweight/.style={font=\small, text=red, label distance=0.0mm},
		reducednodeweight/.style={font=\small, text=gray, label distance=0.0mm},
		pe2/.style={cyan!20,draw=black},
		pe1/.style={yellow!20,draw=black},
	]
	\pgfdeclarelayer{background}
	\pgfsetlayers{background,main}

	\begin{scope}
		\coordinate (origin) at (0,0) {};
		\node[point]  (a) at ($(origin) + (2.0\xdist, 0)$) {};
		\node[minpoint,label={[nodeweight]right:$7$}]   at (a) {};
		\node[point]  (b) at ($(origin) + (4.0\xdist, 0)$) {};
		\node[minpoint, label={[nodeweight]right:$2$}]   at (b) {};
		\node[point]  (v) at ($(origin) + (1.0\xdist, -2\xdist)$) {};
		\node[minpoint, label={left:$v$}, label={[nodeweight]200:$11$}]   at (v) {};
		\node[point]  (u) at ($(origin) + (3.0\xdist, -2\xdist)$) {};
		\node[minpoint, label={right:$u$}, label={[nodeweight]300:$7$}]   at (u) {};
		\node[point]  (x1) at ($(origin) + (1.0\xdist, -4\xdist)$) {};
		\node[minpoint,label={left:$x_1$}, label={[nodeweight]below:$2$}]   at (x1) {};
		\node[point]  (x2) at ($(origin) + (3.0\xdist, -4\xdist)$) {};
		\node[minpoint,label={right:$x_2$}, label={[nodeweight]below:$1$}]   at (x2) {};

		\draw[edge] (a) -- (v);
		\draw[edge] (a) -- (u);
		\draw[edge] (b) -- (u);
		\draw[edge] (v) -- (u);
		\draw[edge] (v) -- (x1);
		\draw[edge] (v) -- (x2);
		\draw[edge] (u) -- (x1);
		\draw[edge] (u) -- (x2);

		\coordinate (l1) at ($(origin) +(-1\xdist, -\xdist)$);
		\coordinate (l2) at ($(origin) +(+5.5\xdist, -\xdist)$);
		\coordinate (l2') at ($(origin) +(+6\xdist, -\xdist)$);
		\draw[edge, gray] (l1) -- node[pos=0.05, above, color=black] {$V_i^g$} node[pos=0.05, below, color=black] {$V_i$}  (l2);
		\node[align=center] at ($(x1) + (\xdist,-1.5\xdist)$) {$G_i$\\ $\omega_i(v) \ge \omega_i(N_i(v)) - \omega_i(u)$};
	\end{scope}
	\begin{scope}[xshift=12\xdist]
		\coordinate (origin) at (0,0) {};
		\node[point]  (a) at ($(origin) + (2.0\xdist, 0)$) {};
		\node[minpoint,label={[nodeweight]right:$7$}]   at (a) {};
		\node[point]  (b) at ($(origin) + (4.0\xdist, 0)$) {};
		\node[minpoint, label={[nodeweight]right:$2$}]   at (b) {};
		\node[point]  (v) at ($(origin) + (1.0\xdist, -2\xdist)$) {};
		\node[minpoint, label={left:$v$}, label={[nodeweight]200:$11$}]   at (v) {};
		\node[point]  (u) at ($(origin) + (3.0\xdist, -2\xdist)$) {};
		\node[minpoint, label={right:$u$}, label={[nodeweight]300:$7$}]   at (u) {};
		\node[reduced]  (x1) at ($(origin) + (1.0\xdist, -4\xdist)$) {};
		\node[minpoint,label={left:$x_1$}, label={[reducednodeweight]below:$2$}]   at (x1) {};
		\node[reduced]  (x2) at ($(origin) + (3.0\xdist, -4\xdist)$) {};
		\node[minpoint,label={right:$x_2$}, label={[reducednodeweight]below:$1$}]   at (x2) {};

		\draw[edge] (a) -- (v);
		\draw[edge] (a) -- (u);
		\draw[edge] (b) -- (u);
		\draw[edge] (v) -- (u);
		\draw[reducededge] (v) -- (x1);
		\draw[reducededge] (v) -- (x2);
		\draw[reducededge] (u) -- (x1);
		\draw[reducededge] (u) -- (x2);

		\coordinate (l1) at ($(origin) +(-1\xdist, -\xdist)$);
		\coordinate (l2) at ($(origin) +(+5.5\xdist, -\xdist)$);
		\coordinate (l1') at ($(origin) +(-1.5\xdist, -\xdist)$);
		\draw[edge, gray] (l1) -- node[pos=0.05, above, color=black] {$V_i^g$} node[pos=0.05, below, color=black] {$V_i$}  (l2);
		\node[align=center] at ($(x1) + (\xdist,-1.5\xdist)$) {$G_i' = G_i - v$};
		\draw[->,decorate, decoration={snake,amplitude=0.3mm}] (l2') -- (l1');

	\end{scope}

\end{tikzpicture}
	\caption{Example for \nameref{dred:ExtendedSingleEdge} from the perspective of PE $i$. It holds ${\color{red}11}=w_i(v)\geq \w_i(N(v)) - \w_i(u) = {\color{red} 10}$. Thus, the reduction excludes the common local neighbors $x_1$ and~$x_2$ of the adjacent local vertices $v$ and $u$. This results in the \hbox{reduced graph $G_i'=G_i-v$.}}
	\label{fig:exampleExtendedSingleEdge}
\end{figure}

\begin{remark}[Exhaustively Reducing Degree One Vertices]
	With the \nameref{dred:MetaRule}, we can apply \emph{Degree-One} by Gu~\etal~\cite{TowardsComputiGuJi2021} to reduce local degree one vertices with local neighbors.
	If an interface vertex $v\in V_i$ with $N_i(v)=\set{u}$ has weight $\w_i(v)\geq \w_i(u)$, \nameref{def:distRedModel} applies and $v$ is proposed to be included.
	However, so far it is not guaranteed that $v$ is ever reduced if its single neighbor $u$ has weight $\w_i(u)>\w_i(v)$.
	To that end, we support so-called \emph{move proposal} where PE $i$ transfers the ownership of $v$ to PE $j=\rank(u)$.
	Of course a move requires communication, but PE $i$ can freely schedule this communication.
	If PE $j$ receives the move, the ghost $v$ becomes a local degree one vertex with a local neighbor $u$ if $v$ was not reduced yet in the meantime at PE $j$.
	A special case occurs if $u$ is moved to PE $j$ as well.
	In this case $\set{u,v}$ form an isolated edge in $\G$.
	Both PEs reduce $u$ and $v$ the same way, so that no further communication is necessary.
	They follow \Cref{lem:includeOperation} to reduce the edge.
\end{remark}

To conclude the section, we note that all proposed distributed reduction rules adhere to the \nameref{def:distRedModel} and that there are no graph modifications outside this model.

\section{Distributed Reduction Algorithms.}
\label{sec:distReductionAlgos}
Now, we present a  distributed reduction algorithm that applies our previously developed reduction rules.
First, we discuss local work in the reduction process.
Then, we describe two approaches for the exchange of reduction progress at the border between adjacent PEs.

\subsection{Local Reduction Phase.}\label{sec:localReducePhase}
Each PE $i$ tests and applies the distributed reduction rules exhaustively to its local subgraph $G_i$.
From the perspective of a PE, this approach is similar to sequential reduction algorithms~\cite{FindingNearOpGrossma2023, TowardsComputiGuJi2021, ExactlySolvingLamm2019}.
The distributed reduction rules are processed in a \emph{fixed} order $O$.
Each rule $r$ is tested for all vertices.
If we can successfully apply $r$, i.e., reduce the graph, we restart this process with the very first rule in $O$.
Otherwise, the next rule in $O$ is considered.
This processes ensures that once no rule is left, $G_i'$ cannot be reduced any further by PE $i$.

\ifshowFull
Our reduction order follows the intuition to start with reduction rules that can reduce interface vertices as well as other locally.
If interface vertices are reduced, there is a chance that the local reduction subgraph becomes less connected to the remaining graph stored at other PEs.
Ideally, a good reduction impact of the following distributed reductions entails more reduction applications by rules which require a higher degree of locality.
Therefore, we put \emph{Degree One}~\cite{TowardsComputiGuJi2021} (\nameref{dred:MetaRule}) first and \nameref{dred:NeighborhoodRemoval} second.
Likewise, both rules can also be efficiently tested in $\mathcal{O}(1)$ and $\mathcal{O}(\Delta)$, where $\Delta$ is the maximum degree, given a vertex $v\in V_i$.
The first two data reduction rules are followed by \nameref{dred:SimplicialWeightTransfer} and \nameref{dred:SimplicialVertex}.
Afterward, we test for the first two cases of \emph{V-Shape}~\cite{TowardsComputiGuJi2021} (\nameref{dred:MetaRule}) to target vertices of degree two, and the excluding data reduction rules: \nameref{dred:BasicSingleEdge}, and \nameref{dred:ExtendedSingleEdge}.
Then, we try \emph{Neighborhood Folding}~\cite{ExactlySolvingLamm2019} (\nameref{dred:MetaRule}), followed by two rules of higher computational cost,~\nameref{dred:HeavyVertex} and \emph{Generalized Neighborhood Folding}~\cite{ExactlySolvingLamm2019} (\nameref{dred:MetaRule}).
The last two rules need to find an {\MWIS} for at least one subproblem.
The size of the subproblem size can be controlled with an extra parameter to circumvent too long running times for single subproblems.
For our experiments we choose $10$ as the maximum number of vertices in the subproblem.
The {\MWIS} is determined with the state-of-the-art branch-and-reduce solver {\kamis}~\cite{ExactlySolvingLamm2019}.
\fi

Furthermore, we apply dependency checking.
This is a common technique to prune reduction tests by skipping vertices whose local context has not changed.
By local context, we refer to the subgraph which is \textit{read} in order to evaluate for a vertex whether a reduction rule applies to it.

\subsection{Communicating Reduction Progress.}
After applying all reduction rules locally, the distributed graph is not necessarily exhaustively reduced from a global perspective.
It is possible that a local subgraph can be further reduced if the reduction progress at the border is synchronized with adjacent PEs as receiving updates for reduced ghost vertices might re-open the search space in the local reduction phase.
Moreover, communication is necessary to break ties between proposals trying to include vertices connected by a cut-edge.
To that end, we introduce two different message types that enable the exchange of reduction progress for interface vertices $v \in V_i$:

\noindent (1) \textit{Weight Decrease.} If the weight of $v\in V_i$ is decreased, it can be updated at the adjacent PEs by sending $v$ together with its new weight $w_i'(v)$
to the PEs $\recv_i(v)$.

\noindent (2) \textit{Vertex Status.} A vertex status message informs the adjacent PEs that $v$ is either excluded, moved, or proposed to be included.

These messages are sent to the adjacent PEs $\recv_i(v)=\set{\rank(v):v\in \N_i(u)\cap \ghosts{i}}$,~\ie the PEs that own ghost neighbors of $v$.

\subsection{Synchronous Reduction Algorithm.}\label{sec:synchronousApproach}

\Cref{algo:KaDisReduS} shows a high-level description of our synchronous distributed reduction algorithm {\KaDisReduS}.
The algorithm repeats its four steps until no further global reduction progress can be achieved.
Once all PEs have completed the local reduce phase, the border is synchronized using irregular all-to-all exchanges.
\paragraph{Local Reduction.} PE $i$ reduces the local subgraph exhaustively as described in \Cref{sec:localReducePhase}.
Since messages are not sent until the local reduce phase has finished, we track weight-modifications, reduced and moved interface vertices.
This allows us to filter for redundant updates.
Therefore, we insert interface vertices of $V_i$ into a set $W_i$ if their weight is decreased.
Note that, if the weight of an interface vertex changes multiple times in the local reduce phase, only the final weight is relevant to other PEs.
Furthermore, we build a sequence of local vertices $M_i\subseteq V_i\times \set{\texttt{included}, \texttt{excluded}, \texttt{moved}}$ which were included, reduced, or moved.
\paragraph{Exchange Weight Updates.} Weight decrease updates for $v\in W_i$ are sent to all adjacent PEs if $v$ remains in $G_i'$ or was moved to its adjacent PE.
\paragraph{Filtering.} We drop vertex status entries from the sequence $M_i$ of the form $(v,\,\texttt{moved})$ if the single neighbor of $v$ was reduced later on.
In such a case, we can circumvent the move and include an isolated vertex $v$ into the solution.
Note that by the end of this step, the pruned sequence contains each vertex at most once \hbox{in some entry.}
\paragraph{Exchange Status Updates.}
We synchronize interface vertex statuses by sending $(v,s)\in M_i$ to each PE of $v$ that was adjacent to $v$ when the status $s$ was set.
Note that this can entail new modified interface vertices.
When a proposal to include a ghost $v$ is received and not reduced yet, the proposal succeeds.
Thus, local neighbors can be excluded and might entail new exclude vertex messages.
Moreover, our include operation \Cref{lem:includeOperation} tie-breaks conflicting proposals in the last step.

\begin{algorithm}
	\caption{{\KaDisReduS}}
	\label{algo:KaDisReduS}
	\begin{algorithmic}
		\Procedure{\KaDisReduS}{$G_i$}\Comment{subgraph $G_i$ of $G$ }
		\State{$(G_i', M_i, o_i)\gets (G_i, <>, 0)$}\Comment{initialisation}
		\While{global reduction progress}
		\State{$(G_i',\,o_i,\,W_i,\,M_i) \gets\,$LocalReduce($G_i',\,M_i,\,o_i$)}
		\State{$G_i'\gets\,$ ExchWeightUpdates($G_i',\,W_i$)}
		\State{$M_i\gets\,$ FilterMoves($G_i',\,M_i$)}
		\State{$(G_i',\,o_i,\,M_i)\gets\,$ ExchStatusUpdates($G_i',\,M_i$)}
		\EndWhile
		\State \Return{$G_i',\,o_i$}\Comment{local reduced graph and offset}
		\EndProcedure{}
	\end{algorithmic}
\end{algorithm}

\subsection{Asynchronous Reduction Algorithm.}\label{sec:asynchronousApproach}
A drawback of {\KaDisReduS} is the synchronization overhead when PEs wait for others to finish their reduce phase.
To that end, we present {\KaDisReduA} -- a distributed reduction algorithm using \emph{asynchronous} communication.

	{\KaDisReduA} mitigates idle times by sending and receiving border updates directly during the local reduce phase by exchanging updates in an asynchronous fashion.
However, exchanging messages involves a start-up overhead which is usually orders of magnitude larger than the time required to exchange a machine word.
Therefore, we employ adaptive message buffering, using the asynchronous message queue BriefKAsten~\cite{EngineeringADSander2023}.
That is, instead of exchanging single update messages, we write them into buffers that accumulate and exchange them once a certain threshold has been reached or the buffer is explicitly flushed.
The buffer size threshold is an additional parameter that allows us to choose a good trade-off between the start-up overhead, the sent message-sizes, and idle times.

\subsection{Reconstruction.}\label{sec:reconstruction}
Given the reduced graph $\G'$ and an MWIS for $\G'$ we want to reconstruct a solution for $\G$.
Compared to a sequential reduce algorithm, our distributed reduce algorithms also require communication for this final step.
The reason for that are folded vertices which are moved to another PE.
This introduces reverse dependencies between PEs in the reconstruction phase which require a message exchange.

\section{Distributed Weight Independent Set Solvers.}
\label{sec:distributedSolvers}

Our distributed reduction algorithms allow us to build exact and heuristic distributed algorithms for the {\MWIS} problem.
As a first step towards a broad range of distributed {\MWIS} solvers, we present two heuristic algorithms --- both in synchronous \hbox{and asynchronous variants.}

\paragraph{Reduce-And-Greedy}.
\label{sec:reduce-and-greedy}
The first algorithm is a reduce-and-greedy algorithm.
It applies the distributed reduce algorithm followed by a distributed variant of Luby's algorithm to find a greedy solution for the reduced graph.
A vertex $v\in V_i$ is included at PE $i$ if it maximizes its weight $\w_i(v)$ among its neighbors $N_i(v)$.
Included interface vertices are communicated with the adjacent PEs so that the border is synchronized.
If an interface vertex $v$ has the same weight as a ghost neighbor $v$, the rank of the PEs is used to break ties.

\paragraph{Reduce-And-Peel}.
\label{sec:reduce-and-peel}
Furthermore, we propose a distributed reduce-and-peel solver.
This scheme has been successfully applied in the sequential solver {\htwis}~\cite{TowardsComputiGuJi2021}.
The main idea is to iteratively apply two steps.
First, apply reductions exhaustively in a reduce phase.
Then, a vertex $v \in V$ that is unlikely to be in any \emph{MWIS} is \emph{peeled}.
Excluding $v$ might enable new reduction applications in the first phase.
Often, this achieves a better solution quality than plain greedy decisions.

We propose a synchronous and asynchronous reduce-and-peel solver.
They use the respective distributed reduce algorithm to exhaustively reduce the distributed graph.
Afterwards, each PE peels a local vertex $v\in V_i$ before proceeding with the reduce phase.
We choose $u\in V_i$ so that it maximizes $N_i(v)-\w_i(v)$ similar to {\htwis}.
Note that both communication approaches reach a communication barrier when they finish the reduce phase.
Otherwise, a PE might peel all the vertices before it has received reduction messages of other PEs.

\section{Experiments.}
\label{sec:experiments}
In our evaluation, we use up to 16 compute nodes (1\,024 cores) of HoreKa\footnote{We also conducted experiments on SuperMUC-NG with similar results.
}
where each node is equipped with two Intel Xeon Platinum 8368
processors with 38 cores each and 256~GB of main memory. Compute nodes are connected by an InfiniBand 4X HDR 200~GBit/s interconnect.

We compare the following algorithms (see \cite{SourceCode, Reproducibility} for accompanying source code).

\noindent \textbf{\KaDisReduS/\KaDisReduA.} Our distributed synchronous and asynchronous reduction algorithms as presented in \Cref{sec:synchronousApproach} and \Cref{sec:asynchronousApproach}.

\noindent \textbf{\sG/\aG.} Our distributed (a)synchronous variant of Luby's (greedy) algorithm (\Cref{sec:distributedSolvers}).

\noindent \textbf{\sRG/\aRG.} Our distributed (a)synchronous reduce-and-greedy solver applying \KaDisReduSA{} (\Cref{sec:reduce-and-greedy}).

\noindent \textbf{\sRnP/\aRnP.} Our distributed (a)synchronous reduce-and-peel solvers using \KaDisReduSA{} (\Cref{sec:reduce-and-peel}).

\noindent \textbf{\htwis.} Sequential state-of-the-art reduce-and-peel algorithm~\cite{TowardsComputiGuJi2021}.

All algorithms are implemented in C++ and compiled using GNU GCC13.
For interprocess communication, we use IntelMPI 2021.11 and the MPI-wrapper~{\kamping}~\cite{DBLP:conf/sc/UhlSHHKSS024}.
For (buffered) asynchronous communication, we use the {\msgqueue} BriefKAsten~\cite{EngineeringADSander2023, uhl2025briefkasten}.
Moreover, we use the branch-and-reduce solver {\kamis}~\cite{ExactlySolvingLamm2019} to solve MWIS subproblems in our generalized data reductions.
Note that we do not compare our solvers against {\kamis}, since, unlike our algorithms, it is an exact solver.

We evaluated our algorithm in strong- and weak-scaling experiments.
\paragraph{Strong Scaling.} We consider 47 real-world graphs with uniform random weights in $[1,200]$.
The graphs have one to 118 million vertices and average degrees between two and 160 
\ifshowFull
(see \Cref{appendix:tab:instances} for an overview).
\fi
All experiments are run four times with a time limit of two hours.
For our distributed algorithms, the input graph is edge-balanced.

\paragraph{Weak Scaling.} We consider Erdős–Rényi ({\myGNM}), 2D random-geometric ({\myRGG}), and random hyperbolic graphs ({\myRHG}) with a power-law exponent of $\gamma=2.8$ as input generated with the distributed graph generator KaGen~\cite{DBLP:journals/jpdc/FunkeLMPSSSL19}.
Here, we choose a fixed number $N=2^{20}$ vertices and $M=2^{22}$ edges per core so that the graph grows linearly with the number of cores.

\subsection{Evaluating our Reduction Algorithms.}
We first analyze the reduction impact of our distributed reduction algorithms.
\Cref{fig:relative-reduction-impact} shows the reduced graph size (number of vertices) relative to the input graph and the difference for $p$ cores relative to one core when using \KaDisReduA{} for our real-world instances.
\ifshowFull
Results for \KaDisReduS{} are very similiar and are shown in \Cref{appendix:fig:kernelsize-runningtime-reduction} 
alongside the reduction impact on the number of edges.
\fi
In general, the reduction impact worsens with an increasing number of cores.
Most likely, limitations of reducing border vertices seem to be the reason for this observation.
Still, this increase in the number of non-reduced vertices remains below $10\%$ relative to the input graph on the median (below $30\%$ for $Q3$ of the instances).

\begin{figure}
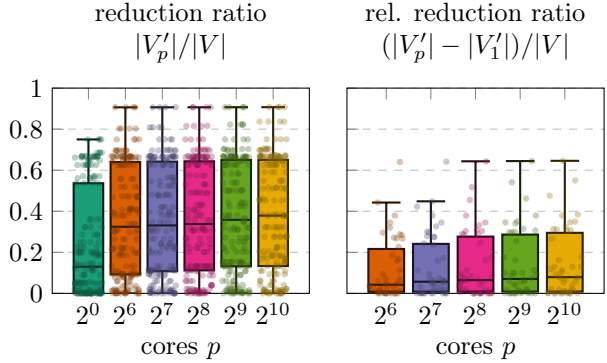

	\centering
	\begin{tikzpicture}
		\begin{groupplot}[group style={
						group size=2 by 1,
						horizontal sep=5mm,
						yticklabels at=edge left,
					},
				width=50mm,
				legend style={draw=none},
				ytick pos=left,
				ymajorgrids,
				ymax={1.0},
				ymin={0},
				ymode=normal,
				xlabel={cores $p$},
				boxplot/draw direction=y,
				grid style=dashed,
				title style={align=center},
			]
			\nextgroupplot[title={reduction ratio \\ $|V_p'|/|V|$},
				xtick={1, 2, 3, 4, 5, 6},
				xticklabels={$2^0$,$2^6$,$2^7$,$2^8$,$2^9$,$2^{10}$},
			]
			\input{plots/boxplot-rel-kernel-vertices-aRG/plain_plot.tex}
			\nextgroupplot[title={rel. reduction ratio\\$(|V_p'|-|V_1'|)/|V|$},
				ymajorgrids,
				xtick={1, 2, 3, 4, 5},
				xticklabels={$2^6$,$2^7$,$2^8$,$2^9$,$2^{10}$},
				boxplot/draw direction=y,
			]
			\input{plots/boxplot-rel-kernel-vertices-change-aRG/plain_plot.tex}
		\end{groupplot}
	\end{tikzpicture}
	\caption{(Relative) reduction impact on the number of vertices when running \KaDisReduA{} in our strong-scaling experiments on up to 1\,024 cores.}
	\label{fig:relative-reduction-impact}
\end{figure}

\begin{figure}
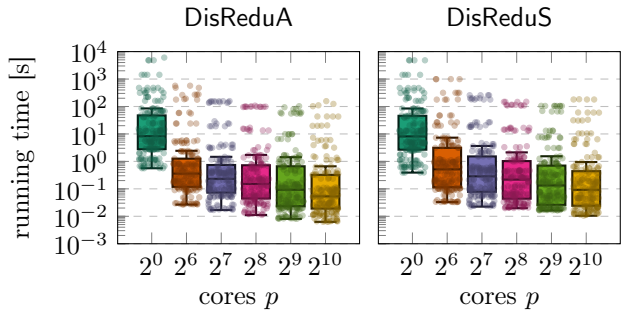

	\centering
	\begin{tikzpicture}
		\begin{groupplot}[group style={
						group size=2 by 1,
						horizontal sep=2.5mm,
						vertical sep=15mm,
						ylabels at=edge left,
						yticklabels at=edge left,
					},
				width=0.65\columnwidth,
				legend style={draw=none},
				xlabel={cores $p$},
				xtick={1,2,3,4,5,6},
				xticklabels={$2^0$,$2^6$,$2^7$,$2^8$,$2^9$,$2^{10}$},
				ytick pos=left,
				ylabel={running time [s]},
				ymajorgrids,
				ymode=log,
				ymax={10000.0},
				ymin={0.001},
				ytick={0.001, 0.01, 0.1, 1, 10, 100, 1000, 10000},
				boxplot/draw direction=y,
				grid style=dashed,
			]
			\nextgroupplot[title=\KaDisReduA{}]
			\input{plots/boxplot-reduction-time-aRG/plain_plot.tex}
			\nextgroupplot[title=\KaDisReduS{}]
			\input{plots/boxplot-reduction-time-RG/plain_plot.tex}

		\end{groupplot}
	\end{tikzpicture}
	\caption{Running time of {\KaDisReduA} (left) and {\KaDisReduS} (right) on up to $1024$ cores in our strong-scaling experiments.}
	\label{fig:reduction-time-strong-scaling}
\end{figure}
\Cref{fig:reduction-time-strong-scaling} shows the running time of \KaDisReduA{} and \KaDisReduS{}.
Both variants scale well with increasing number of cores. On geometric mean, the asynchronous {\KaDisReduA} is $1.24\times$ (at $64$ cores) and $1.5\times$ (at 1\,024 cores) faster than {\KaDisReduS}.

The reduction impact can be enhanced when running a (distributed) graph partitioning algorithm prior to the reduction phase.
For example, when partitioning the input graphs with {\dKaMinPar}~\cite{DBLP:conf/europar/SandersS23} on 1\,024 cores using an edge imbalance of $3\,\%$, we obtain reduced graphs with $|V'|/|V|=0.25$ instead of $0.38$ on median 
\ifshowFull
(see~\Cref{appendix:tab:part} for more details).
\fi
Note that we were not able to compute a partitioning for three graphs because of memory limitations.
However, on geometric mean, the running time increases by a factor of \numprint{9.7}.%

\begin{table*}
	\centering
	\begin{tabular}{rrrrrrrrrrrrrrrr}
		\multicolumn{1}{r}{} &  & \multicolumn{3}{c}{$\w(\I)/\w(\I_{best})$} &  & \multicolumn{3}{c}{$t [s]$} & & \multicolumn{3}{c}{$t_{\htwis}/t$} \\
		\multicolumn{1}{r}{} &  & \multicolumn{1}{c}{1}                      & \multicolumn{1}{c}{64} & \multicolumn{1}{c}{\numprint{1024}}           &  & \multicolumn{1}{c}{1}   & \multicolumn{1}{c}{64} & \multicolumn{1}{c}{\numprint{1024}} &  & \multicolumn{1}{c}{1} & \multicolumn{1}{c}{64} & \multicolumn{1}{c}{\numprint{1024}} \\
		\cmidrule{1-1}  \cmidrule{3-5}  \cmidrule{7-9}  \cmidrule{11-13}
		\sG    &  & \numprint{0.9151}          & \numprint{0.9151}          & \numprint{0.9151}           &  & \textbf{\numprint{2.02}} & \textbf{\numprint{0.18}} & \numprint{0.05}         &  & \textbf{\numprint{2.2}} & \textbf{\numprint{24.7}} & \numprint{96.4}          \\
\sRG   &  & \numprint{0.9481}          & \numprint{0.9391}          & \numprint{0.9371}           &  & \numprint{6.73}          & \numprint{0.43}          & \numprint{0.12}         &  & \numprint{0.7}          & \numprint{10.5}          & \numprint{37.6}          \\
\sRnP  &  & \numprint{0.9976}          & \numprint{0.9839}          & \numprint{0.9767}           &  & \numprint{14.61}         & \numprint{0.99}          & \numprint{0.49}         &  & \numprint{0.3}          & \numprint{4.6}           & \numprint{9.2}           \\
\aG    &  & \numprint{0.9151}          & \numprint{0.9151}          & \numprint{0.9151}           &  & \numprint{2.47}          & \numprint{0.19}          & \textbf{\numprint{0.03}}&  & \numprint{1.8}          & \numprint{23.5}          & \textbf{\numprint{128.9}}\\
\aRG   &  & \numprint{0.9481}          & \numprint{0.9391}          & \numprint{0.9371}           &  & \numprint{7.23}          & \numprint{0.39}          & \numprint{0.09}         &  & \numprint{0.6}          & \numprint{11.7}          & \numprint{50.6}          \\
\aRnP  &  & \numprint{0.9976}          & \textbf{\numprint{0.9882}} & \textbf{\numprint{0.9775}}  &  & \numprint{15.91}         & \numprint{0.87}          & \numprint{0.14}         &  & \numprint{0.3}          & \numprint{5.2}           & \numprint{33.2}          \\
\cmidrule{1-1}  \cmidrule{3-5}  \cmidrule{7-9}  \cmidrule{11-13}                                                                                                                                                                                                        
\htwis &  & \textbf{\numprint{0.9981}} & -                          & -                           &  & \numprint{4.51}          & -                        & -                       &  & \numprint{1.0}          & -                        & -                        \\

	\end{tabular}
  \caption{Summarized results for 1, 64, and \numprint{1024} cores: Solution quality compared to best found solution by any algorithm, speedup over {\htwis}, running time.}\label{tab:soa}
\end{table*}

\subsection{Comparing Heuristic Solvers.}
We now compare our distributed (heuristic) solver with the sequential solver {\htwis}.
Common results were found for all configurations (including \htwis) for 36/47 graphs.
{\htwis} only supports 32 bit integers which leaves 41 instances for a direct comparison.
In the experiments, {\htwis} generated wrong solution weights for two graphs and for another two graphs the solution quality was far off from ours so that we excluded them to prevent that wrong solutions are taken into account.
For a 5-th instance, our reduce-and-peel solvers were unable to output a solution in time for $p=1$.
\Cref{tab:soa} summarizes the results.
In terms of solution quality, the reduce-and-peel solvers perform best while {\htwis} is a factor \numprint{3.2} and \numprint{3.5} faster than {\sRnP} and {\aRnP}, respectively, on one core.
For \numprint{1024}, {\aRnP} still maintains a solution quality of \numprint{0.978} which is only about 2\% worse than the average solution quality of {\htwis}.
At this scale, {\aRnP} is noticeably better than the solution quality achieved by the {\aG} (\numprint{0.915}) and {\aRG} (\numprint{0.937}).
In terms of running time, the greedy approaches are the fastest, followed by the reduce-and-greedy and the reduce-and-peel algorithms.
For sufficient large $p$, we note that the asynchronous approaches are faster than their synchronous counterparts.
For example, on \numprint{1024} cores, {\aRnP} is a factor \numprint{3.5} faster than {\sRnP}.
As a result, we achieve speedups of \numprint{33}, \numprint{51}, and \numprint{129} over {\htwis} with {\aRnP}, {\aG}, and {\aRG}, respectively.

\subsection{Weak-Scaling Experiments.}
\Cref{fig:throughput} shows the throughput for all the weak-scaling instances
and \Cref{tab:weak} the corresponding solution quality (on \numprint{1024} cores).
{\sRG} was not able to process {\myGNM} for $p\geq 64$ due to its higher memory consumption.
Apart from {\sRG} for {\myGNM}, we observe a good scaling behavior of the throughput.
We observe only small differences between the synchronous and asynchronous approaches in terms of solution quality and running time.
Interestingly, the reduce-and-greedy and reduce-and-peel approaches have the same solution quality for {\myRHG}.
This is due to the fact that {\myRHG} can be reduced to less than $0.01\%$ vertices relative to the input graph.
{\myRGG} can be reduced to $34\,\%$ vertices and the reduced graph of {\myGNM} still has $\numprint{98.1}\,\%$ of the initial size.
As a result, the solution quality of {\aRG} and {\aG} for {\myGNM} is the same.
Nonetheless, we note a strong improvement in solution quality with {\aRnP} which indicates that reductions become applicable once sufficient vertices were peeled.

For random hyperbolic graphs, we were able to run an even larger configuration with $M=2^{24}$ edges per core ({\myRHGBig}) without exceeding the available main memory.
For \numprint{1024} cores, this is an input graph with more than one billion vertices and 17 billion of edges.
\ifshowFull
(see \Cref{appendix:tab:weak-scaling-results} for details).
\fi

\begin{figure}
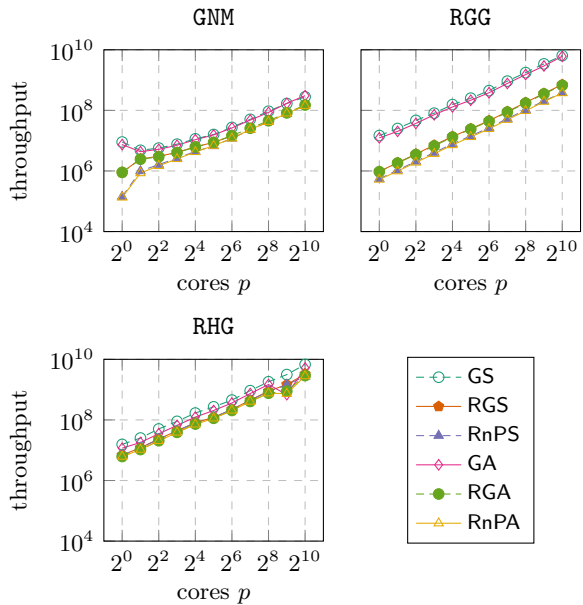

	\centering
	\begin{tikzpicture}
		\begin{groupplot}[group style={
						group size=2 by 2,
						horizontal sep=5mm,
						vertical sep=17mm,
						ylabels at=edge left,
						yticklabels at=edge left,
					},
				width=45mm,
				height=40mm,
				cycle list name=lineplotlist,
				legend style={draw=none},
				xtick={1,4, 16, 64, 256, 1024},
				xticklabels={$2^0$, $2^2$, $2^4$, $2^6$, $2^8$, $2^{10}$},
				xlabel style={font=\small},
				ylabel style={font=\small},
				xticklabel style={font=\small},
				yticklabel style={font=\small},
				ymin=10000,
				ymax=10000000000,
				ymode={log},
				xmode={log},
				log basis x=2,
				ytick pos=left,
				ylabel={throughput},
				xlabel={cores $p$},
				xmajorgrids=true,        %
				ymajorgrids=true,        %
				grid style=dashed,
			]
			\nextgroupplot[title=\myGNM,
				legend to name=group legend,
				legend columns=1,
				legend cell align=left,
				legend style={/tikz/every even column/.append style={column sep=1em}, draw=black, font=\footnotesize}]
			\input{plots/lineplot-throughput-gnm-undirected-N20-M22/plot.tex}

			\nextgroupplot[title=\myRGG]
			\input{plots/lineplot-throughput-rgg2d-N20-M22/plot.tex}

			\nextgroupplot[title=\myRHG]
			\input{plots/lineplot-throughput-rhg-N20-M22-g2.8/plot.tex}

			\nextgroupplot[hide axis]
		\end{groupplot}
    \node[anchor=center, xshift=0cm] at (group c2r2.center) {\pgfplotslegendfromname{group legend}};
	\end{tikzpicture}
  \caption{Throughput (edges per second) of our different distributed algorithms on up to \numprint{1024} cores.}\label{fig:throughput}
\end{figure}

\begin{table}
	\begin{tabular}{rrrrrrrr}
		\multicolumn{1}{r}{} &  & \multicolumn{6}{c}{$\w(\I)/\w(\I_{best})$} \\                          		\multicolumn{1}{r}{} &  & \multicolumn{1}{c}{\sG}                    & \multicolumn{1}{c}{\sRG} & \multicolumn{1}{c}{\sRnP}   & \multicolumn{1}{c}{\aG} & \multicolumn{1}{c}{\aRG} & \multicolumn{1}{c}{\aRnP} 	\\	\cmidrule{1-1}  \cmidrule{3-8}

		\myGNM & &\numprint{0.89} & - & \numprint{0.99} & \numprint{0.89} & \numprint{0.89} & \textbf{\numprint{1.00}} \\
\myRGG & &\numprint{0.91} & \numprint{0.96} & \numprint{1.00} & \numprint{0.91} & \numprint{0.96} & \textbf{\numprint{1.00}} \\ 
\myRHG & &\numprint{0.96} & \numprint{1.00} & \numprint{1.00} & \numprint{0.96} & \numprint{1.00} & \textbf{\numprint{1.00}} \\ 
\myRHGBig & &\numprint{0.85} & \numprint{1.00} & \numprint{1.00} & \numprint{0.85} & \numprint{1.00} & \textbf{\numprint{1.00}} \\ 

	\end{tabular}
	\caption{Weak scaling results: Solution quality for \numprint{1024} cores.}%
	\label{tab:weak}
\end{table}

\section{Conclusion.}
\label{sec:conclusion}
In this work, we present the first reduction rules for the maximum weight independent set problem for distributed memory.
They are based on sequential data reductions rules which can be tested and applied on a subgraph.
We use these novel rules to develop synchronous and asynchronous reduction algorithms which given a distributed input graph return an equivalent instance of smaller size.
Furthermore, we also present two heuristic approaches (reduce-and-greedy and reduce-and-peel) solving the maximum weight independent set problem in distributed memory.
We evaluate all algorithms on up to $1024$ cores, achieving substantial speedups over the state-of-the-art sequential algorithm while maintaining a good reduction ratio even on large processor configurations on graphs with (up to) billions of vertices.

Regarding future work, we plan to further improve our reduction algorithm by incorporating additional data reduction rules.
Moreover, it would be also very interesting to work on more (in)exact distributed-memory solvers for the \MWIS{} problem in general.

\section*{Acknowledgments.}

\begin{wrapfigure}{R}{.33\columnwidth}
      \vspace{-1.25\baselineskip}
      \includegraphics[width=.33\columnwidth]{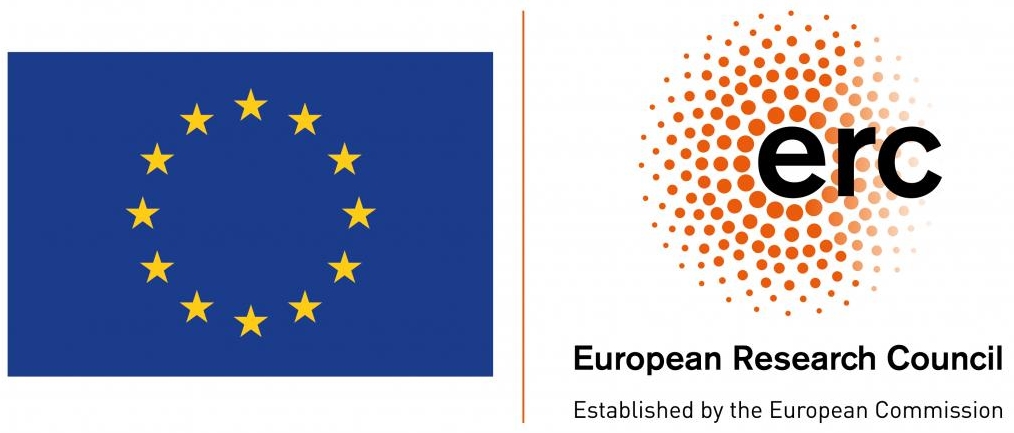}
\end{wrapfigure}	

This project has received funding from the European Research Council (ERC) under the European Union’s Horizon 2020 research and innovation programme (grant agreement No. 882500).
Moreover, the authors acknowledge support by DFG grant SCHU 2567/3-1.
This work was performed on the HoreKa supercomputer funded by the Ministry of Science, Research and the Arts Baden-Württemberg and by the Federal Ministry of Education and Research.
The authors gratefully acknowledge the Gauss Centre for Supercomputing e.V. (\protect\url{www.gauss-centre.eu}) for funding this project by providing computing time on the GCS Supercomputer SuperMUC-NG at Leibniz Supercomputing Centre (\protect\url{www.lrz.de}).

\clearpage
\bibliographystyle{siamplain}
\bibliography{lib}

\begin{thebibliography}{10}

\bibitem{openstreetmap}
{\em Openstreetmap}, \url{https://www.openstreetmap.org}.

\bibitem{8MaximalInde2000}
{\em 8. Maximal Independent Sets (MIS)}, Society for Industrial and Applied
  Mathematics, 1 2000, pp.~91--102,
  \url{https://doi.org/10.1137/1.9780898719772.ch8},
  \url{https://doi.org/10.1137/1.9780898719772.ch8}.

\bibitem{RecentAdvancesAbuKh2022}
{\sc F.~N. Abu-Khzam, S.~Lamm, M.~Mnich, A.~Noe, C.~Schulz, and D.~Strash},
  {\em Recent Advances in~Practical Data Reduction}, Springer Nature
  Switzerland, 1 2022, pp.~97--133,
  \url{https://doi.org/10.1007/978-3-031-21534-6\_6},
  \url{https://doi.org/10.1007/978-3-031-21534-6\_6}.

\bibitem{BranchAndReduAkiba2016}
{\sc T.~Akiba and Y.~Iwata}, {\em Branch-and-reduce exponential/fpt algorithms
  in practice: A case study of vertex cover}, Theoretical Computer Science, 609
  (2016), pp.~211--225, \url{https://doi.org/10.1016/j.tcs.2015.09.023},
  \url{https://linkinghub.elsevier.com/retrieve/pii/S030439751500852X}.

\bibitem{DBLP:conf/dimacs/2012}
{\sc D.~A. Bader, H.~Meyerhenke, P.~Sanders, and D.~Wagner}, eds., {\em Graph
  Partitioning and Graph Clustering, 10th {DIMACS} Implementation Challenge
  Workshop, Georgia Institute of Technology, Atlanta, GA, USA, February 13-14,
  2012. Proceedings}, vol.~588 of Contemporary Mathematics, American
  Mathematical Society, 2013, \url{https://doi.org/10.1090/CONM/588},
  \url{https://doi.org/10.1090/conm/588}.

\bibitem{TemporalMapLaBarth2016}
{\sc L.~Barth, B.~Niedermann, M.~N\"{o}llenburg, and D.~Strash}, {\em Temporal
  map labeling: a new unified framework with experiments}, in Proceedings of
  the 24th {ACM} {SIGSPATIAL} {International} {Conference} on {Advances} in
  {Geographic} {Information} { Systems}, Burlingame California, October 2016,
  ACM, pp.~1--10, \url{https://doi.org/10.1145/2996913.2996957},
  \url{https://dl.acm.org/doi/10.1145/2996913.2996957}.

\bibitem{BRSLLP}
{\sc P.~Boldi, M.~Rosa, M.~Santini, and S.~Vigna}, {\em Layered label
  propagation: A multiresolution coordinate-free ordering for compressing
  social networks}, in Proceedings of the 20th international conference on
  World Wide Web, S.~Srinivasan, K.~Ramamritham, A.~Kumar, M.~P. Ravindra,
  E.~Bertino, and R.~Kumar, eds., ACM Press, 2011, pp.~587--596.

\bibitem{BoVWFI}
{\sc P.~Boldi and S.~Vigna}, {\em The {W}eb{G}raph framework {I}: {C}ompression
  techniques}, in Proc. of the Thirteenth International World Wide Web
  Conference (WWW 2004), Manhattan, USA, 2004, ACM Press, pp.~595--601.

\bibitem{SourceCode}
{\sc J.~Borowitz, E.~Gro\ss{}mann, and M.~Schimek}, {\em Source code of
  \textsf{KaDisRedu}.}
\newblock The source code of \textsf{KaDisRedu} is available at Zenodo
  \url{https://doi.org/10.5281/zenodo.17296045} and maintained at GitHub
  \url{https://github.com/jabo17/kadisredu}.

\bibitem{Reproducibility}
{\sc J.~Borowitz, E.~Gro\ss{}mann, and M.~Schimek}, {\em Source code to
  reproduce artifacts.}
\newblock The source code to reproduce the artifacts is available at Zenodo
  \url{https://doi.org/10.5281/zenodo.17310080} and maintained at GitHub
  \url{https://github.com/jabo17/kadisredu-reproducibility}.

\bibitem{FindingMaximumButenk2002}
{\sc S.~Butenko, P.~Pardalos, I.~Sergienko, V.~Shylo, and P.~Stetsyuk}, {\em
  Finding maximum independent sets in graphs arising from coding theory}, in
  SAC02: 2002 ACM Symposium on Applied Computing, ACM, 3 2002, pp.~542--546,
  \url{https://doi.org/10.1145/508791.508897},
  \url{https://doi.org/10.1145/508791.508897}.

\bibitem{ComputingANeaChang2017}
{\sc L.~Chang, W.~Li, and W.~Zhang}, {\em Computing a near-maximum independent
  set in linear time by reducing-peeling}, in SIGMOD/PODS'17: International
  Conference on Management of Data, ACM, 5 2017, pp.~1181--1196,
  \url{https://doi.org/10.1145/3035918.3035939},
  \url{https://doi.org/10.1145/3035918.3035939}.

\bibitem{AcceleratingLoDahlum2016}
{\sc J.~Dahlum, S.~Lamm, P.~Sanders, C.~Schulz, D.~Strash, and R.~F. Werneck},
  {\em Accelerating Local Search for the Maximum Independent Set Problem},
  Springer International Publishing, 6 2016, pp.~118--133,
  \url{https://doi.org/10.1007/978-3-319-38851-9\_9},
  \url{https://doi.org/10.1007/978-3-319-38851-9\_9}.

\bibitem{NewInstancesFDong2021}
{\sc Y.~Dong, A.~V. Goldberg, A.~Noe, N.~Parotsidis, M.~G.~C. Resende, and
  Q.~Spaen}, {\em New instances for maximum weight independent set from a
  vehicle routing application}, Operations Research Forum, 2 (2021), p.~48,
  \url{https://doi.org/10.1007/s43069-021-00084-x},
  \url{https://link.springer.com/10.1007/s43069-021-00084-x}.

\bibitem{AMetaheuristicDong2022}
{\sc Y.~Dong, A.~V. Goldberg, A.~Noe, N.~Parotsidis, M.~G.~C. Resende, and
  Q.~Spaen}, {\em A metaheuristic algorithm for large maximum weight
  independent set problems}, 2022,
  \url{https://doi.org/10.48550/ARXIV.2203.15805},
  \url{https://arxiv.org/abs/2203.15805}.
\newblock Version Number: 1.

\bibitem{ThereAndBackFigiel2022}
{\sc A.~Figiel, V.~Froese, A.~e. Nichterlein, and R.~Niedermeier}, {\em There
  and back again: On applying data reduction rules by undoing others}, in 30th
  Annual European Symposium on Algorithms, {ESA} 2022, September 5-9, 2022,
  Berlin/Potsdam, Germany, S.~Chechik, G.~Navarro, E.~Rotenberg, and G.~Herman,
  eds., vol.~244 of LIPIcs, Schloss Dagstuhl - Leibniz-Zentrum f{\"{u}}r
  Informatik, 2022, pp.~53:1--53:15,
  \url{https://doi.org/10.4230/LIPICS.ESA.2022.53},
  \url{https://doi.org/10.4230/LIPIcs.ESA.2022.53}.

\bibitem{DBLP:journals/jpdc/FunkeLMPSSSL19}
{\sc D.~Funke, S.~Lamm, U.~Meyer, M.~Penschuck, P.~Sanders, C.~Schulz,
  D.~Strash, and M.~von Looz}, {\em Communication-free massively distributed
  graph generation}, J. Parallel Distributed Comput., 131 (2019), pp.~200--217,
  \url{https://doi.org/10.1016/J.JPDC.2019.03.011},
  \url{https://doi.org/10.1016/j.jpdc.2019.03.011}.

\bibitem{SomeSimplifiedGarey1974}
{\sc M.~R. Garey, D.~S. Johnson, and L.~Stockmeyer}, {\em Some simplified
  np-complete problems}, in Proceedings of the sixth annual {ACM} symposium on
  {Theory} of computing - {STOC} '74, Seattle, Washington, United States, 1974,
  ACM Press, pp.~47--63, \url{https://doi.org/10.1145/800119.803884},
  \url{http://portal.acm.org/citation.cfm?doid=800119.803884}.

\bibitem{BoostingDataRGellne2021}
{\sc A.~Gellner, S.~Lamm, C.~Schulz, D.~Strash, and B.~Zav\'{a}lnij}, {\em
  Boosting data reduction for the maximum weight independent set problem using
  increasing transformations}, in Proceedings of the {Symposium} on {Algorithm}
  {Engineering} and {Experiments}, {ALENEX} 2021, {Virtual} {Conference}, {
  January} 10-11, 2021, M.~Farach-Colton and S.~Storandt, eds., SIAM, 2021,
  pp.~128--142, \url{https://doi.org/10.1137/1.9781611976472.10},
  \url{https://doi.org/10.1137/1.9781611976472.10}.

\bibitem{EvaluationOfLGemsa2014}
{\sc A.~Gemsa, M.~N\"{o}llenburg, and I.~Rutter}, {\em Evaluation of labeling
  strategies for rotating maps}, in Experimental {Algorithms}, D.~Hutchison,
  T.~Kanade, J.~Kittler, J.~M. Kleinberg, A.~Kobsa, F.~Mattern, J.~C. Mitchell,
  M.~Naor, O.~Nierstrasz, C.~Pandu~Rangan, B.~Steffen, D.~Terzopoulos,
  D.~Tygar, G.~Weikum, J.~Gudmundsson, and J.~Katajainen, eds., vol.~8504,
  Springer International Publishing, Cham, 2014, pp.~235--246,
  \url{https://doi.org/10.1007/978-3-319-07959-2\_20},
  \url{http://link.springer.com/10.1007/978-3-319-07959-2\_20}.
\newblock Series Title: Lecture Notes in Computer Science.

\bibitem{DistributedKerGeorge2018}
{\sc T.~George and D.~Hespe}, {\em Distributed kernelization for independent
  sets}, November 2018.

\bibitem{FindingNearOpGrossma2023}
{\sc E.~Gro\ss{}mann, S.~Lamm, C.~Schulz, and D.~Strash}, {\em Finding
  near-optimal weight independent sets at scale}, in Proceedings of the
  {Genetic} and {Evolutionary} { Computation} {Conference}, Lisbon Portugal,
  July 2023, ACM, pp.~293--302, \url{https://doi.org/10.1145/3583131.3590353},
  \url{https://dl.acm.org/doi/10.1145/3583131.3590353}.

\bibitem{AcceleratingReGrossma2024}
{\sc E.~Gro\ss{}mann, K.~Langedal, and C.~Schulz}, {\em Accelerating reductions
  using graph neural networks and a new concurrent local search for the maximum
  weight independent set problem},  (2024),
  \url{http://arxiv.org/abs/2412.14198v1},
  \url{https://arxiv.org/abs/2412.14198v1}.

\bibitem{AComprehensiveGrossma2024}
{\sc E.~Gro\ss{}mann, K.~Langedal, and C.~Schulz}, {\em A comprehensive survey
  of data reduction rules for the maximum weighted independent set problem},
  (2024), \url{http://arxiv.org/abs/2412.09303v1},
  \url{https://arxiv.org/abs/2412.09303v1}.

\bibitem{TowardsComputiGuJi2021}
{\sc J.~Gu, W.~Zheng, Y.~Cai, and P.~Peng}, {\em Towards computing a
  near-maximum weighted independent set on massive graphs}, in Proceedings of
  the 27th {ACM} {SIGKDD} {Conference} on { Knowledge} {Discovery} \& {Data}
  {Mining}, Virtual Event Singapore, August 2021, ACM, pp.~467--477,
  \url{https://doi.org/10.1145/3447548.3467232},
  \url{https://dl.acm.org/doi/10.1145/3447548.3467232}.

\bibitem{TargetedBranchHespe2021}
{\sc D.~Hespe, S.~Lamm, and C.~Schorr}, {\em Targeted branching for the maximum
  independent set problem}, in 19th International Symposium on Experimental
  Algorithms, { SEA} 2021, June 7-9, 2021, Nice, France, D.~Coudert and
  E.~Natale, eds., vol.~190 of LIPIcs, Schloss Dagstuhl - Leibniz-Zentrum
  f{\"{u}}r Informatik, 2021, pp.~17:1--17:21,
  \url{https://doi.org/10.4230/LIPICS.SEA.2021.17},
  \url{https://doi.org/10.4230/LIPIcs.SEA.2021.17}.

\bibitem{ScalableKernelHespe2019}
{\sc D.~Hespe, C.~Schulz, and D.~Strash}, {\em Scalable kernelization for
  maximum independent sets}, ACM Journal of Experimental Algorithmics, 24
  (2019), pp.~1--22, \url{https://doi.org/10.1145/3355502},
  \url{https://dl.acm.org/doi/10.1145/3355502}.

\bibitem{DBLP:conf/ipps/HoltgreweSS10}
{\sc M.~Holtgrewe, P.~Sanders, and C.~Schulz}, {\em Engineering a scalable high
  quality graph partitioner}, in 24th {IEEE} International Symposium on
  Parallel and Distributed Processing, {IPDPS} 2010, Atlanta, Georgia, USA,
  19-23 April 2010 - Conference Proceedings, {IEEE}, 2010, pp.~1--12,
  \url{https://doi.org/10.1109/IPDPS.2010.5470485},
  \url{https://doi.org/10.1109/IPDPS.2010.5470485}.

\bibitem{DistributedGreJooC2016}
{\sc C.~Joo, X.~Lin, J.~Ryu, and N.~B. Shroff}, {\em Distributed greedy
  approximation to maximum weighted independent set for scheduling with fading
  channels}, IEEE/ACM Transactions on Networking, 24 (2016), pp.~1476--1488,
  \url{https://doi.org/10.1109/tnet.2015.2417861},
  \url{https://doi.org/10.1109/tnet.2015.2417861}.

\bibitem{DBLP:conf/stoc/KarpW84}
{\sc R.~M. Karp and A.~Wigderson}, {\em A fast parallel algorithm for the
  maximal independent set problem}, in Proceedings of the 16th Annual {ACM}
  Symposium on Theory of Computing, April 30 - May 2, 1984, Washington, DC,
  {USA}, R.~A. DeMillo, ed., {ACM}, 1984, pp.~266--272,
  \url{https://doi.org/10.1145/800057.808690},
  \url{https://doi.org/10.1145/800057.808690}.

\bibitem{FindingNearOpLamm2017}
{\sc S.~Lamm, P.~Sanders, C.~Schulz, D.~Strash, and R.~F. Werneck}, {\em
  Finding near-optimal independent sets at scale}, Journal of Heuristics, 23
  (2017), pp.~207--229, \url{https://doi.org/10.1007/s10732-017-9337-x},
  \url{http://link.springer.com/10.1007/s10732-017-9337-x}.

\bibitem{ExactlySolvingLamm2019}
{\sc S.~Lamm, C.~Schulz, D.~Strash, R.~Williger, and H.~Zhang}, {\em Exactly
  solving the maximum weight independent set problem on large real-world
  graphs}, in Proceedings of the {Twenty}-{First} {Workshop} on { Algorithm}
  {Engineering} and {Experiments}, {ALENEX} 2019, { San} {Diego}, {CA}, {USA},
  {January} 7-8, 2019, S.~G. Kobourov and H.~Meyerhenke, eds., SIAM, 2019,
  pp.~144--158, \url{https://doi.org/10.1137/1.9781611975499.12},
  \url{https://doi.org/10.1137/1.9781611975499.12}.

\bibitem{TargetedBranchLanged2024}
{\sc K.~Langedal, D.~Hespe, and P.~Sanders}, {\em Targeted branching for the
  maximum independent set problem using graph neural networks}, in 22nd
  International Symposium on Experimental Algorithms (SEA 2024), L.~Liberti,
  ed., vol.~301 of Leibniz International Proceedings in Informatics (LIPIcs),
  Dagstuhl, Germany, 2024, Schloss Dagstuhl -- Leibniz-Zentrum f{\"u}r
  Informatik, pp.~20:1--20:21,
  \url{https://doi.org/10.4230/LIPIcs.SEA.2024.20},
  \url{https://drops.dagstuhl.de/entities/document/10.4230/LIPIcs.SEA.2024.20}.

\bibitem{snapnets}
{\sc J.~Leskovec and A.~Krevl}, {\em {SNAP Datasets}: {Stanford} large network
  dataset collection}.
\newblock \url{http://snap.stanford.edu/data}, June 2014.

\bibitem{ApplicationOfLiuJ2023}
{\sc J.~Liu, S.~Shao, and C.~Zhang}, {\em Application of causal inference
  techniques to the maximum weight independent set problem},  (2023),
  \url{http://arxiv.org/abs/2301.05510v1},
  \url{https://arxiv.org/abs/2301.05510v1}.

\bibitem{ASimpleParallLuby1985}
{\sc M.~Luby}, {\em A simple parallel algorithm for the maximal independent set
  problem}, in Proceedings of the seventeenth annual {ACM} symposium on {
  Theory} of computing - {STOC} '85, Providence, Rhode Island, United States,
  1985, ACM Press, pp.~1--10, \url{https://doi.org/10.1145/22145.22146},
  \url{http://portal.acm.org/citation.cfm?doid=22145.22146}.

\bibitem{DBLP:journals/dc/MetivierRSZ11}
{\sc Y.~M{\'{e}}tivier, J.~M. Robson, N.~Saheb{- }Djahromi, and A.~Zemmari},
  {\em An optimal bit complexity randomized distributed {MIS} algorithm},
  Distributed Comput., 23 (2011), pp.~331--340,
  \url{https://doi.org/10.1007/S00446-010-0121-5},
  \url{https://doi.org/10.1007/s00446-010-0121-5}.

\bibitem{AHybridIteratNoguei2018}
{\sc B.~Nogueira, R.~G.~S. Pinheiro, and A.~Subramanian}, {\em A hybrid
  iterated local search heuristic for the maximum weight independent set
  problem}, Optimization Letters, 12 (2018), pp.~567--583,
  \url{https://doi.org/10.1007/s11590-017-1128-7},
  \url{https://doi.org/10.1007/s11590-017-1128-7}.

\bibitem{nr}
{\sc R.~A. Rossi and N.~K. Ahmed}, {\em The network data repository with
  interactive graph analytics and visualization}, in AAAI, 2015,
  \url{https://networkrepository.com}.

\bibitem{DBLP:journals/tog/SanderNCH08}
{\sc P.~V. Sander, D.~Nehab, E.~Chlamtac, and H.~Hoppe}, {\em Efficient
  traversal of mesh edges using adjacency primitives}, {ACM} Trans. Graph., 27
  (2008), p.~144, \url{https://doi.org/10.1145/1409060.1409097},
  \url{https://doi.org/10.1145/1409060.1409097}.

\bibitem{DBLP:conf/europar/SandersS23}
{\sc P.~Sanders and D.~Seemaier}, {\em Distributed deep multilevel graph
  partitioning}, in Euro-Par 2023: Parallel Processing - 29th International
  Conference on Parallel and Distributed Computing, Limassol, Cyprus, August 28
  - September 1, 2023, Proceedings, J.~Cano, M.~D. Dikaiakos, G.~A.
  Papadopoulos, M.~Peric{\`{a}}s, and R.~Sakellariou, eds., vol.~14100 of
  Lecture Notes in Computer Science, Springer, 2023, pp.~443--457,
  \url{https://doi.org/10.1007/978-3-031-39698-4\_30},
  \url{https://doi.org/10.1007/978-3-031-39698-4\_30}.

\bibitem{EngineeringADSander2023}
{\sc P.~Sanders and T.~N. Uhl}, {\em Engineering a distributed-memory triangle
  counting algorithm}, in 2023 IEEE International Parallel and Distributed
  Processing Symposium (IPDPS), IEEE, 5 2023, pp.~702--712,
  \url{https://doi.org/10.1109/ipdps54959.2023.00076},
  \url{https://doi.org/10.1109/ipdps54959.2023.00076}.

\bibitem{uhl2025briefkasten}
{\sc N.~Uhl}, {\em Briefkasten}.
\newblock \url{https://github.com/niklas-uhl/BriefKAsten}, 2025.
\newblock GitHub repository, accessed 2025-07-23.

\bibitem{DBLP:conf/sc/UhlSHHKSS024}
{\sc T.~N. Uhl, M.~Schimek, L.~H{\"{u}}bner, D.~Hespe, F.~Kurpicz, D.~Seemaier,
  C.~Stelz, and P.~Sanders}, {\em Kamping: Flexible and (near) zero-overhead
  {C++} bindings for { MPI}}, in Proceedings of the International Conference
  for High Performance Computing, Networking, Storage, and Analysis, { SC}
  2024, Atlanta, GA, USA, November 17-22, 2024, {IEEE}, 2024, p.~44,
  \url{https://doi.org/10.1109/SC41406.2024.00050},
  \url{https://dl.acm.org/doi/10.1109/SC41406.2024.00050}.

\bibitem{DistributedNeaWang2023}
{\sc X.~Wang, D.~Wen, W.~Zhang, Y.~Zhang, and L.~Qin}, {\em Distributed
  near-maximum independent set maintenance over large-scale dynamic graphs}, in
  2023 {IEEE} 39th {International} {Conference} on {Data} { Engineering}
  ({ICDE}), Anaheim, CA, USA, April 2023, IEEE, pp.~2538--2550,
  \url{https://doi.org/10.1109/ICDE55515.2023.00195},
  \url{https://ieeexplore.ieee.org/document/10184836/}.

\bibitem{EfficientReducXiao2021}
{\sc M.~Xiao, S.~Huang, Y.~Zhou, and B.~Ding}, {\em Efficient reductions and a
  fast algorithm of maximum weighted independent set}, in Proceedings of the
  {Web} {Conference} 2021, Ljubljana Slovenia, April 2021, ACM, pp.~3930--3940,
  \url{https://doi.org/10.1145/3442381.3450130},
  \url{https://dl.acm.org/doi/10.1145/3442381.3450130}.

\end{thebibliography}

\clearpage
\ifshowFull
\appendix
\section{Omitted Reduction Model Proofs}\label{appendix:sec:proofs}
This section provides the omitted proofs for the key properties from our distributed reduction model.
\setcounter{theorem}{1}
\begin{lemma}[Upper Bound for Ghost Weights, \Cref{lem:weights}]\label{appendix:lem:weights}
	Let $v\in \ghosts{i}$.
	It follows that $\w_i(v)\geq \w_{\rank(v)}(v)$ or $v$ is globally reduced,~\ie $v\not\in V$.
\end{lemma}
\begin{proof}
	Let $v\in \ghosts{i}$ and ssume $v\in V$.
	Consider PE $j=\rank(v)$ which owns the vertex $v$ and its weight.
	According to the \nameref{def:distRedModel}, the weight of $v$ can only be modified by PE~$j$.
	Further, for interface vertices the weight modification is limited to decreasing the weight.
	Thus, PE $j$ always knows an upper bound for the weight \hbox{of $v$,~\ie $w_i(v)\geq w_j(v)$.}
\end{proof}

\begin{lemma}[Neighborhood, \Cref{lem:neighborhood}]\label{appendix:lem:neighborhood}
	Let $v\in V$.
	If $v\in V_i$, it holds $N(v)\subseteq N_i(v)$ where the vertices of $N_i(v)\setminus N(v)$ are already reduced \hbox{by other PEs.}
	If $v\in \ghosts{i}$, it holds $N(v)\cap V_i \subseteq N_i(v)$.
\end{lemma}
\begin{proof}
	First, consider $v\in V_i\cap V$ and let $u\in N(v)$.
	If $u\in V_i$, both vertices are local and therefore, $u\in N_i(v)$.
	New cut-edges cannot be inserted in our \nameref{def:distRedModel}.
	Cut-edges from the distributed input graph $G$ are replicated at the respective PEs.
	Thus, we obtain $u\in N_i(v)$ if $\set{u,v}$ is a cut-edge in $\G$.

	Now consider, $v\in \ghosts{i}\cap V$.
	At PE $i$ we stored only the neighbors of $v$ that were assigned to PE $i$.
	Since new cut-edges cannot be inserted, we \hbox{obtain $N(v)\cap V_i \subseteq N_i(v)$.}
\end{proof}

\begin{lemma}[Include Alternatives, \Cref{lem:IncInterfaceVertices}]\label{appendix:lem:IncInterfaceVertices}
	Let $\G=(V,E,\w)$ be the global reduced graph.
	Let $G_i=(\cV_i, E_i, \w_i)$ the local graph at PE $i$ and $G_j=(\cV_j, E_j, \w_j)$ the local graph at PE $j\neq i$.
	Assume, it exists $v\in V_i$ and $u\in V_j$ which are now proposed to be included at PE $i$ and PE $j$, respectively.
	Further, assume $v\in N_j(u)$ at PE $j$ so that $v$ is reduced when $u$ is proposed to be included.
	Then, it holds $u\in N_i(v)$, $\w_i(v)=\w_j(u)$, $N_i(v)\cap V \subseteq \set{u}$, and $\N_j(u)\cap V \subseteq \set{v}$.
	If $u$ exists for the given $v$, then there is no local graph $G_k=(\cV_k, E_k, \w_k)$ at some PE $k\not\in\set{i,j}$ with a vertex $u'\in V_k\setminus \set{u}$ that reduces $v$ as well.
\end{lemma}
\begin{proof}
	We give a proof by induction on the number of interface vertices $l$ proposed to be included at PE~$i$.
	Without lost of generality we assume that each of these interface vertices is reduced by another PE.
	Otherwise there exists no gost neighbor that is proposed to be included as well.
	In the induction beginning and step, we assume $\G$, $G_i$, $G_j$, $G_k$, $v\in V_i$ and $u\in V_j$ are given as in \Cref{lem:IncInterfaceVertices} where $v$ is the $l$-th proposed to be included vertex at PE $i$.

	\textit{Induction beginning}.
	First, it holds $\w_j(u)\geq \solw{G_j[N_j(u)]} \geq w_j(v)$ because $u$ is proposed to be included with $v\in N_j(u)$.
	Further, we obtain $\w_j(u)\geq \w_j(v)\geq \w_i(v)$ with \Cref{lem:weights}.

	We start by proving that $u\in N_i(v)$.
	For a proof by contradiction, assume $u\not\in N_i(v)$.
	Then $u$ was reduced as ghost at PE $i$.
	In our \nameref{def:distRedModel} this is only possible if a vertex was proposed to be included before $v$ at PE $i$.
	Since, we assumed that $v$ is the first vertex which is proposed to be included at PE $i$, we obtain a contradition.

	With $u\in N_i(v)$, we now analogously obtain $\w_i(v)\geq \w_j(u)$ since $v$ was proposed to be included as well.
	Thus, $u$ and $v$ have equal weights.

	Next, we show $N_i(v)\cap V \subseteq \set{u}$,~\ie $u$ is the only neighbor of $v$ in $G_i$ that might not yet be globally reduced.
	To that end, assume there exists a vertex $x\in (N_i(v)\cap V)\setminus \set{u}$.
	Note, for $x$ holds $x\in V_j$ since $x\in V$.
	Then there are two cases to consider.
	First, assume $x$ is also a neighbor of $u$ in $G_i$.
	Then, $X=\set{x, u}$ is independent in $G_j$ because both vertices are ghosts in $G_j$ and therefore, they are not connected by and edge.
	The independent set $X$ has larger weight than $u$ since $w_j(x)>0$ and $w_j(v)=w_j(u)$.
	This contradicts that $u$ is proposed to be included because $\w_j(u)<\w_j(X)\leq \solw{G_j[N_j(u)]}$.
	Thus, $x$ cannot exist and we obtain $N_i(v)\cap V\subseteq \set{u}$.
	The proof for $N_j(u)\cap V\subseteq \set{v}$ is analogous.

	It remains to show that there is no PE $k\not\in \set{i,j}$ with another $u'\in V_k\setminus {u}$ which reduces $v$,~\ie $v\in N_k(u')$.
	Analogous to $u$, we can show that $u'\in N_i(v)$ and further $\w_i(v)=\w_i(u')$.
	Since both $u$ and $u'$ are ghosts in $G_i$, they are not adjacent in $G_i$.
	Again, we obtain an independent set $\set{u,u'}$ in $G_i$ of larger weight than $\w_i(v)$ which contradicts that $v$ is proposed to be included.
	Thus, it must hold $u'=u$.

	\textit{Induction hypothesis}.
	Now assume, these properties hold until the $l$-th proposed to be included vertex at PE $i$.

	\textit{Induction step}.
	We show that this still holds for the ${l+1}$ proposed to be included vertex at PE $j$.
	Similar as in the beginning, we only need to prove that $u\in N_i(v)$ to obtain $w_i(v)=w_j(u)$.
	Therefore, assume $u\not\in N_i(v)$.
	Then $u$ must have been reduced as ghost at PE $i$ by a previous proposal to include a vertex.
	Let $G_i^*=(\cV_i^*, E_i^*,\w_i^*)$ denote the local graph at PE $i$ when a vertex $v'\in V_i^*$ was proposed to be included with $u\in N_i^*(v')$.
	We distinguish two cases.

	First, assume $v'$ is still a neighbor of $u$ in $G_j$,~\ie $v'\in N_j(u)$.
	Then, our induction hypothesis applies for $v'$ and $u$ and we obtain $\w_i^*(v)=\w_j(u)$.
	Moreover, $v'$ and $v$ are two non-adjacent ghosts in $G_j$.
	Together they form an independent set in $G_j$ with a weight of $\w_j(\set{v,v'})=\w_j(u)+w_j(v) > \w_j(u)$.
	This contradicts that $u$ is proposed to be included.

	Second, assume $v'$ is already reduced at PE $j$,~\ie $v'\not\in N_j(u)$.
	Let $G_j^*(\cV_j^*,E_j^*,\w_j*)$ be the local graph at PE $j$ when $v'$ was reduced at PE $j$.
	Since $v'$ is a ghost at $G_j^*$, a vertex $u'\in V_j^*$ was proposed to be included with $v'\in N_j^*(u')$.
	We can apply our induction hypothesis to $v'$ and $u'$ and obtain $w_i^*(v')=w_j^*(u')$ with $u'\in N_j^*(v')$.
	Remember that $u\in N_i^*(v')$.
	Thus, $u$ and $u'$ are both neighbors of $v'$ in $G_i^*$.
	Further, they form an independent set because they are ghosts in $G_i^*$.
	This contradicts $v'$ being proposed to be included because $\w_i^*(u,u')=\w_i^*(u)+\w_i^*(v)$.

	Overall, this contradicts that there is a vertex $v'$ which reduced $u$ at PE $i$,~\ie $u\in V_i$.
	We obtain $u\in N_i(v)$ and $w_i(v)=w_j(u)$.
	The remainder of the proof for the induction step is analogous to the one in the induction beginning.
\end{proof}

\begin{lemma}[Include Operation, \Cref{lem:includeOperation}]\label{appendix:lem:includeOperation}%
	Let ${\G=(V,E,\w)}$ be the global reduced graph and let ${G_i=(V_i,E_i,\w_i)}$ be the local graph at PE $i$.
	Furthermore, let $v\in V_i$ be a vertex that was not yet reduced by any PE, \ie $v \in V$.
	If $v$ is an interface vertex, assume $\w_i(v)\geq \solw{G_i[\N_i(v)]}$ in addition.
	Assume, $v$ can be included into an MWIS of $\G$ so that $\G$ is reduced to $\G'=\G-N(v)$ with ${\solw{\G} = \solw{\G'}+\w_i(v)}$.

	Then, only PE $i$ needs to modify its local graph so that $\G'$ is obtained globally.
	The local reduced graph is given as ${G_i'=G_i-N_i[v]}$.
	Furthermore, add $v$ to the set ${\addedI_i'=\addedI_i\cup \set{v}}$ at PE $i$ if $v$ is an interface vertex.
	An MWIS of $\G$ can be reconstructed after exchanging the solution proposals of $N_i[v]$ with the adjacent PEs of $v$.
	Reconstruct the solution as follows:
	\begin{itemize}
		\item If $u$ exists with $j < i$, set ${\I_i = \I_i'}$ at PE $i$ and~${\I_j = \I_j' \cup \set{u}}$ at PE~$j$.
		\item Else if $u$ exists with $j > i$, set ${\I_i = \I_i'\cup \set{v}}$ at PE~$i$ and ${\I_j = \I_j'}$ at PE~$j$.
		\item Else, add $v$ to the local solution,~\ie $\I_i = \I_i'\cup \set{v}$.
	\end{itemize}
\end{lemma}
\begin{proof}
	Assume $v$ is given as stated above.

	First, we note that the global graph $\G$ is reduced to ${\G'=\G-N(v)}$ when $v$ is included into an MWIS of $\G$ with a reduction offset of $\solw{\G}=\solw{\G'}+\w(v)=\solw{\G'}+\w_i(v)$.
	With \Cref{lem:neighborhood}, it holds $N(v)\subseteq N_i(v)$ where $(N_i(v)\setminus N(v))\cap V=\emptyset$.
	Therefore, the local reduced graph is given by ${G_i'=G_i-N_i[v]}$.

	If $v$ is not an interface vertex, the reconstruction is straightforward because $v$ has no ghost neighbors (anymore).
	In this case, there exists no neighbor that can be proposed to be included at another PE.
	Therefore, the solution can be reconstructed by adding $v$ to the local share $\I_i$ of an MWIS for $\G$.

	In the next part we assume $v$ to be an interface vertex.
	If there is no neighbor of $v$ that is proposed to be included as well, the reconstruction is possible without any conflict.
	According to \Cref{lem:IncInterfaceVertices}, there can at most be one neighbor that is proposed to be included as well.
	Let this neighbor be $u\in N(v)$ and $G_j^*$ denote the reduced graph at PE $j$ when $u$ is proposed to be included.
	Without lost of generality, we assume $\G$ is the global reduced graph where PE $j$ already considers the local reduced graph $G_j^*$.
	Then both vertices, $u$ and $v$, have equal weights and remaining neighbors $u$ and $v$, respectively.
	Therefore, $u$ and $v$ can be both in an MWIS of $\G$.%
	For the reconstruction of the solution, we can now choose between $v$ and $u$ while both are eligible for an MWIS of $\G$.
	The tie-breaking by the ranks of both vertices ensures that exactly one vertex is included.
	Moreover, it ensures that either all or none of the conflicting proposals between PE $i$ and $j$ are accepted at PE $i$.
\end{proof}

\setcounter{theorem}{6}
\begin{lemma}[Remaining Vertex, \Cref{lem:remainingVertex}]\label{appendix:lem:remainingVertex}
	Let $u,v\in V_i$ be adjacent with $\max\set{\w_i(x): x\in (N_i(u)\cap \ghosts{i})\setminus N_i(v)} + \w_i(v) \leq \w_i(u)$.
	Then it holds that $u\in V$ if $v\in V$.
\end{lemma}
\begin{proof}
	Consider neighbors $u,v\in V_i$ with $\max\set{\w_i(x): x\in (N_i(u)\cap \ghosts{i})\setminus N_i(v)} \leq \w_i(u)$.
	For a proof by contradiction, assume $u\not\in V$ while $v\in V$.
	Then $u$ was reduced due to proposing to include an interface vertex $x$ at another PE $j$.
	Let $G_j^*$ refer to PE $j$ when $x$ was proposed to be included.
	Furthermore, $x$ cannot be adjacent to $v$; otherwise $v\not\in V$.

	Moreover, it holds that $x\in \N_i(u)$ at PE $i$.
	Otherwise $x$ had another neighbor at PE $\pe{i}$ which was proposed to be included and further reduced $x$.
	Let $G_u^*(\cV_i^*,E_i^*,\w_i^*)$ be the local graph at PE $i$ when a vertex $y\in V_i^*$ was proposed to be included with $x\in N_i^*(y)$.
	Let $\G^*(V^*, E^*, \w^*)$ be the global reduced graph where PE $i$ considered $G_i^*$ and PE $j$ considered $G_j^*$.
	With \Cref{lem:IncInterfaceVertices} we obtain for $x$ and $y$ that $y\in N_j^*(x)$ and $w_i^*(y)=w_j^*(x)$.
	Thus, $\set{y, u}$ for an independent set of larger weight in $G_j^*$.
	This contradicts that $x$ has been proposed to be included at PE $j$.
	Consequently, $x$ was not yet reduced by some $y$ at PE $i$ and therfore, it holds $x\in \N_i(u)$.

	We show now that $\w_i(u)=\w_j^*(x)$.
	First, $x$ is a ghost neighbor of $u$ with $x\not\in N_i(v)$ at PE $i$ and therefore, $\w_i(u)\geq \max\set{\w_i(x): x\in (N_i(u)\cap \ghosts{i})\setminus N_i(v)}\geq \w_i(x)$.
	With \Cref{lem:weights}, we obtain $\w_i(u)\geq w_i(x) \geq w_j^*(x)$.
	Second, it holds $\w_j^*(x)\geq w_i(u)$ because $x$ was proposed to be included with $u\in N_j^*(x)$.
	Thus, we obtain $\w_i(u)=\w_j^*(x)$.
	However, this contradicts $\w_i(u)\geq \max\set{\w_i(x): (x\in N_i(u)\cap \ghosts{i})} + \w_i(v) \geq \w_i(x)+\w_i(v) > \w_i(x)$.
	We conclude that $x$ was not proposed to be included yet and further, $x$ is not reduced yet by any PE.
\end{proof}

\section{Omitted Reduction Proofs}\label{appendix:sec:reductionProofs}
This section provides the omitted proofs for our distributed reduction rules.

\makeatletter
\setcounter{@distributedreduction}{1}
\makeatother
\begin{distributedreduction}[Distributed Heavy Vertex, \Cref{dred:HeavyVertex}]\label{appendix:dred:HeavyVertex}
	Let $v\in V_i$ with $\w_i(v)\geq \solw{G_i[N_i(v)]}$.
	Then (propose to) include $v$ following Remark~\ref{rem:includeVertex}.
\end{distributedreduction}
\begin{proof}
	Let $\G=(V,E,\w)$ be the global reduced graph and $G_i=(V_i,E_i,\w_i)$ be the local graph at PE $i$.
	Further, let $v\in V_i$ with $\w_i(v)\geq \solw{G_i[N_i(v)]}$.
	First, assume $v$ is not yet reduced by any PE,~\ie $v\in V$.
	We show that $\w(v)\geq \solw{\G[N(v)]}$.
  Then, it holds with \emph{Heavy Vertex}~\cite{ExactlySolvingLamm2019} that $v$ is in an MWIS of $\G$.

	First, we obtain $N(x)\subseteq N_i(x)$ for every $x\in N[v]$ with \Cref{lem:neighborhood}.
	Furthermore, for $u\in N(v)$ holds that $\w_i(u)\geq \w(u)$ with \cref{lem:weights}.
	Thus, it holds $\w(v)=\w_i(v)\geq \solw{G_i[N_i(v)]} \geq \solw{\G[N(v)]}$.
	With \Cref{lem:includeOperation}, we can (propose to) include $v$ into an MWIS of $\G$.

	If $v$ was already reduced by another PE, then $v$ and its neighbors were globally reduced.
	In this case, we can still propose to include $v$ as explained in Remark~\ref{rem:includeVertex}.
\end{proof}

\begin{distributedreduction}[Dist. Neighborhood Removal, \Cref{dred:NeighborhoodRemoval}]\label{appendix:dred:NeighborhoodRemoval}
	Let $v\in V_i$ with $\w_i(v)\geq\w_i(\N_i(v))$.
	Then (propose to) include $v$ following Remark~\ref{rem:includeVertex}.
\end{distributedreduction}
\begin{proof}
	We show that this reduction is a special case of \nameref{dred:HeavyVertex}.
	Let $v\in V_i$ with $\w_i(v)\geq\w_i(\N_i(v))$.
	Then the weight of the neighborhood is an upper bound of an MWIS in the induced neighborhood graph,~\ie $\w_i(\N_i(v))\geq \solw{G_i[\N_i(v)]}$.
	Thus, we can apply the \nameref{dred:HeavyVertex}.
	This yields the described reduced graph, offset and reconstruction.
\end{proof}

\begin{distributedreduction}[Distributed Simplicial Vertex, \Cref{dred:SimplicialVertex}]\label{appendix:dred:SimplicialVertex}
	Let $v\in V_i$ be a simplicial vertex with maximum weight in its neighborhood in $G_i$,~\ie it holds $\w_i(v)\geq \max \set{\w_i(u): u\in \N_i(v)}$.
	Then (propose to) include $v$ following Remark~\ref{rem:includeVertex}.
\end{distributedreduction}
\begin{proof}
	To prove correctness, we show that this reduction is a special case of \nameref{dred:HeavyVertex}.
	Let $v\in V_i$ be a simplicial vertex with $\w_i(v)\geq \max \set{\w_i(u): u\in \N_i(v)}$.
	Then, $\N_i(v)$ forms a clique in $G_i$.
	Thus, an {\MWIS} of the induced neighborhood graph of $G_i[\N_i(v)]$ can consist of at most one vertex.
	It follows ${\w_i(v) \geq \max \set{\w_i(u): u\in \N_i(v)} \geq \solw{G_i[\N_i(v)]}}$.
	Thus, we can apply \nameref{dred:HeavyVertex}.
	This yields the described reduced graph, offset and proposed solution reconstruction.
\end{proof}

\begin{distributedreduction}[Dist. Simp. Weight Transfer, \Cref{dred:SimplicialWeightTransfer}]\label{appendix:dred:SimplicialWeightTransfer}
	Let $v\in V_i\setminus N_i(\ghosts{i})$ be a simplicial vertex, let $S(v)\subseteq N_i(v)$ be the set of all simplicial vertices.
	Further, let $w_i(v)\geq w_i(u)$ for all $u\in S(v)$.
	If $\w_i(v) < \max \set{w_i(u): u \in \N_i(v)}$, fold $v$ into $N_i(v)$.
	\begin{tabular}{ll}
		\reducedGraphRow{$G_i'=G_i-X$ with                                                    \\ $X=\set{u\in \N_i[v]: \w_i(u)\leq \w_i(v)}$ \\ and set $w_i'(u)= w_i(u)-w_i(v)$ \\ for all $x\in N_i(v)\setminus X$.} \\
		\offsetRow{$\solw{\G} = \solw{\G'}+\w(v)$.}                                           \\
		\reconstructionRow{If $\I_i'\cap\N_i(v) = \emptyset$, then $\I_i =\I_i'\cup \set{v}$, \\ else $\I_i=\I_i'$.}
	\end{tabular}
\end{distributedreduction}
\begin{proof}
	Let $v\in V_i\setminus N_i(\ghosts{i})$ be a simplicial vertex with $\w_i(v) < \max \set{w_i(u): u \in \N_i(v)}$.
	Note that $v$ is a local non-interface vertex.
	Further, let $S(v)\subseteq N_i(v)$.
	First, we observe that the proposed modification ensures the \nameref{def:distRedModel}.
	In the neighborhood of $v$ are only local neighbors which can be interface vertices.
	Therefore, ghost vertices remain part of the graph and their weight is not modified.
	Interface vertices are only excluded or their weight is decreased.

	It is possible that some of these interface vertices are already excluded in $\G$,~\ie $U=\N_i(v)\setminus\N_i(v)\neq \emptyset$.
	Independent of whether vertices of $U$ are excluded, $v$ remains simplicial in $G_i$.
	Then $v$ is simplicial in $G$ since $G[N[v]]=G_i[N_i[v]\setminus U]$.
  Thus, the second case of \emph{Simplicial Weight Transfer}~\cite{ExactlySolvingLamm2019} applies for $v$ in $G$.
	This yields the reduced graph $G_i'$, offset, and reconstruction.
\end{proof}

\begin{distributedreduction}[Distributed Basic Single-Edge, \Cref{dred:BasicSingleEdge}]\label{appendix:dred:BasicSingleEdge}
	Let $u,v\in V_i$ be adjacent vertices with $\w_i(\N_i(u)\setminus \N_i(v)) \leq \w_i(u)$, then exclude $v$.

	\begin{tabular}{ll}
		\textit{Reduced Graph}  & $G_i'=G_i-v$             \\
		\textit{Offset}         & $\solw{\G} = \solw{\G'}$ \\
		\textit{Reconstruction} & $\I_i = \I_i'$
	\end{tabular}
\end{distributedreduction}
\begin{proof}
	Let $u,v\in V_i$ be adjacent vertices with $\w_i(\N_i(u)\setminus \N_i(v)) \leq \w_i(u)$.

	Assume $v\in V$; otherwise $v$ is already reduced in $\G$ and we only reduce a redundant vertex $v$ in $G_i$.
	It holds that $\max\{\w_i(x): x\in(N_i(u)\cap \ghosts{i})\setminus N_i(v)\} + w_i(v)\leq \w_i(\N_i(u)\setminus \N_i(v)) \leq \w_i(u)$ and since $v\in V$, we obtain $u\in V$ with \cref{lem:remainingVertex}.
	Now consider $U=N_i(u)\setminus N_i(v)$.
	The set $U$ might contain border vertices which are already reduced or received smaller weights.
	Thus, it holds $\w(U) \geq \w(N(u)\setminus N(v))$, and we obtain $\w(u)\geq \w(N(u)\setminus N(v))$.
	Therefore, we can apply \emph{Basic Single-Edge} by Gu~\etal\cite{TowardsComputiGuJi2021} for $u$ and $v$ to exclude $v$ in $\G$.
	The reduction application yields the reduced graph $G_i'=G_i-v$ together with $\solw{G} = \solw{G'}$ and $\I_i = \I_i'$.
	In conclusion, we can either apply \emph{Basic Single-Edge} by Gu~\etal\cite{TowardsComputiGuJi2021} to $v$ in $\G$ or $v$ is already reduced \hbox{by another PE.}
\end{proof}

\begin{distributedreduction}[Dist. Extended Single-Edge, \Cref{dred:ExtendedSingleEdge}]\label{appendix:dred:ExtendedSingleEdge}
	Let $u,v\in V_i$ be adjacent vertices with $\w_i(\N_i(v)) - \w_i(u) \leq \w_i(v)$, then exclude $X=\N_i(v)\cap\N_i(u)\setminus \ghosts{i}$.

	\begin{tabular}{ll}
		\textit{Reduced Graph}  & $G_i'=G_i-X$           \\
		\textit{Offset}         & $\solw{G} = \solw{G'}$ \\
		\textit{Reconstruction} & $\I_i = \I_i'$
	\end{tabular}
\end{distributedreduction}
\begin{proof}
	Let $u,v\in V$ be adjacent vertices with $\w_i(\N_i(v)) - \w_i(u) \leq \w_i(v)$.
	Furthermore, let $X=\N_i(v)\cap\N_i(u)\setminus \ghosts{i}$ be non-empty.
	We show with a case distinction that we can reduce $X$ from $G_i$ and consequently exclude them from an MWIS for $\G$ independent of whether $v$ or $u$ are already reduced in $\G$.

	First, we consider $v\not\in V$.
	We show that $X\cap V=\emptyset$.
	Then we reduce only redundant vertices in $G_i$ but do not change $\G$.
	For a proof by contradiction, assume it exists $x\in X\cap V$.
	With \cref{lem:remainingVertex} and $x\in V$, we obtain for $v$ and $x$ that $v\in V$ which contradicts $v\not\in V$.
	Note that we can apply \cref{lem:remainingVertex} because $u\not\in \ghosts{i}$ and therefore $\max\set{\w_i(y): y\in (N_i(v)\cap \ghosts{i})\setminus N_i(x)}+\w_i(x) \leq \w_i(\N_i(v)) - \w_i(u) \leq \w_i(v)$.

	Next, we consider $v,u\in V$.
	With \cref{lem:weights} and \cref{lem:neighborhood}, we obtain $\w(v)=\w_i(v)\geq \N_i(v)-\w_i(u) \geq \N(u) - \w(u)$.
	Thus, we can exclude the remaining vertices of $X$ in $G$ with \emph{Extended Single-Edge}~\cite{TowardsComputiGuJi2021}.%

	Last, assume $v\in V$ and $u\not\in V$.
	Then, we obtain in $G$ that $\w(v)\geq \N(v)$ with \cref{lem:weights} and \cref{lem:neighborhood}.
	Thus, the vertices of $X$ can be excluded, because we can replace in every independent set of $G$, the vertices of $X$ by $v$ to obtain a solution which has at least \hbox{the same weight}.
	In every case, the local reduced graph is given as $G_i=G-X$, with offset $\solw{G}=\solw{G'}$, and $\I_i=\I_i'$.
\end{proof}

\section{Evaluation}
\subsection{Data Sets}
\Cref{appendix:tab:instances} lists all graphs used in our strong-scaling experiments in our evaluation.
\begin{table*}
	\centering

    \textsf{\begin{tabular}{lrrrrr}
  & \multicolumn{1}{c}{Graph} & \multicolumn{1}{c}{Weights} & \multicolumn{1}{c}{$|V|$} & \multicolumn{1}{c}{$|E|$} & \multicolumn{1}{c}{$d$} \\ \toprule
  \multirow{5}{*}{\textsc{10-th DIMACS}~\cite{DBLP:conf/dimacs/2012, DBLP:conf/ipps/HoltgreweSS10, openstreetmap}} & USA-road-d & uf $[1, 200]$ & \numprint{23947347} & \numprint{28854312} & \numprint{2.41} \\
  &asia.osm & uf $[1, 200]$ & \numprint{11950757} & \numprint{12711603} & \numprint{2.13} \\
  &delaunay\textunderscore n24 & uf $[1, 200]$ & \numprint{16777216} & \numprint{50331601} & \numprint{6.00} \\
  &europe.osm & uf $[1, 200]$ & \numprint{50912018} & \numprint{54054660} & \numprint{2.12} \\
  & germany.osm & uf $[1, 200]$ & \numprint{11548845} & \numprint{12369181} & \numprint{2.14} \\
 \midrule
  \multirow{15}{*}{\makecell{\textsc{Laboratory for}\\ \textsc{Web Algorithmics}~\cite{BoVWFI, BRSLLP}}}   & arabic-2005 & uf $[1, 200]$ & \numprint{22744080} & \numprint{553903073} & \numprint{48.71} \\
  & dewiki-2013 & uf $[1, 200]$ & \numprint{1532354} & \numprint{33093029} & \numprint{43.19} \\
  & enwiki-2013 & uf $[1, 200]$ & \numprint{4206785} & \numprint{91939728} & \numprint{43.71} \\
  & enwiki-2018 & uf $[1, 200]$ & \numprint{5616717} & \numprint{117244295} & \numprint{41.75} \\
  & enwiki-2022 & uf $[1, 200]$ & \numprint{6492490} & \numprint{144588656} & \numprint{44.54} \\
  & frwiki-2013 & uf $[1, 200]$ & \numprint{1352053} & \numprint{31037302} & \numprint{45.91} \\
  & hollywood-2011 & uf $[1, 200]$ & \numprint{2180759} & \numprint{114492816} & \numprint{105.00} \\
  & imdb-2021 & uf $[1, 200]$ & \numprint{2996317} & \numprint{5369472} & \numprint{3.58} \\
  & indochina-2004 & uf $[1, 200]$ & \numprint{7414866} & \numprint{150984819} & \numprint{40.72} \\
  & it-2004 & uf $[1, 200]$ & \numprint{41291594} & \numprint{1027474947} & \numprint{49.77} \\
  & itwiki-2013 & uf $[1, 200]$ & \numprint{1016867} & \numprint{23429644} & \numprint{46.08} \\
  & ljournal-2008 & uf $[1, 200]$ & \numprint{5363260} & \numprint{49514271} & \numprint{18.46} \\
  & sk-2005 & uf $[1, 200]$ & \numprint{50636154} & \numprint{1810063330} & \numprint{71.49} \\
  & twitter-2010 & uf $[1, 200]$ & \numprint{41652230} & \numprint{1202513046} & \numprint{57.74} \\
  & uk-2005 & uf $[1, 200]$ & \numprint{39459925} & \numprint{783027125} & \numprint{39.69} \\ 
  \midrule
  \multirow{1}{*}{\textsc{Mesh}~\cite{DBLP:journals/tog/SanderNCH08}} & buddha & uf & \numprint{1087716} & \numprint{1631574} & \numprint{3.00} \\ \midrule
  \multirow{9}{*}{\textsc{Network Repository}~\cite{nr}} & Flan\textunderscore 1565 & uf $[1, 200]$ & \numprint{1564794} & \numprint{57920625} & \numprint{74.03} \\
  & Geo\textunderscore 1438 & uf $[1, 200]$ & \numprint{1437960} & \numprint{30859365} & \numprint{42.92} \\
  & HV15R & uf $[1, 200]$ & \numprint{2017169} & \numprint{162357569} & \numprint{160.98} \\
  & Hook\textunderscore 1498 & uf $[1, 200]$ & \numprint{1498023} & \numprint{29709711} & \numprint{39.67} \\
  & circuit5M & uf $[1, 200]$ & \numprint{5558326} & \numprint{26983926} & \numprint{9.71} \\
  & dielFilterV3real & uf $[1, 200]$ & \numprint{1102824} & \numprint{44101598} & \numprint{79.98} \\
  & soc-flickr-und & uf $[1, 200]$ & \numprint{1715255} & \numprint{15555041} & \numprint{18.14} \\
  & soc-sinaweibo & uf $[1, 200]$ & \numprint{58655849} & \numprint{261321033} & \numprint{8.91} \\
  & webbase-2001 & uf $[1, 200]$ & \numprint{118142155} & \numprint{854809761} & \numprint{14.47} \\
  \midrule
  \multirow{17}{*}{\textsc{SNAP}~\cite{snapnets}} & com-Friendster & uf $[1, 200]$ & \numprint{65608366} & \numprint{1806067135} & \numprint{55.06} \\
  & as-skitter & uf & \numprint{1696415} & \numprint{11095298} & \numprint{13.08} \\
  & com-youtube & ? & \numprint{1134890} & \numprint{2987624} & \numprint{5.27} \\
  & roadNet-CA & uf & \numprint{1965206} & \numprint{2766607} & \numprint{2.82} \\
  & roadNet-PA & uf & \numprint{1088092} & \numprint{1541898} & \numprint{2.83} \\
  & roadNet-PA & ? & \numprint{1088092} & \numprint{1541898} & \numprint{2.83} \\
  & roadNet-TX & uf & \numprint{1379917} & \numprint{1921660} & \numprint{2.79} \\
  & soc-LiveJournal1 & uf & \numprint{4847571} & \numprint{42851237} & \numprint{17.68} \\
  & soc-pokec-relationships & uf & \numprint{1632803} & \numprint{22301964} & \numprint{27.32} \\
  & wiki-Talk & uf & \numprint{2394385} & \numprint{4659565} & \numprint{3.89} \\
  & wiki-topcats & uf $[1, 200]$ & \numprint{1791489} & \numprint{25444207} & \numprint{28.41} \\
  & Bump\textunderscore 2911 & uf $[1, 200]$ & \numprint{2911419} & \numprint{62409240} & \numprint{42.87} \\
  & Cube\textunderscore Coup\textunderscore dt6 & uf $[1, 200]$ & \numprint{2164760} & \numprint{62520692} & \numprint{57.76} \\
  & Long\textunderscore Coup\textunderscore dt6 & uf $[1, 200]$ & \numprint{1470152} & \numprint{42809420} & \numprint{58.24} \\
  & ML\textunderscore Geer & uf $[1, 200]$ & \numprint{1504002} & \numprint{54687985} & \numprint{72.72} \\
  & Queen\textunderscore 4147 & uf $[1, 200]$ & \numprint{4147110} & \numprint{162676087} & \numprint{78.45} \\
  & Serena & uf $[1, 200]$ & \numprint{1391349} & \numprint{31570176} & \numprint{45.38}

\end{tabular}
}

  \caption{This table shows the meta information for all the graphs in our strong-scaling experiments. Note that `uf' (together with the range $[0, 200]$) means that the weights were chosen uniform at random (from this range), `?' indicates that we do not know the distribution. Most of the graphs are road networks, web-crawling graphs, social networks, and graphs that model structural problems.}
	\label{appendix:tab:instances}
\end{table*}

\subsection{Further Evaluation Data}
In \Cref{appendix:fig:kernelsize-runningtime-reduction}, we give the kernels' relative vertex size, the relative difference between the kernel sizes on one and more cores, and the running time of our reduction algorithms.
The top row shows the results for our asynchronous variant \KaDisReduA{}, the bottom row the results for \KaDisReduS{}.
While the reduction impact is almost the same, we see that the (median) scaling behavior of \KaDisReduA{} is better than that of \KaDisReduS{}.

For graph-specific results of {\KaDisReduA} and {\KaDisReduS} for 1, 64, and \numprint{1024} cores, see \Cref{tab:reduce-detailed-DisReduA} and \Cref{tab:reduce-detailed-DisReduS}, respectively.
The running times of both reduction algorithms for each instance are compared in \Cref{tab:reduce-detailed-red-t}.

In \Cref{tab:detailed-aRnP} and \Cref{tab:detailed-sRnP}, we compare our distributed reduce-and-peel solvers against {\htwis}.
For one core, {\sRnP} and {\aRnP} had for two instances a timeout.
Six instances were too large to solve them with {\htwis} because the implementation supports only a 32 bit integer representation.
Furthermore, for another two graphs, {\htwis} found wrong solution weights.
For two graphs, \textsf{hollywood-2011} and \textsf{imdb-2021}, the considered algorithms found solutions but we excluded them from the overall comparison because the solution by {\htwis} were far off from ours.

\begin{figure*}[t]
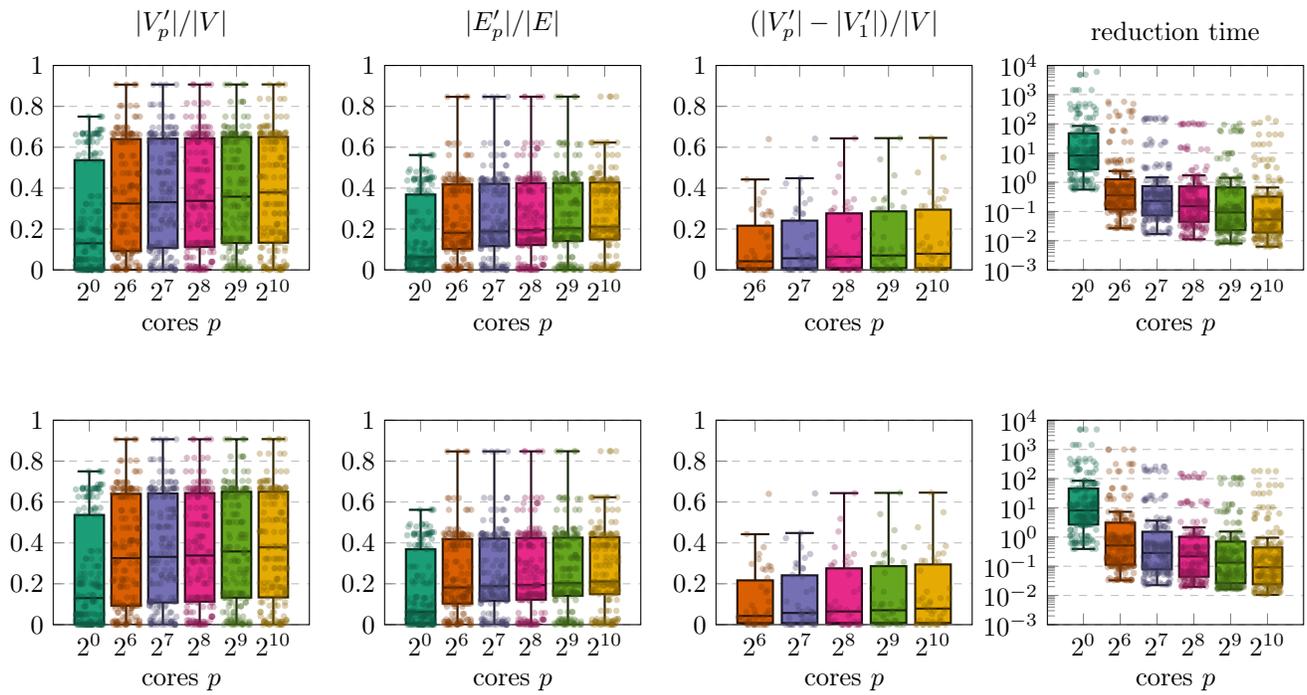

	\centering
	\begin{tikzpicture}
		\begin{groupplot}[group style={
						group size=4 by 2,
						horizontal sep=1cm,
						vertical sep=2cm,
						ylabels at=edge left,
					},
				width=0.33\linewidth,
				legend style={draw=none},
				xtick={1, 2, 3, 4, 5, 6},
				xticklabels={$2^0$,$2^6$,$2^7$,$2^8$,$2^9$,$2^{10}$},
				ytick pos=left,
				xlabel={cores $p$},
				grid style=dashed,
			]
			\nextgroupplot[title=$|V_p'|/|V|$,
				ymajorgrids,
				ymode=normal,
				ymax={1.0},
				ymin={1.5343738354209146e-07},
				boxplot/draw direction=y,
				xlabel={cores $p$},
				ylabel={},
			]
			\input{plots/boxplot-rel-kernel-vertices-aRG/plain_plot.tex}
			\nextgroupplot[title=$|E_p'|/|E|$,
				ymajorgrids,
				ymode=normal,
				ymax={1.0},
				ymin={1.5343738354209146e-07},
				boxplot/draw direction=y,
				xlabel={cores $p$},
				ylabel={},
			]
			\input{plots/boxplot-rel-kernel-edges-aRG/plain_plot.tex}
			\nextgroupplot[title=$(|V_p'|-|V_1'|)/|V|$,
				ymajorgrids,
				ymode=normal,
				ymax={1.0},
				ymin={0},
				xtick={1, 2, 3, 4, 5},
				xticklabels={$2^6$,$2^7$,$2^8$,$2^9$,$2^{10}$},
				boxplot/draw direction=y,
				xlabel={cores $p$},
				ylabel={},
			]
			\input{plots/boxplot-rel-kernel-vertices-change-aRG/plain_plot.tex}
			\nextgroupplot[title=reduction time,
				ymajorgrids,
				ymode=log,
				ymax={10000.0},
				ymin={0.001},
				ytick={0.001, 0.01, 0.1, 1, 10, 100, 1000, 10000},
				boxplot/draw direction=y,
				xlabel={cores $p$},
				ylabel={},
			]
			\input{plots/boxplot-reduction-time-aRG/plain_plot.tex}
			\nextgroupplot[title={},
				ymajorgrids,
				ymode=normal,
				ymax={1.0},
				ymin={1.5343738354209146e-07},
				boxplot/draw direction=y,
				xlabel={cores $p$},
				ylabel={},
			]
			\input{plots/boxplot-rel-kernel-vertices-RG/plain_plot.tex}
                        \nextgroupplot[title={},
				ymajorgrids,
				ymode=normal,
				ymax={1.0},
				ymin={1.5343738354209146e-07},
				boxplot/draw direction=y,
				xlabel={cores $p$},
				ylabel={},
			]
			\input{plots/boxplot-rel-kernel-edges-RG/plain_plot.tex}
			\nextgroupplot[title={},
				ymajorgrids,
				ymode=normal,
				ymax={1.0},
				ymin={0},
				xtick={1, 2, 3, 4, 5},
				xticklabels={$2^6$,$2^7$,$2^8$,$2^9$,$2^{10}$},
				boxplot/draw direction=y,
				xlabel={cores $p$},
				ylabel={},
			]
			\input{plots/boxplot-rel-kernel-vertices-change-RG/plain_plot.tex}
			
			\nextgroupplot[title={},
				ymajorgrids,
				ymode=log,
				ymax={10000.0},
				ymin={0.001},
				ytick={0.001, 0.01, 0.1, 1, 10, 100, 1000, 10000},
				boxplot/draw direction=y,
				xlabel={cores $p$},
				ylabel={},
			]
			\input{plots/boxplot-reduction-time-RG/plain_plot.tex}

		\end{groupplot}
	\end{tikzpicture}

	\caption{Reduction impact and running time (in seconds) for our asynchronous reduction algorithm {\KaDisReduA} (top row) and our synchronous reduction algorithm {\KaDisReduS} (bottom row) in the strong-scaling experiment.}
	\label{appendix:fig:kernelsize-runningtime-reduction}
\end{figure*}

\begin{table*}
	\centering
	\begin{tabular}{rrrrrrrrrrrrrrrrrrr}
		\multicolumn{1}{r}{} &  & \multicolumn{4}{c}{$|E'|/|E|$} &                        & \multicolumn{4}{c}{$|V'|/|V|$} &                          & \multicolumn{4}{c}{$t_{\text{reduce}}$~[s]}                                                                                                                                                                                                                \\
		\multicolumn{1}{r}{} &  & \multicolumn{1}{c}{1}          & \multicolumn{1}{c}{64} & \multicolumn{1}{c}{256}        & \multicolumn{1}{c}{\numprint{1024}} &                                             & \multicolumn{1}{c}{1} & \multicolumn{1}{c}{64} & \multicolumn{1}{c}{256} & \multicolumn{1}{c}{\numprint{1024}} &  & \multicolumn{1}{c}{1} & \multicolumn{1}{c}{64} & \multicolumn{1}{c}{256} & \multicolumn{1}{c}{\numprint{1024}} \\
		\cmidrule{1-1}  \cmidrule{3-6}  \cmidrule{8-11}  \cmidrule{13-16}
		\KaDisReduS{} & &\textbf{\numprint{0.06}} & \textbf{\numprint{0.18}} & \textbf{\numprint{0.19}} & \textbf{\numprint{0.21}} & &\textbf{\numprint{0.13}} & \numprint{0.33} & \textbf{\numprint{0.34}} & \textbf{\numprint{0.38}} & &\textbf{\numprint{11.27}} & \numprint{0.73} & \numprint{0.29} & \numprint{0.15}\\
\KaDisReduA{} & &\textbf{\numprint{0.06}} & \textbf{\numprint{0.18}} & \textbf{\numprint{0.19}} & \textbf{\numprint{0.21}} & &\textbf{\numprint{0.13}} & \textbf{\numprint{0.32}} & \textbf{\numprint{0.34}} & \textbf{\numprint{0.38}} & &\numprint{12.26} & \textbf{\numprint{0.59}} & \textbf{\numprint{0.26}} & \textbf{\numprint{0.10}}\\

	\end{tabular}
	\caption{Summarized results for the reduction impact and time on 1, 64, 256 and \numprint{1024} cores (strong-scaling experiments). The table shows the number of vertices and edges in the reduced graph relative to the input graph on median over all instances. The right column shows the geometric mean running time of the reduce phase.}\label{appendix:tab:redImpact}
\end{table*}

\begin{table*}
	\centering
	\begin{tabular}{rrrrrrrrrrrrrrr}

		\multicolumn{1}{r}{} &  & \multicolumn{2}{c}{1 core}     &                             & \multicolumn{2}{c}{\numprint{1024} cores} &                                & \multicolumn{3}{c}{\numprint{1024} cores (partitioned)}                                                                                                      \\
		\multicolumn{1}{r}{} &  & \multicolumn{1}{c}{$|V'|/|V|$} & \multicolumn{1}{c}{$t$~[s]} &                                           & \multicolumn{1}{c}{$|V'|/|V|$} & \multicolumn{1}{c}{$t$~[s]}                             &  & \multicolumn{1}{c}{$\Delta cut$} & \multicolumn{1}{c}{$|V'|/|V|$} & \multicolumn{1}{c}{$t$~[s]} \\
		\cmidrule{1-1}  \cmidrule{3-4}  \cmidrule{6-7}  \cmidrule{9-11}
		\KaDisReduA{} & &\numprint{0.14} & \numprint{15.59} & &\numprint{0.38} & \numprint{0.20} & &\numprint{-0.29} & \numprint{0.25} &\numprint{3.69}\\

	\end{tabular}
	\caption{The table shows the effect on the reduction impact of additional graph partitioning for {\KaDisReduA} on \numprint{1024} cores for our strong-scaling experiments. It shows the median of the ratios of the |$V'$| in the reduced graph relative to the number of vertices in the input graph $|V|$, the geometric mean running time (including the time of the partitioning phase), and the median cut improvement relative to number of edges in the input graph.}\label{appendix:tab:part}
\end{table*}

\begin{table*}
  \centering
	\begin{tabular}{rrrrr}
		\multicolumn{1}{r}{} &  & \multicolumn{2}{c}{$|V'|/|V|$} 	\\
    \multicolumn{1}{r}{} &  & \multicolumn{1}{c}{\sRG} & \multicolumn{1}{c}{\aRG} \\
    \cmidrule{1-1}  \cmidrule{3-4}

		\myGNM & &- & \textbf{\numprint{0.9815}}\\ 
\myRGG & &\numprint{0.3382} & \textbf{\numprint{0.3382}}\\ 
\myRHG & &\textbf{\numprint{0.0001}} & \numprint{0.0001}\\ 
\myRHGBig & &\textbf{\numprint{0.0022}} & \numprint{0.0022}\\ 

	\end{tabular}
  \caption{Weak scaling results of the reduced graph size in terms of vertices $|V'|$ relative to the number of vertices $|V|$ in the input graph for \numprint{1024} cores.}
  \label{appendix:tab:redImpactWeak}
\end{table*}

\begin{table*}
	\begin{tabular}{rrrrrrrrrrrrrrr}
		\multicolumn{1}{r}{} &  & \multicolumn{6}{c}{$\w(\I)/\w(\I_{best})$} &                          & \multicolumn{6}{c}{$t$~[s]}                                                                                                                                                                                                                                                       \\
		\multicolumn{1}{r}{} &  & \multicolumn{1}{c}{\sG}                    & \multicolumn{1}{c}{\sRG} & \multicolumn{1}{c}{\sRnP}   & \multicolumn{1}{c}{\aG} & \multicolumn{1}{c}{\aRG} & \multicolumn{1}{c}{\aRnP} &  & \multicolumn{1}{c}{\sG} & \multicolumn{1}{c}{\sRG} & \multicolumn{1}{c}{\sRnP} & \multicolumn{1}{c}{\aG} & \multicolumn{1}{c}{\aRG} & \multicolumn{1}{c}{\aRnP} \\
		\cmidrule{1-1}  \cmidrule{3-8}  \cmidrule{10-15}

		\myGNM & &\numprint{0.89} & - & \numprint{0.99} & \numprint{0.89} & \numprint{0.89} & \textbf{\numprint{1.00}} & &\numprint{17.78} & - & \numprint{29.56} & \textbf{\numprint{14.09}} & \numprint{30.79} & \numprint{30.28}\\ 
\myRGG & &\numprint{0.91} & \numprint{0.96} & \numprint{1.00} & \numprint{0.91} & \numprint{0.96} & \textbf{\numprint{1.00}} & &\textbf{\numprint{0.70}} & \numprint{6.04} & \numprint{11.32} & \numprint{0.71} & \numprint{6.41} & \numprint{11.50}\\ 
\myRHG & &\numprint{0.96} & \numprint{1.00} & \numprint{1.00} & \numprint{0.96} & \numprint{1.00} & \textbf{\numprint{1.00}} & &\textbf{\numprint{0.62}} & \numprint{1.29} & \numprint{1.63} & \numprint{0.81} & \numprint{1.45} & \numprint{1.52}\\ 
\myRHGBig & &\numprint{0.85} & \numprint{1.00} & \numprint{1.00} & \numprint{0.85} & \numprint{1.00} & \textbf{\numprint{1.00}} & &\textbf{\numprint{1.81}} & \numprint{6.32} & \numprint{6.51} & \numprint{7.65} & \numprint{6.71} & \numprint{7.94}\\ 

	\end{tabular}
  \caption{Weak scaling results for \numprint{1024} cores: Solution quality (left) and the running time (right). Note that the solution quality and running time were rounded to two decimal places and are highlighted in bold if the result is the best overall before rounding.}
	\label{appendix:tab:weak-scaling-results}
\end{table*}

\begin{table*}
  \caption{Reduction results of {\KaDisReduA} for \numprint{1}, \numprint{64}, and \numprint{1024} (from left to right): reduction time $t_{\mathrm{reduce}}$ in seconds, self speed-up, and the reduction ratios in terms of vertices and edges.
    \textsf{roadNet-PA*} refers to the uniformly weighted graph.
  }\label{tab:reduce-detailed-DisReduA}
  \begin{center}
    \small{\textsf{
    \begin{tabular}{rrrrrrrrrrrrrrrr}
      \mc{1}{r}{} && \mc{3}{c}{$t_{\text{reduce}}$~[s]} && \mc{2}{c}{speed-up} && \mc{3}{c}{$|V'|/|V|$~[\%]} && \mc{3}{c}{$|E'|/|E|$~[\%]} \\ 
      \mc{1}{r}{} && \mc{1}{c}{1} & \mc{1}{c}{64} & \mc{1}{c}{1024} && \mc{1}{c}{64} & \mc{1}{c}{1024} && \mc{1}{c}{1} & \mc{1}{c}{64} & \mc{1}{c}{1024} && \mc{1}{c}{1} & \mc{1}{c}{64} & \mc{1}{c}{1024} \\ 
      \cmidrule{1-1}  \cmidrule{3-5}  \cmidrule{7-8}  \cmidrule{10-12}  \cmidrule{14-16} 
      Bump\textunderscore 2911 & &\numprint{5.59} & \numprint{0.19} & \numprint{0.02} & &\numprint{30.0} & \numprint{281.7} & &\numprint{66.9} & \numprint{66.9} & \numprint{66.9} & &\numprint{44.7} & \numprint{44.7} & \numprint{44.7}\\ 
Cube\textunderscore Coup\textunderscore . & &\numprint{5.26} & \numprint{0.15} & \numprint{0.02} & &\numprint{35.0} & \numprint{250.8} & &\numprint{62.5} & \numprint{62.5} & \numprint{62.6} & &\numprint{39.0} & \numprint{39.0} & \numprint{39.0}\\ 
Flan\textunderscore 1565 & &\numprint{4.61} & \numprint{0.13} & \numprint{0.01} & &\numprint{36.6} & \numprint{330.0} & &\numprint{64.8} & \numprint{64.9} & \numprint{65.0} & &\numprint{41.9} & \numprint{42.0} & \numprint{42.1}\\ 
Geo\textunderscore 1438 & &\numprint{2.76} & \numprint{0.11} & \numprint{0.02} & &\numprint{26.0} & \numprint{182.9} & &\numprint{66.7} & \numprint{66.7} & \numprint{66.7} & &\numprint{44.3} & \numprint{44.4} & \numprint{44.4}\\ 
HV15R & &\numprint{14.57} & \numprint{0.38} & \numprint{0.04} & &\numprint{38.5} & \numprint{342.1} & &\numprint{68.6} & \numprint{69.7} & \numprint{70.5} & &\numprint{45.2} & \numprint{46.7} & \numprint{47.3}\\ 
Hook\textunderscore 1498 & &\numprint{3.46} & \numprint{0.09} & \numprint{0.02} & &\numprint{38.9} & \numprint{189.0} & &\numprint{57.1} & \numprint{66.5} & \numprint{66.6} & &\numprint{36.9} & \numprint{44.0} & \numprint{44.2}\\ 
Long\textunderscore Coup\textunderscore . & &\numprint{3.63} & \numprint{0.11} & \numprint{0.01} & &\numprint{32.5} & \numprint{256.5} & &\numprint{62.6} & \numprint{62.6} & \numprint{62.6} & &\numprint{39.0} & \numprint{39.1} & \numprint{39.1}\\ 
ML\textunderscore Geer & &\numprint{4.75} & \numprint{0.11} & \numprint{0.02} & &\numprint{43.9} & \numprint{289.5} & &\numprint{75.0} & \numprint{75.0} & \numprint{75.0} & &\numprint{56.2} & \numprint{56.2} & \numprint{56.2}\\ 
Queen\textunderscore 4147 & &\numprint{12.99} & \numprint{0.34} & \numprint{0.04} & &\numprint{37.8} & \numprint{309.4} & &\numprint{66.2} & \numprint{66.2} & \numprint{66.3} & &\numprint{43.7} & \numprint{43.8} & \numprint{43.9}\\ 
Serena & &\numprint{2.79} & \numprint{0.10} & \numprint{0.02} & &\numprint{27.5} & \numprint{165.7} & &\numprint{66.7} & \numprint{66.7} & \numprint{66.7} & &\numprint{44.3} & \numprint{44.4} & \numprint{44.4}\\ 
USA-road-d & &\numprint{15.66} & \numprint{0.82} & \numprint{0.06} & &\numprint{19.2} & \numprint{256.1} & &\numprint{2.0} & \numprint{13.3} & \numprint{13.3} & &\numprint{2.8} & \numprint{14.9} & \numprint{14.9}\\ 
arabic-2005 & &\numprint{385.78} & \numprint{24.01} & \numprint{5.92} & &\numprint{16.1} & \numprint{65.2} & &\numprint{25.0} & \numprint{29.9} & \numprint{32.2} & &\numprint{21.3} & \numprint{24.9} & \numprint{28.9}\\ 
as-skitter & &\numprint{1.99} & \numprint{0.19} & \numprint{0.04} & &\numprint{10.7} & \numprint{54.6} & &\numprint{4.6} & \numprint{21.4} & \numprint{22.4} & &\numprint{2.6} & \numprint{11.4} & \numprint{12.1}\\ 
asia.osm & &\numprint{5.57} & \numprint{0.17} & \numprint{0.03} & &\numprint{31.8} & \numprint{185.9} & &\numprint{0.4} & \numprint{2.1} & \numprint{2.2} & &\numprint{0.6} & \numprint{2.6} & \numprint{2.8}\\ 
buddha & &\numprint{1.13} & \numprint{0.03} & \numprint{0.01} & &\numprint{39.8} & \numprint{121.5} & &\numprint{12.3} & \numprint{40.7} & \numprint{67.3} & &\numprint{13.0} & \numprint{38.7} & \numprint{62.3}\\ 
circuit5M & &\numprint{6.47} & \numprint{0.57} & \numprint{0.42} & &\numprint{11.3} & \numprint{15.5} & &\numprint{0.0} & \numprint{38.0} & \numprint{38.0} & &\numprint{0.5} & \numprint{36.8} & \numprint{38.3}\\ 
com-Friend. & &\numprint{1418.00} & - & \numprint{6.80} & &- & \numprint{208.5} & &\numprint{16.9} & - & \numprint{42.4} & &\numprint{6.6} & - & \numprint{22.7}\\ 
com-youtube & &\numprint{0.56} & \numprint{0.08} & \numprint{0.02} & &\numprint{7.3} & \numprint{28.7} & &\numprint{0.0} & \numprint{0.7} & \numprint{0.7} & &\numprint{0.0} & \numprint{0.3} & \numprint{0.4}\\ 
delaunay\textunderscore n. & &\numprint{84.92} & \numprint{0.71} & \numprint{0.06} & &\numprint{119.6} & \numprint{1389.8} & &\numprint{56.1} & \numprint{90.6} & \numprint{90.7} & &\numprint{49.2} & \numprint{84.7} & \numprint{84.8}\\ 
dewiki-2013 & &\numprint{24.42} & \numprint{0.71} & \numprint{0.09} & &\numprint{34.2} & \numprint{274.8} & &\numprint{36.0} & \numprint{80.3} & \numprint{83.7} & &\numprint{10.0} & \numprint{42.2} & \numprint{50.2}\\ 
dielFilter. & &\numprint{2.79} & \numprint{0.19} & \numprint{0.02} & &\numprint{14.7} & \numprint{113.7} & &\numprint{0.0} & \numprint{64.0} & \numprint{64.6} & &\numprint{0.0} & \numprint{41.8} & \numprint{42.8}\\ 
enwiki-2013 & &\numprint{40.44} & \numprint{1.26} & \numprint{0.24} & &\numprint{32.0} & \numprint{169.5} & &\numprint{20.3} & \numprint{46.9} & \numprint{51.5} & &\numprint{5.9} & \numprint{16.3} & \numprint{20.1}\\ 
enwiki-2018 & &\numprint{47.42} & \numprint{1.71} & \numprint{0.33} & &\numprint{27.7} & \numprint{143.3} & &\numprint{14.5} & \numprint{41.2} & \numprint{45.8} & &\numprint{4.3} & \numprint{13.6} & \numprint{16.9}\\ 
enwiki-2022 & &\numprint{68.87} & \numprint{2.25} & \numprint{0.35} & &\numprint{30.6} & \numprint{196.7} & &\numprint{18.5} & \numprint{47.3} & \numprint{51.0} & &\numprint{5.7} & \numprint{18.0} & \numprint{21.1}\\ 
europe.osm & &\numprint{24.92} & \numprint{0.81} & \numprint{0.08} & &\numprint{30.6} & \numprint{309.8} & &\numprint{0.1} & \numprint{1.3} & \numprint{1.5} & &\numprint{0.2} & \numprint{1.5} & \numprint{1.7}\\ 
frwiki-2013 & &\numprint{13.10} & \numprint{0.45} & \numprint{0.11} & &\numprint{29.0} & \numprint{121.5} & &\numprint{28.6} & \numprint{55.8} & \numprint{59.1} & &\numprint{9.1} & \numprint{21.5} & \numprint{24.4}\\ 
germany.osm & &\numprint{5.70} & \numprint{0.20} & \numprint{0.02} & &\numprint{28.8} & \numprint{246.1} & &\numprint{0.1} & \numprint{1.3} & \numprint{1.5} & &\numprint{0.1} & \numprint{1.4} & \numprint{1.6}\\ 
hollywood-. & &\numprint{67.71} & \numprint{2.37} & \numprint{0.41} & &\numprint{28.6} & \numprint{164.1} & &\numprint{0.0} & \numprint{18.5} & \numprint{38.5} & &\numprint{0.0} & \numprint{12.9} & \numprint{39.0}\\ 
imdb-2021 & &\numprint{1.29} & \numprint{0.24} & \numprint{0.03} & &\numprint{5.3} & \numprint{42.1} & &\numprint{0.0} & \numprint{0.1} & \numprint{0.6} & &\numprint{0.0} & \numprint{0.1} & \numprint{1.0}\\ 
indochina-. & &\numprint{161.16} & \numprint{4.68} & \numprint{1.13} & &\numprint{34.4} & \numprint{143.1} & &\numprint{19.6} & \numprint{22.3} & \numprint{27.6} & &\numprint{20.2} & \numprint{22.2} & \numprint{30.9}\\ 
it-2004 & &\numprint{1433.77} & \numprint{183.58} & \numprint{43.29} & &\numprint{7.8} & \numprint{33.1} & &\numprint{22.5} & \numprint{26.5} & \numprint{27.6} & &\numprint{28.4} & \numprint{30.8} & \numprint{33.1}\\ 
itwiki-2013 & &\numprint{9.13} & \numprint{0.42} & \numprint{0.07} & &\numprint{22.0} & \numprint{135.8} & &\numprint{19.9} & \numprint{51.9} & \numprint{54.9} & &\numprint{5.1} & \numprint{18.3} & \numprint{21.1}\\ 
ljournal-2. & &\numprint{10.09} & \numprint{0.66} & \numprint{0.09} & &\numprint{15.2} & \numprint{110.4} & &\numprint{1.5} & \numprint{8.3} & \numprint{9.3} & &\numprint{3.2} & \numprint{5.3} & \numprint{6.8}\\ 
roadNet-CA & &\numprint{1.26} & \numprint{0.05} & \numprint{0.01} & &\numprint{26.0} & \numprint{144.3} & &\numprint{5.0} & \numprint{9.4} & \numprint{16.8} & &\numprint{6.0} & \numprint{10.3} & \numprint{17.5}\\ 
roadNet-PA & &\numprint{0.68} & \numprint{0.03} & \numprint{0.01} & &\numprint{24.8} & \numprint{86.6} & &\numprint{5.3} & \numprint{9.3} & \numprint{18.3} & &\numprint{6.4} & \numprint{10.3} & \numprint{19.1}\\ 
roadNet-PA* & &\numprint{0.68} & \numprint{0.03} & \numprint{0.01} & &\numprint{23.7} & \numprint{91.6} & &\numprint{5.4} & \numprint{9.2} & \numprint{18.3} & &\numprint{6.5} & \numprint{10.2} & \numprint{19.1}\\ 
roadNet-TX & &\numprint{0.83} & \numprint{0.03} & \numprint{0.01} & &\numprint{25.4} & \numprint{86.5} & &\numprint{4.8} & \numprint{7.8} & \numprint{15.2} & &\numprint{5.9} & \numprint{8.8} & \numprint{16.1}\\ 
sk-2005 & &\numprint{5130.96} & \numprint{526.40} & \numprint{123.59} & &\numprint{9.7} & \numprint{41.5} & &\numprint{32.3} & \numprint{39.4} & \numprint{45.3} & &\numprint{23.9} & \numprint{28.3} & \numprint{36.9}\\ 
soc-LiveJo. & &\numprint{10.18} & \numprint{0.83} & \numprint{0.10} & &\numprint{12.3} & \numprint{97.6} & &\numprint{1.4} & \numprint{31.0} & \numprint{32.2} & &\numprint{2.2} & \numprint{18.4} & \numprint{19.6}\\ 
soc-flickr. & &\numprint{2.71} & \numprint{0.29} & \numprint{0.05} & &\numprint{9.2} & \numprint{50.7} & &\numprint{0.2} & \numprint{3.7} & \numprint{4.1} & &\numprint{0.1} & \numprint{2.3} & \numprint{2.5}\\ 
soc-pokec-. & &\numprint{12.61} & \numprint{0.43} & \numprint{0.05} & &\numprint{29.1} & \numprint{236.7} & &\numprint{53.7} & \numprint{69.5} & \numprint{69.8} & &\numprint{47.9} & \numprint{61.9} & \numprint{62.2}\\ 
soc-sinawe. & &\numprint{57.39} & \numprint{5.79} & \numprint{1.28} & &\numprint{9.9} & \numprint{45.0} & &\numprint{0.0} & \numprint{0.0} & \numprint{0.0} & &\numprint{0.0} & \numprint{0.0} & \numprint{0.0}\\ 
twitter-20. & &\numprint{479.26} & \numprint{58.08} & \numprint{4.61} & &\numprint{8.3} & \numprint{104.0} & &\numprint{0.4} & \numprint{5.9} & \numprint{6.2} & &\numprint{0.1} & \numprint{0.6} & \numprint{0.6}\\ 
uk-2005 & &\numprint{313.67} & \numprint{270.41} & \numprint{20.02} & &\numprint{1.2} & \numprint{15.7} & &\numprint{13.7} & \numprint{19.8} & \numprint{20.5} & &\numprint{10.8} & \numprint{17.2} & \numprint{20.8}\\ 
webbase-20. & &\numprint{222.49} & \numprint{5.42} & \numprint{1.46} & &\numprint{41.1} & \numprint{152.1} & &\numprint{8.9} & \numprint{10.8} & \numprint{10.9} & &\numprint{13.5} & \numprint{15.0} & \numprint{15.6}\\ 
wiki-Talk & &\numprint{1.10} & \numprint{0.12} & \numprint{0.04} & &\numprint{9.5} & \numprint{30.4} & &\numprint{0.0} & \numprint{0.0} & \numprint{0.0} & &\numprint{0.0} & \numprint{0.0} & \numprint{0.0}\\ 
wiki-topca. & &\numprint{7.43} & \numprint{0.31} & \numprint{0.09} & &\numprint{24.2} & \numprint{84.0} & &\numprint{11.2} & \numprint{33.9} & \numprint{37.8} & &\numprint{4.2} & \numprint{13.4} & \numprint{15.5}\\ 

    \end{tabular}
  }}
  \end{center}
\end{table*}

\begin{table*}
  \caption{Reduction results of {\KaDisReduS} for \numprint{1}, \numprint{64}, and \numprint{1024} (from left to right): reduction time $t_{\mathrm{reduce}}$ in seconds, self speed-up, and the reduction ratios in terms of vertices and edges.
    \textsf{roadNet-PA*} refers to the uniformly weighted graph.
  }\label{tab:reduce-detailed-DisReduS}
  \begin{center}
    \small{\textsf{
    \begin{tabular}{rrrrrrrrrrrrrrrr}
      \mc{1}{r}{} && \mc{3}{c}{$t_{\text{reduce}}$~[s]} && \mc{2}{c}{speed-up} && \mc{3}{c}{$|V'|/|V|$~[\%]} && \mc{3}{c}{$|E'|/|E|$~[\%]} \\ 
      \mc{1}{r}{} && \mc{1}{c}{1} & \mc{1}{c}{64} & \mc{1}{c}{1024} && \mc{1}{c}{64} & \mc{1}{c}{1024} && \mc{1}{c}{1} & \mc{1}{c}{64} & \mc{1}{c}{1024} && \mc{1}{c}{1} & \mc{1}{c}{64} & \mc{1}{c}{1024} \\ 
      \cmidrule{1-1}  \cmidrule{3-5}  \cmidrule{7-8}  \cmidrule{10-12}  \cmidrule{14-16} 
      Bump\textunderscore 2911 & &\numprint{5.33} & \numprint{0.17} & \numprint{0.02} & &\numprint{31.8} & \numprint{253.5} & &\numprint{66.9} & \numprint{66.9} & \numprint{66.9} & &\numprint{44.7} & \numprint{44.7} & \numprint{44.7}\\ 
Cube\textunderscore Coup\textunderscore . & &\numprint{5.10} & \numprint{0.16} & \numprint{0.02} & &\numprint{31.9} & \numprint{272.2} & &\numprint{62.5} & \numprint{62.5} & \numprint{62.6} & &\numprint{39.0} & \numprint{39.0} & \numprint{39.0}\\ 
Flan\textunderscore 1565 & &\numprint{4.42} & \numprint{0.13} & \numprint{0.02} & &\numprint{33.5} & \numprint{255.5} & &\numprint{64.8} & \numprint{64.9} & \numprint{65.0} & &\numprint{41.9} & \numprint{42.0} & \numprint{42.1}\\ 
Geo\textunderscore 1438 & &\numprint{2.64} & \numprint{0.09} & \numprint{0.01} & &\numprint{28.7} & \numprint{211.6} & &\numprint{66.7} & \numprint{66.7} & \numprint{66.7} & &\numprint{44.3} & \numprint{44.4} & \numprint{44.4}\\ 
HV15R & &\numprint{14.40} & \numprint{0.44} & \numprint{0.05} & &\numprint{32.5} & \numprint{298.6} & &\numprint{68.6} & \numprint{69.7} & \numprint{70.5} & &\numprint{45.2} & \numprint{46.7} & \numprint{47.3}\\ 
Hook\textunderscore 1498 & &\numprint{3.32} & \numprint{0.09} & \numprint{0.01} & &\numprint{36.8} & \numprint{256.3} & &\numprint{57.1} & \numprint{66.5} & \numprint{66.6} & &\numprint{36.9} & \numprint{44.0} & \numprint{44.2}\\ 
Long\textunderscore Coup\textunderscore . & &\numprint{3.46} & \numprint{0.11} & \numprint{0.01} & &\numprint{30.6} & \numprint{274.9} & &\numprint{62.6} & \numprint{62.6} & \numprint{62.6} & &\numprint{39.0} & \numprint{39.1} & \numprint{39.1}\\ 
ML\textunderscore Geer & &\numprint{4.71} & \numprint{0.10} & \numprint{0.01} & &\numprint{45.3} & \numprint{340.8} & &\numprint{75.0} & \numprint{75.0} & \numprint{75.0} & &\numprint{56.2} & \numprint{56.2} & \numprint{56.2}\\ 
Queen\textunderscore 4147 & &\numprint{12.57} & \numprint{0.38} & \numprint{0.03} & &\numprint{32.7} & \numprint{428.7} & &\numprint{66.2} & \numprint{66.2} & \numprint{66.3} & &\numprint{43.7} & \numprint{43.8} & \numprint{43.9}\\ 
Serena & &\numprint{2.67} & \numprint{0.10} & \numprint{0.01} & &\numprint{27.6} & \numprint{196.0} & &\numprint{66.7} & \numprint{66.7} & \numprint{66.7} & &\numprint{44.3} & \numprint{44.4} & \numprint{44.4}\\ 
USA-road-d & &\numprint{13.54} & \numprint{0.98} & \numprint{0.09} & &\numprint{13.8} & \numprint{149.0} & &\numprint{2.0} & \numprint{13.3} & \numprint{13.3} & &\numprint{2.8} & \numprint{14.9} & \numprint{15.0}\\ 
arabic-2005 & &\numprint{383.80} & \numprint{26.47} & \numprint{10.43} & &\numprint{14.5} & \numprint{36.8} & &\numprint{25.0} & \numprint{30.0} & \numprint{32.2} & &\numprint{21.3} & \numprint{24.9} & \numprint{28.9}\\ 
as-skitter & &\numprint{1.77} & \numprint{0.22} & \numprint{0.10} & &\numprint{8.2} & \numprint{18.3} & &\numprint{4.6} & \numprint{21.5} & \numprint{22.4} & &\numprint{2.6} & \numprint{11.5} & \numprint{12.1}\\ 
asia.osm & &\numprint{4.55} & \numprint{0.21} & \numprint{0.05} & &\numprint{22.2} & \numprint{84.2} & &\numprint{0.4} & \numprint{2.1} & \numprint{2.2} & &\numprint{0.6} & \numprint{2.6} & \numprint{2.8}\\ 
buddha & &\numprint{1.07} & \numprint{0.03} & \numprint{0.02} & &\numprint{30.8} & \numprint{49.0} & &\numprint{12.3} & \numprint{40.7} & \numprint{67.3} & &\numprint{13.0} & \numprint{38.7} & \numprint{62.3}\\ 
circuit5M & &\numprint{5.45} & \numprint{0.67} & \numprint{0.37} & &\numprint{8.2} & \numprint{14.8} & &\numprint{0.0} & \numprint{38.0} & \numprint{38.0} & &\numprint{0.5} & \numprint{36.8} & \numprint{38.3}\\ 
com-Friend. & &\numprint{1414.62} & - & \numprint{7.26} & &- & \numprint{194.9} & &\numprint{16.9} & - & \numprint{42.4} & &\numprint{6.6} & - & \numprint{22.7}\\ 
com-youtube & &\numprint{0.39} & \numprint{0.08} & \numprint{0.04} & &\numprint{4.9} & \numprint{8.9} & &\numprint{0.0} & \numprint{0.7} & \numprint{0.7} & &\numprint{0.0} & \numprint{0.3} & \numprint{0.4}\\ 
delaunay\textunderscore n. & &\numprint{84.48} & \numprint{0.81} & \numprint{0.08} & &\numprint{104.3} & \numprint{1106.8} & &\numprint{56.1} & \numprint{90.7} & \numprint{90.8} & &\numprint{49.2} & \numprint{84.7} & \numprint{84.9}\\ 
dewiki-2013 & &\numprint{24.26} & \numprint{1.10} & \numprint{0.22} & &\numprint{22.0} & \numprint{110.7} & &\numprint{36.0} & \numprint{80.3} & \numprint{83.7} & &\numprint{10.0} & \numprint{42.2} & \numprint{50.2}\\ 
dielFilter. & &\numprint{2.79} & \numprint{0.19} & \numprint{0.03} & &\numprint{14.5} & \numprint{111.0} & &\numprint{0.0} & \numprint{64.0} & \numprint{64.6} & &\numprint{0.0} & \numprint{41.8} & \numprint{42.8}\\ 
enwiki-2013 & &\numprint{39.87} & \numprint{3.13} & \numprint{0.53} & &\numprint{12.7} & \numprint{75.4} & &\numprint{20.3} & \numprint{47.0} & \numprint{51.4} & &\numprint{5.9} & \numprint{16.4} & \numprint{20.1}\\ 
enwiki-2018 & &\numprint{46.97} & \numprint{3.75} & \numprint{0.66} & &\numprint{12.5} & \numprint{71.3} & &\numprint{14.5} & \numprint{41.3} & \numprint{45.8} & &\numprint{4.3} & \numprint{13.6} & \numprint{16.9}\\ 
enwiki-2022 & &\numprint{68.37} & \numprint{5.00} & \numprint{0.56} & &\numprint{13.7} & \numprint{122.3} & &\numprint{18.5} & \numprint{47.4} & \numprint{51.0} & &\numprint{5.7} & \numprint{18.0} & \numprint{21.1}\\ 
europe.osm & &\numprint{20.36} & \numprint{0.88} & \numprint{0.11} & &\numprint{23.1} & \numprint{189.0} & &\numprint{0.1} & \numprint{1.3} & \numprint{1.5} & &\numprint{0.2} & \numprint{1.5} & \numprint{1.7}\\ 
frwiki-2013 & &\numprint{13.04} & \numprint{0.93} & \numprint{0.24} & &\numprint{14.0} & \numprint{53.2} & &\numprint{28.6} & \numprint{55.9} & \numprint{59.1} & &\numprint{9.1} & \numprint{21.5} & \numprint{24.5}\\ 
germany.osm & &\numprint{4.64} & \numprint{0.24} & \numprint{0.05} & &\numprint{19.6} & \numprint{99.0} & &\numprint{0.1} & \numprint{1.3} & \numprint{1.5} & &\numprint{0.1} & \numprint{1.4} & \numprint{1.6}\\ 
hollywood-. & &\numprint{67.62} & \numprint{3.62} & \numprint{0.43} & &\numprint{18.7} & \numprint{156.3} & &\numprint{0.0} & \numprint{18.5} & \numprint{38.5} & &\numprint{0.0} & \numprint{12.9} & \numprint{39.0}\\ 
imdb-2021 & &\numprint{0.83} & \numprint{0.12} & \numprint{0.06} & &\numprint{6.7} & \numprint{13.6} & &\numprint{0.0} & \numprint{0.1} & \numprint{0.6} & &\numprint{0.0} & \numprint{0.1} & \numprint{1.0}\\ 
indochina-. & &\numprint{165.55} & \numprint{7.30} & \numprint{1.42} & &\numprint{22.7} & \numprint{116.4} & &\numprint{19.6} & \numprint{22.3} & \numprint{27.5} & &\numprint{20.2} & \numprint{22.2} & \numprint{30.8}\\ 
it-2004 & &\numprint{1428.94} & \numprint{179.36} & \numprint{60.13} & &\numprint{8.0} & \numprint{23.8} & &\numprint{22.5} & \numprint{26.5} & \numprint{27.6} & &\numprint{28.4} & \numprint{30.8} & \numprint{33.1}\\ 
itwiki-2013 & &\numprint{9.08} & \numprint{0.62} & \numprint{0.15} & &\numprint{14.6} & \numprint{61.7} & &\numprint{19.9} & \numprint{51.9} & \numprint{54.8} & &\numprint{5.1} & \numprint{18.3} & \numprint{21.1}\\ 
ljournal-2. & &\numprint{9.53} & \numprint{0.91} & \numprint{0.19} & &\numprint{10.4} & \numprint{49.1} & &\numprint{1.5} & \numprint{8.4} & \numprint{9.3} & &\numprint{3.2} & \numprint{5.3} & \numprint{6.8}\\ 
roadNet-CA & &\numprint{1.16} & \numprint{0.06} & \numprint{0.02} & &\numprint{20.0} & \numprint{50.3} & &\numprint{5.0} & \numprint{9.5} & \numprint{16.8} & &\numprint{6.0} & \numprint{10.3} & \numprint{17.5}\\ 
roadNet-PA & &\numprint{0.61} & \numprint{0.03} & \numprint{0.03} & &\numprint{18.1} & \numprint{20.9} & &\numprint{5.3} & \numprint{9.3} & \numprint{18.3} & &\numprint{6.4} & \numprint{10.3} & \numprint{19.0}\\ 
roadNet-PA* & &\numprint{0.61} & \numprint{0.03} & \numprint{0.02} & &\numprint{18.3} & \numprint{26.3} & &\numprint{5.4} & \numprint{9.3} & \numprint{18.3} & &\numprint{6.5} & \numprint{10.3} & \numprint{19.1}\\ 
roadNet-TX & &\numprint{0.74} & \numprint{0.04} & \numprint{0.02} & &\numprint{20.4} & \numprint{36.0} & &\numprint{4.8} & \numprint{7.8} & \numprint{15.2} & &\numprint{5.9} & \numprint{8.9} & \numprint{16.1}\\ 
sk-2005 & &\numprint{4822.13} & \numprint{991.81} & \numprint{180.91} & &\numprint{4.9} & \numprint{26.7} & &\numprint{32.3} & \numprint{39.4} & \numprint{45.3} & &\numprint{23.9} & \numprint{28.3} & \numprint{36.9}\\ 
soc-LiveJo. & &\numprint{9.71} & \numprint{1.21} & \numprint{0.19} & &\numprint{8.0} & \numprint{52.5} & &\numprint{1.4} & \numprint{31.1} & \numprint{32.2} & &\numprint{2.2} & \numprint{18.5} & \numprint{19.6}\\ 
soc-flickr. & &\numprint{2.45} & \numprint{0.33} & \numprint{0.11} & &\numprint{7.4} & \numprint{22.5} & &\numprint{0.2} & \numprint{3.7} & \numprint{4.1} & &\numprint{0.1} & \numprint{2.4} & \numprint{2.5}\\ 
soc-pokec-. & &\numprint{12.54} & \numprint{0.59} & \numprint{0.13} & &\numprint{21.2} & \numprint{94.6} & &\numprint{53.7} & \numprint{69.5} & \numprint{69.8} & &\numprint{47.9} & \numprint{61.9} & \numprint{62.2}\\ 
soc-sinawe. & &\numprint{45.75} & \numprint{6.07} & \numprint{1.17} & &\numprint{7.5} & \numprint{39.2} & &\numprint{0.0} & \numprint{0.0} & \numprint{0.0} & &\numprint{0.0} & \numprint{0.0} & \numprint{0.0}\\ 
twitter-20. & &\numprint{473.28} & \numprint{47.06} & \numprint{4.32} & &\numprint{10.1} & \numprint{109.6} & &\numprint{0.4} & \numprint{5.9} & \numprint{6.2} & &\numprint{0.1} & \numprint{0.6} & \numprint{0.6}\\ 
uk-2005 & &\numprint{309.44} & \numprint{315.96} & \numprint{31.52} & &\numprint{1.0} & \numprint{9.8} & &\numprint{13.7} & \numprint{19.8} & \numprint{20.5} & &\numprint{10.8} & \numprint{17.2} & \numprint{20.8}\\ 
webbase-20. & &\numprint{208.00} & \numprint{20.10} & \numprint{6.38} & &\numprint{10.3} & \numprint{32.6} & &\numprint{8.9} & \numprint{10.8} & \numprint{10.9} & &\numprint{13.5} & \numprint{15.0} & \numprint{15.6}\\ 
wiki-Talk & &\numprint{0.64} & \numprint{0.10} & \numprint{0.04} & &\numprint{6.2} & \numprint{14.3} & &\numprint{0.0} & \numprint{0.0} & \numprint{0.0} & &\numprint{0.0} & \numprint{0.0} & \numprint{0.0}\\ 
wiki-topca. & &\numprint{7.24} & \numprint{0.59} & \numprint{0.17} & &\numprint{12.2} & \numprint{43.7} & &\numprint{11.2} & \numprint{33.9} & \numprint{37.7} & &\numprint{4.2} & \numprint{13.4} & \numprint{15.5}\\ 

    \end{tabular}
  }}
  \end{center}
\end{table*}
 
\begin{table*}
  \caption{Reduction running times of {\KaDisReduA} and {\KaDisReduS} for \numprint{1}, \numprint{64}, \numprint{256} and \numprint{1024} cores (from left to right) in seconds. 
    \textsf{roadNet-PA*} refers to the uniformly weighted graph.
  }\label{tab:reduce-detailed-red-t}
  \begin{center}
    \small{\textsf{
    \begin{tabular}{rrrrrrrrrrr}
        \mc{1}{r}{} && \mc{9}{c}{$t_{\text{reduce}}$~[s]} \\ 
        \mc{1}{r}{} && \mc{4}{c}{\KaDisReduS} && \mc{4}{c}{\KaDisReduA} \\ 
        \mc{1}{r}{} && \mc{1}{c}{1} & \mc{1}{c}{64} & \mc{1}{c}{256} & \mc{1}{c}{1024} && \mc{1}{c}{1} & \mc{1}{c}{64} & \mc{1}{c}{256} & \mc{1}{c}{1024} \\ 
        \cmidrule{1-1}  \cmidrule{3-6} \cmidrule{8-11} 
      Bump\textunderscore 2911 & &\textbf{\numprint{5.333}} & \textbf{\numprint{0.168}} & \textbf{\numprint{0.050}} & \numprint{0.021} & &\numprint{5.595} & \numprint{0.187} & \numprint{0.054} & \textbf{\numprint{0.020}}\\ 
Cube\textunderscore Coup\textunderscore . & &\textbf{\numprint{5.098}} & \numprint{0.160} & \numprint{0.048} & \textbf{\numprint{0.019}} & &\numprint{5.263} & \textbf{\numprint{0.150}} & \textbf{\numprint{0.046}} & \numprint{0.021}\\ 
Flan\textunderscore 1565 & &\textbf{\numprint{4.419}} & \numprint{0.132} & \textbf{\numprint{0.039}} & \numprint{0.017} & &\numprint{4.612} & \textbf{\numprint{0.126}} & \textbf{\numprint{0.039}} & \textbf{\numprint{0.014}}\\ 
Geo\textunderscore 1438 & &\textbf{\numprint{2.643}} & \textbf{\numprint{0.092}} & \numprint{0.036} & \textbf{\numprint{0.012}} & &\numprint{2.759} & \numprint{0.106} & \textbf{\numprint{0.030}} & \numprint{0.015}\\ 
HV15R & &\textbf{\numprint{14.403}} & \numprint{0.443} & \numprint{0.135} & \numprint{0.048} & &\numprint{14.573} & \textbf{\numprint{0.378}} & \textbf{\numprint{0.118}} & \textbf{\numprint{0.043}}\\ 
Hook\textunderscore 1498 & &\textbf{\numprint{3.318}} & \numprint{0.090} & \textbf{\numprint{0.027}} & \textbf{\numprint{0.013}} & &\numprint{3.457} & \textbf{\numprint{0.089}} & \textbf{\numprint{0.027}} & \numprint{0.018}\\ 
Long\textunderscore Coup\textunderscore . & &\textbf{\numprint{3.462}} & \numprint{0.113} & \textbf{\numprint{0.035}} & \textbf{\numprint{0.013}} & &\numprint{3.634} & \textbf{\numprint{0.112}} & \numprint{0.038} & \numprint{0.014}\\ 
ML\textunderscore Geer & &\textbf{\numprint{4.712}} & \textbf{\numprint{0.104}} & \textbf{\numprint{0.033}} & \textbf{\numprint{0.014}} & &\numprint{4.749} & \numprint{0.108} & \numprint{0.037} & \numprint{0.016}\\ 
Queen\textunderscore 4147 & &\textbf{\numprint{12.574}} & \numprint{0.385} & \numprint{0.111} & \textbf{\numprint{0.029}} & &\numprint{12.988} & \textbf{\numprint{0.344}} & \textbf{\numprint{0.106}} & \numprint{0.042}\\ 
Serena & &\textbf{\numprint{2.670}} & \textbf{\numprint{0.097}} & \textbf{\numprint{0.030}} & \textbf{\numprint{0.014}} & &\numprint{2.794} & \numprint{0.101} & \numprint{0.034} & \numprint{0.017}\\ 
USA-road-d & &\textbf{\numprint{13.536}} & \numprint{0.983} & \numprint{0.269} & \numprint{0.091} & &\numprint{15.660} & \textbf{\numprint{0.817}} & \textbf{\numprint{0.215}} & \textbf{\numprint{0.061}}\\ 
arabic-2005 & &\textbf{\numprint{383.803}} & \numprint{26.474} & \numprint{16.784} & \numprint{10.435} & &\numprint{385.780} & \textbf{\numprint{24.014}} & \textbf{\numprint{13.459}} & \textbf{\numprint{5.918}}\\ 
as-skitter & &\textbf{\numprint{1.773}} & \numprint{0.215} & \textbf{\numprint{0.121}} & \numprint{0.097} & &\numprint{1.995} & \textbf{\numprint{0.186}} & \numprint{0.136} & \textbf{\numprint{0.037}}\\ 
asia.osm & &\textbf{\numprint{4.548}} & \numprint{0.205} & \numprint{0.069} & \numprint{0.054} & &\numprint{5.572} & \textbf{\numprint{0.175}} & \textbf{\numprint{0.055}} & \textbf{\numprint{0.030}}\\ 
buddha & &\textbf{\numprint{1.075}} & \numprint{0.035} & \numprint{0.022} & \numprint{0.022} & &\numprint{1.129} & \textbf{\numprint{0.028}} & \textbf{\numprint{0.017}} & \textbf{\numprint{0.009}}\\ 
circuit5M & &\textbf{\numprint{5.451}} & \numprint{0.668} & \numprint{0.450} & \textbf{\numprint{0.368}} & &\numprint{6.468} & \textbf{\numprint{0.571}} & \textbf{\numprint{0.393}} & \numprint{0.417}\\ 
com-youtube & &\textbf{\numprint{0.393}} & \numprint{0.080} & \numprint{0.043} & \numprint{0.044} & &\numprint{0.564} & \textbf{\numprint{0.078}} & \textbf{\numprint{0.034}} & \textbf{\numprint{0.020}}\\ 
delaunay\textunderscore n. & &\textbf{\numprint{84.479}} & \numprint{0.810} & \numprint{0.240} & \numprint{0.076} & &\numprint{84.925} & \textbf{\numprint{0.710}} & \textbf{\numprint{0.199}} & \textbf{\numprint{0.061}}\\ 
dewiki-2013 & &\textbf{\numprint{24.258}} & \numprint{1.103} & \textbf{\numprint{0.429}} & \numprint{0.219} & &\numprint{24.417} & \textbf{\numprint{0.715}} & \numprint{0.479} & \textbf{\numprint{0.089}}\\ 
dielFilter. & &\numprint{2.794} & \numprint{0.193} & \textbf{\numprint{0.061}} & \textbf{\numprint{0.025}} & &\textbf{\numprint{2.791}} & \textbf{\numprint{0.190}} & \numprint{0.071} & \textbf{\numprint{0.025}}\\ 
enwiki-2013 & &\textbf{\numprint{39.871}} & \numprint{3.135} & \numprint{1.130} & \numprint{0.529} & &\numprint{40.443} & \textbf{\numprint{1.263}} & \textbf{\numprint{0.552}} & \textbf{\numprint{0.239}}\\ 
enwiki-2018 & &\textbf{\numprint{46.973}} & \numprint{3.752} & \numprint{1.242} & \numprint{0.659} & &\numprint{47.417} & \textbf{\numprint{1.712}} & \textbf{\numprint{0.859}} & \textbf{\numprint{0.331}}\\ 
enwiki-2022 & &\textbf{\numprint{68.372}} & \numprint{5.004} & \numprint{1.447} & \numprint{0.559} & &\numprint{68.872} & \textbf{\numprint{2.250}} & \textbf{\numprint{0.981}} & \textbf{\numprint{0.350}}\\ 
europe.osm & &\textbf{\numprint{20.363}} & \numprint{0.883} & \textbf{\numprint{0.284}} & \numprint{0.108} & &\numprint{24.923} & \textbf{\numprint{0.814}} & \numprint{0.304} & \textbf{\numprint{0.080}}\\ 
frwiki-2013 & &\textbf{\numprint{13.037}} & \numprint{0.928} & \numprint{0.562} & \numprint{0.245} & &\numprint{13.095} & \textbf{\numprint{0.452}} & \textbf{\numprint{0.358}} & \textbf{\numprint{0.108}}\\ 
germany.osm & &\textbf{\numprint{4.635}} & \numprint{0.236} & \numprint{0.078} & \numprint{0.047} & &\numprint{5.702} & \textbf{\numprint{0.198}} & \textbf{\numprint{0.063}} & \textbf{\numprint{0.023}}\\ 
hollywood-. & &\textbf{\numprint{67.620}} & \numprint{3.624} & \numprint{1.027} & \numprint{0.433} & &\numprint{67.708} & \textbf{\numprint{2.369}} & \textbf{\numprint{0.907}} & \textbf{\numprint{0.413}}\\ 
imdb-2021 & &\textbf{\numprint{0.828}} & \textbf{\numprint{0.123}} & \textbf{\numprint{0.066}} & \numprint{0.061} & &\numprint{1.287} & \numprint{0.241} & \numprint{0.081} & \textbf{\numprint{0.031}}\\ 
indochina-. & &\numprint{165.546} & \numprint{7.301} & \numprint{3.375} & \numprint{1.423} & &\textbf{\numprint{161.156}} & \textbf{\numprint{4.678}} & \textbf{\numprint{1.385}} & \textbf{\numprint{1.126}}\\ 
it-2004 & &\textbf{\numprint{1428.942}} & \textbf{\numprint{179.364}} & \numprint{148.049} & \numprint{60.134} & &\numprint{1433.771} & \numprint{183.580} & \textbf{\numprint{97.570}} & \textbf{\numprint{43.286}}\\ 
itwiki-2013 & &\textbf{\numprint{9.081}} & \numprint{0.624} & \textbf{\numprint{0.251}} & \numprint{0.147} & &\numprint{9.134} & \textbf{\numprint{0.416}} & \numprint{0.345} & \textbf{\numprint{0.067}}\\ 
ljournal-2. & &\textbf{\numprint{9.530}} & \numprint{0.915} & \numprint{0.455} & \numprint{0.194} & &\numprint{10.093} & \textbf{\numprint{0.662}} & \textbf{\numprint{0.283}} & \textbf{\numprint{0.091}}\\ 
roadNet-CA & &\textbf{\numprint{1.160}} & \numprint{0.058} & \numprint{0.025} & \numprint{0.023} & &\numprint{1.263} & \textbf{\numprint{0.049}} & \textbf{\numprint{0.018}} & \textbf{\numprint{0.009}}\\ 
roadNet-PA & &\textbf{\numprint{0.609}} & \numprint{0.034} & \numprint{0.022} & \numprint{0.029} & &\numprint{0.682} & \textbf{\numprint{0.027}} & \textbf{\numprint{0.016}} & \textbf{\numprint{0.008}}\\ 
roadNet-PA* & &\textbf{\numprint{0.611}} & \numprint{0.033} & \numprint{0.020} & \numprint{0.023} & &\numprint{0.678} & \textbf{\numprint{0.029}} & \textbf{\numprint{0.015}} & \textbf{\numprint{0.007}}\\ 
roadNet-TX & &\textbf{\numprint{0.743}} & \numprint{0.036} & \numprint{0.022} & \numprint{0.021} & &\numprint{0.830} & \textbf{\numprint{0.033}} & \textbf{\numprint{0.017}} & \textbf{\numprint{0.010}}\\ 
sk-2005 & &\textbf{\numprint{4822.129}} & \numprint{991.808} & \numprint{114.915} & \numprint{180.908} & &\numprint{5130.964} & \textbf{\numprint{526.402}} & \textbf{\numprint{101.910}} & \textbf{\numprint{123.591}}\\ 
soc-LiveJo. & &\textbf{\numprint{9.710}} & \numprint{1.209} & \numprint{0.469} & \numprint{0.185} & &\numprint{10.181} & \textbf{\numprint{0.828}} & \textbf{\numprint{0.459}} & \textbf{\numprint{0.104}}\\ 
soc-flickr. & &\textbf{\numprint{2.454}} & \numprint{0.332} & \textbf{\numprint{0.184}} & \numprint{0.109} & &\numprint{2.708} & \textbf{\numprint{0.293}} & \numprint{0.265} & \textbf{\numprint{0.053}}\\ 
soc-pokec-. & &\textbf{\numprint{12.544}} & \numprint{0.592} & \textbf{\numprint{0.262}} & \numprint{0.133} & &\numprint{12.614} & \textbf{\numprint{0.434}} & \numprint{0.394} & \textbf{\numprint{0.053}}\\ 
soc-sinawe. & &\textbf{\numprint{45.752}} & \numprint{6.069} & \textbf{\numprint{1.980}} & \textbf{\numprint{1.166}} & &\numprint{57.388} & \textbf{\numprint{5.794}} & \numprint{2.076} & \numprint{1.276}\\ 
twitter-20. & &\textbf{\numprint{473.280}} & \textbf{\numprint{47.060}} & \textbf{\numprint{14.665}} & \textbf{\numprint{4.318}} & &\numprint{479.263} & \numprint{58.081} & \numprint{15.711} & \numprint{4.606}\\ 
uk-2005 & &\textbf{\numprint{309.438}} & \numprint{315.960} & \numprint{109.286} & \numprint{31.520} & &\numprint{313.671} & \textbf{\numprint{270.408}} & \textbf{\numprint{96.742}} & \textbf{\numprint{20.018}}\\ 
webbase-20. & &\textbf{\numprint{208.001}} & \numprint{20.100} & \numprint{13.307} & \numprint{6.376} & &\numprint{222.492} & \textbf{\numprint{5.416}} & \textbf{\numprint{2.819}} & \textbf{\numprint{1.463}}\\ 
wiki-Talk & &\textbf{\numprint{0.638}} & \textbf{\numprint{0.103}} & \textbf{\numprint{0.067}} & \numprint{0.045} & &\numprint{1.101} & \numprint{0.116} & \numprint{0.181} & \textbf{\numprint{0.036}}\\ 
wiki-topca. & &\textbf{\numprint{7.241}} & \numprint{0.592} & \textbf{\numprint{0.251}} & \numprint{0.166} & &\numprint{7.429} & \textbf{\numprint{0.307}} & \numprint{0.301} & \textbf{\numprint{0.088}}\\ 
\cmidrule{1-1}  \cmidrule{3-6} \cmidrule{8-11} 
\textbf{geo. mean} & &\textbf{\numprint{11.271}} & \numprint{0.729} & \numprint{0.291} & \numprint{0.155} & &\numprint{12.259} & \textbf{\numprint{0.588}} & \textbf{\numprint{0.258}} & \textbf{\numprint{0.100}}\\ 
\cmidrule{1-1}  \cmidrule{3-6} \cmidrule{8-11} 
com-Friend. & &\textbf{\numprint{1414.621}} & - & \numprint{25.473} & \numprint{7.258} & &\numprint{1417.997} & - & \textbf{\numprint{16.514}} & \textbf{\numprint{6.802}}\\ 

    \end{tabular}
  }}
  \end{center}
\end{table*}

\begin{table*}
  \caption{Comparing {\aRnP} with {\htwis} for 1 and \numprint{1024} cores in terms of arithmetic mean solution weight and running time. The speed up \textsf{su} is relative to {\htwis}. Best solution and best running time are \textbf{bold}.
    \textsf{roadNet-PA*} refers to the uniformly weighted graph.
  }\label{tab:detailed-aRnP}
  \begin{center}
    \small{\textsf{
    \begin{tabular}{rrrrrrrrrrr}
        \mc{1}{r}{} && \mc{5}{c}{\aRnP} && \mc{2}{c}{\htwis} \\ 
        \mc{1}{r}{} && \mc{2}{c}{1} && \mc{3}{c}{1024} && \mc{2}{c}{1} \\ 
        \mc{1}{r}{} && \mc{1}{c}{$\w(\I)$} & \mc{1}{c}{$t$~[s]} && \mc{1}{c}{$\w(\I)$} & \mc{1}{c}{$t$~[s]} & \mc{1}{c}{su} && \mc{1}{c}{$\w(\I)$} & \mc{1}{c}{$t$~[s]} \\ 
        \cmidrule{1-1}  \cmidrule{3-4}  \cmidrule{6-8}  \cmidrule{10-11} 
      Bump\textunderscore 2911 & &\numprint{27450888.0} & \numprint{64.09} & &\numprint{26249878.2} & \textbf{\numprint{0.10}} & \numprint{100.6} & &\textbf{\numprint{27709157.0}} & \numprint{9.82}\\ 
Cube\textunderscore Coup\textunderscore . & &\numprint{16427017.0} & \numprint{47.52} & &\numprint{15585244.5} & \textbf{\numprint{0.08}} & \numprint{89.5} & &\textbf{\numprint{16625019.0}} & \numprint{7.13}\\ 
Flan\textunderscore 1565 & &\numprint{10823181.0} & \numprint{44.66} & &\numprint{10085800.0} & \textbf{\numprint{0.08}} & \numprint{124.4} & &\textbf{\numprint{10999528.0}} & \numprint{9.60}\\ 
Geo\textunderscore 1438 & &\numprint{13902670.0} & \numprint{31.20} & &\numprint{13112299.2} & \textbf{\numprint{0.06}} & \numprint{88.0} & &\textbf{\numprint{14024450.0}} & \numprint{5.10}\\ 
HV15R & &\numprint{10391309.0} & \numprint{189.79} & &\numprint{9672003.0} & \textbf{\numprint{0.31}} & \numprint{217.4} & &\textbf{\numprint{10495774.0}} & \numprint{68.47}\\ 
Hook\textunderscore 1498 & &\textbf{\numprint{24390472.0}} & \numprint{19.84} & &\numprint{23770132.5} & \textbf{\numprint{0.07}} & \numprint{60.9} & &\numprint{24387753.0} & \numprint{4.24}\\ 
Long\textunderscore Coup\textunderscore . & &\numprint{11025545.0} & \numprint{32.22} & &\numprint{10406824.8} & \textbf{\numprint{0.07}} & \numprint{75.1} & &\textbf{\numprint{11136959.0}} & \numprint{5.18}\\ 
ML\textunderscore Geer & &\textbf{\numprint{8188887.0}} & \numprint{52.94} & &\numprint{7592268.2} & \textbf{\numprint{0.10}} & \numprint{155.1} & &\numprint{7994470.0} & \numprint{15.56}\\ 
Queen\textunderscore 4147 & &\numprint{26077683.0} & \numprint{134.57} & &\numprint{24650320.5} & \textbf{\numprint{0.23}} & \numprint{127.2} & &\textbf{\numprint{26387860.0}} & \numprint{29.73}\\ 
Serena & &\numprint{13619417.0} & \numprint{31.59} & &\numprint{12986078.2} & \textbf{\numprint{0.06}} & \numprint{92.5} & &\textbf{\numprint{13717579.0}} & \numprint{5.51}\\ 
USA-road-d & &\textbf{\numprint{1427296277.0}} & \numprint{20.50} & &\numprint{1425784493.5} & \textbf{\numprint{0.15}} & \numprint{41.4} & &\numprint{1427192290.0} & \numprint{6.07}\\ 
arabic-2005 & &\textbf{\numprint{1477794453.0}} & \numprint{1124.84} & &\numprint{1477026539.5} & \textbf{\numprint{29.24}} & \numprint{34.6} & &\numprint{1477671476.0} & \numprint{1013.04}\\ 
as-skitter & &\textbf{\numprint{124145588.0}} & \numprint{2.70} & &\numprint{123679325.8} & \textbf{\numprint{0.07}} & \numprint{10.6} & &\numprint{124141373.0} & \numprint{0.79}\\ 
asia.osm & &\textbf{\numprint{703469054.0}} & \numprint{6.80} & &\numprint{703328651.2} & \textbf{\numprint{0.05}} & \numprint{35.1} & &\numprint{703455351.0} & \numprint{1.90}\\ 
buddha & &\textbf{\numprint{57544080.0}} & \numprint{1.74} & &\numprint{57260793.2} & \textbf{\numprint{0.03}} & \numprint{12.7} & &\numprint{57508556.0} & \numprint{0.36}\\ 
circuit5M & &\textbf{\numprint{522198078.0}} & \numprint{7.04} & &\numprint{522187318.0} & \textbf{\numprint{0.69}} & \numprint{3.4} & &\numprint{522196403.0} & \numprint{2.32}\\ 
com-youtube & &\textbf{\numprint{90295285.0}} & \numprint{0.66} & &\numprint{90278682.8} & \textbf{\numprint{0.04}} & \numprint{5.5} & &\textbf{\numprint{90295285.0}} & \numprint{0.22}\\ 
delaunay\textunderscore n. & &\textbf{\numprint{650103112.0}} & \numprint{169.81} & &\numprint{644309225.2} & \textbf{\numprint{0.26}} & \numprint{65.3} & &\numprint{648575756.0} & \numprint{16.71}\\ 
dewiki-2013 & &\textbf{\numprint{76968052.0}} & \numprint{42.50} & &\numprint{76523833.0} & \textbf{\numprint{0.26}} & \numprint{45.8} & &\numprint{74720376.0} & \numprint{11.96}\\ 
dielFilter. & &\textbf{\numprint{9386363.0}} & \numprint{3.11} & &\numprint{8419141.8} & \textbf{\numprint{0.09}} & \numprint{69.1} & &\textbf{\numprint{9386363.0}} & \numprint{5.90}\\ 
enwiki-2013 & &\textbf{\numprint{235386036.0}} & \numprint{69.38} & &\numprint{234462905.5} & \textbf{\numprint{0.66}} & \numprint{32.5} & &\numprint{235008063.0} & \numprint{21.42}\\ 
enwiki-2018 & &\textbf{\numprint{322698398.0}} & \numprint{73.62} & &\numprint{321808260.0} & \textbf{\numprint{0.83}} & \numprint{26.0} & &\numprint{321843964.0} & \numprint{21.45}\\ 
enwiki-2022 & &\textbf{\numprint{366989081.0}} & \numprint{113.95} & &\numprint{365866656.5} & \textbf{\numprint{0.98}} & \numprint{36.6} & &\numprint{366202711.0} & \numprint{35.88}\\ 
frwiki-2013 & &\textbf{\numprint{72589474.0}} & \numprint{26.26} & &\numprint{72248835.2} & \textbf{\numprint{0.24}} & \numprint{35.9} & &\numprint{72468246.0} & \numprint{8.63}\\ 
germany.osm & &\textbf{\numprint{683336494.0}} & \numprint{6.86} & &\numprint{683246365.8} & \textbf{\numprint{0.06}} & \numprint{38.9} & &\numprint{683326451.0} & \numprint{2.17}\\ 
itwiki-2013 & &\textbf{\numprint{58309171.0}} & \numprint{14.53} & &\numprint{58113445.5} & \textbf{\numprint{0.16}} & \numprint{29.6} & &\numprint{58284373.0} & \numprint{4.81}\\ 
ljournal-2. & &\numprint{320046313.0} & \numprint{15.38} & &\numprint{319887314.0} & \textbf{\numprint{0.34}} & \numprint{17.9} & &\textbf{\numprint{320046546.0}} & \numprint{6.06}\\ 
roadNet-CA & &\textbf{\numprint{111345705.0}} & \numprint{1.84} & &\numprint{111225024.8} & \textbf{\numprint{0.03}} & \numprint{18.5} & &\numprint{111325524.0} & \numprint{0.52}\\ 
roadNet-PA & &\textbf{\numprint{61698498.0}} & \numprint{0.99} & &\numprint{61615092.5} & \textbf{\numprint{0.03}} & \numprint{9.8} & &\numprint{61688549.0} & \numprint{0.26}\\ 
roadNet-PA* & &\textbf{\numprint{61723052.0}} & \numprint{0.99} & &\numprint{61641848.2} & \textbf{\numprint{0.02}} & \numprint{10.6} & &\numprint{61710606.0} & \numprint{0.25}\\ 
roadNet-TX & &\textbf{\numprint{78587576.0}} & \numprint{1.21} & &\numprint{78506802.0} & \textbf{\numprint{0.03}} & \numprint{11.2} & &\numprint{78575460.0} & \numprint{0.33}\\ 
soc-LiveJo. & &\textbf{\numprint{284022934.0}} & \numprint{13.88} & &\numprint{283074615.0} & \textbf{\numprint{0.31}} & \numprint{20.6} & &\numprint{283922214.0} & \numprint{6.38}\\ 
soc-flickr. & &\textbf{\numprint{129805768.0}} & \numprint{3.00} & &\numprint{129738317.5} & \textbf{\numprint{0.09}} & \numprint{7.7} & &\numprint{129805546.0} & \numprint{0.73}\\ 
soc-pokec-. & &\textbf{\numprint{83963760.0}} & \numprint{60.50} & &\numprint{83453179.0} & \textbf{\numprint{0.18}} & \numprint{125.3} & &\numprint{83920370.0} & \numprint{22.04}\\ 
wiki-Talk & &\textbf{\numprint{235837346.0}} & \numprint{1.24} & &\numprint{235834009.0} & \textbf{\numprint{0.05}} & \numprint{7.3} & &\textbf{\numprint{235837346.0}} & \numprint{0.34}\\ 
wiki-topca. & &\textbf{\numprint{106674375.0}} & \numprint{12.71} & &\numprint{106295980.8} & \textbf{\numprint{0.23}} & \numprint{16.5} & &\numprint{106654405.0} & \numprint{3.73}\\ 
\cmidrule{1-1}  \cmidrule{3-4}  \cmidrule{6-8}  \cmidrule{10-11} 
\textbf{geo. mean} & &\numprint{87643763.3} & \numprint{15.91} & &\numprint{85879255.7} & \textbf{\numprint{0.14}} & - & &\textbf{\numprint{87687741.9}} & \numprint{4.51}\\ 
\cmidrule{1-1}  \cmidrule{3-4}  \cmidrule{6-8}  \cmidrule{10-11} 
com-Friend. & &\textbf{\numprint{3832609030.0}} & \numprint{2793.41} & &\numprint{3821937482.5} & \textbf{\numprint{17.08}} & - & &- & -\\ 
europe.osm & &\textbf{\numprint{3003671463.0}} & \numprint{30.21} & &\numprint{3003318435.2} & \textbf{\numprint{0.16}} & - & &- & -\\ 
hollywood-. & &\textbf{\numprint{64876723.0}} & \numprint{68.68} & &\numprint{64821869.0} & \textbf{\numprint{1.03}} & \numprint{15.6} & &\numprint{45215004.0} & \numprint{16.02}\\ 
imdb-2021 & &\textbf{\numprint{242337085.0}} & \numprint{1.45} & &\numprint{242321347.8} & \textbf{\numprint{0.05}} & \numprint{6.6} & &\numprint{64162924.0} & \numprint{0.32}\\ 
indochina-. & &- & - & &\numprint{498050975.0} & \textbf{\numprint{2.70}} & \numprint{38.0} & &\textbf{\numprint{498132749.0}} & \numprint{102.51}\\ 
it-2004 & &\textbf{\numprint{2714543237.0}} & \numprint{3079.12} & &\numprint{2713294694.8} & \textbf{\numprint{89.08}} & - & &- & -\\ 
sk-2005 & &- & - & &\textbf{\numprint{3226888130.5}} & \textbf{\numprint{134.77}} & - & &- & -\\ 
soc-sinawe. & &\textbf{\numprint{5872160022.0}} & \numprint{61.86} & &\numprint{5872158595.2} & \textbf{\numprint{1.64}} & - & &- & -\\ 
twitter-20. & &\textbf{\numprint{3025416086.0}} & \numprint{502.11} & &\numprint{3024726153.8} & \textbf{\numprint{6.73}} & - & &- & -\\ 
uk-2005 & &\textbf{\numprint{2509963974.0}} & \numprint{550.23} & &\numprint{2509442802.0} & \textbf{\numprint{26.80}} & - & &- & -\\ 
webbase-20. & &\textbf{\numprint{8432593760.0}} & \numprint{673.09} & &\numprint{8431732396.0} & \textbf{\numprint{18.36}} & - & &- & -\\ 

    \end{tabular}
  }}
  \end{center}
\end{table*}

\begin{table*}
  \caption{Comparing {\sRnP} with {\htwis} for 1 and \numprint{1024} cores in terms of arithmetic mean solution weight and running time. The speed up \textsf{su} is relative to {\htwis}. Best solution and best running time are \textbf{bold}.
      \textsf{roadNet-PA*} refers to the uniformly weighted graph.
  }\label{tab:detailed-sRnP}
  \begin{center}
    \small{\textsf{
    \begin{tabular}{rrrrrrrrrrr}
        \mc{1}{r}{} && \mc{5}{c}{\aRnP} && \mc{2}{c}{\htwis} \\ 
        \mc{1}{r}{} && \mc{2}{c}{1} && \mc{3}{c}{1024} && \mc{2}{c}{1} \\ 
        \mc{1}{r}{} && \mc{1}{c}{$\w(\I)$} & \mc{1}{c}{$t$~[s]} && \mc{1}{c}{$\w(\I)$} & \mc{1}{c}{$t$~[s]} & \mc{1}{c}{su} && \mc{1}{c}{$\w(\I)$} & \mc{1}{c}{$t$~[s]} \\ 
        \cmidrule{1-1}  \cmidrule{3-4}  \cmidrule{6-8}  \cmidrule{10-11} 
      Bump\textunderscore 2911 & &\numprint{27450888.0} & \numprint{58.94} & &\numprint{26005784.0} & \textbf{\numprint{0.15}} & \numprint{65.1} & &\textbf{\numprint{27709157.0}} & \numprint{9.82}\\ 
Cube\textunderscore Coup\textunderscore . & &\numprint{16427017.0} & \numprint{44.07} & &\numprint{15522683.0} & \textbf{\numprint{0.13}} & \numprint{53.9} & &\textbf{\numprint{16625019.0}} & \numprint{7.13}\\ 
Flan\textunderscore 1565 & &\numprint{10823181.0} & \numprint{42.04} & &\numprint{10020887.0} & \textbf{\numprint{0.12}} & \numprint{77.7} & &\textbf{\numprint{10999528.0}} & \numprint{9.60}\\ 
Geo\textunderscore 1438 & &\numprint{13902670.0} & \numprint{28.66} & &\numprint{13095179.0} & \textbf{\numprint{0.12}} & \numprint{43.6} & &\textbf{\numprint{14024450.0}} & \numprint{5.10}\\ 
HV15R & &\numprint{10391309.0} & \numprint{185.29} & &\numprint{9676668.0} & \textbf{\numprint{0.45}} & \numprint{153.2} & &\textbf{\numprint{10495774.0}} & \numprint{68.47}\\ 
Hook\textunderscore 1498 & &\textbf{\numprint{24390472.0}} & \numprint{18.67} & &\numprint{23802947.0} & \textbf{\numprint{0.35}} & \numprint{12.2} & &\numprint{24387753.0} & \numprint{4.24}\\ 
Long\textunderscore Coup\textunderscore . & &\numprint{11025545.0} & \numprint{29.75} & &\numprint{10416909.0} & \textbf{\numprint{0.12}} & \numprint{44.9} & &\textbf{\numprint{11136959.0}} & \numprint{5.18}\\ 
ML\textunderscore Geer & &\textbf{\numprint{8188887.0}} & \numprint{49.65} & &\numprint{7538983.0} & \textbf{\numprint{0.16}} & \numprint{98.9} & &\numprint{7994470.0} & \numprint{15.56}\\ 
Queen\textunderscore 4147 & &\numprint{26077683.0} & \numprint{126.84} & &\numprint{24331850.0} & \textbf{\numprint{0.24}} & \numprint{126.4} & &\textbf{\numprint{26387860.0}} & \numprint{29.73}\\ 
Serena & &\numprint{13619417.0} & \numprint{29.34} & &\numprint{12970594.0} & \textbf{\numprint{0.13}} & \numprint{42.9} & &\textbf{\numprint{13717579.0}} & \numprint{5.51}\\ 
USA-road-d & &\textbf{\numprint{1427296277.0}} & \numprint{17.93} & &\numprint{1425683461.0} & \textbf{\numprint{1.19}} & \numprint{5.1} & &\numprint{1427192290.0} & \numprint{6.07}\\ 
arabic-2005 & &\textbf{\numprint{1477794453.0}} & \numprint{1120.01} & &\numprint{1477016564.0} & \textbf{\numprint{36.01}} & \numprint{28.1} & &\numprint{1477671476.0} & \numprint{1013.04}\\ 
as-skitter & &\textbf{\numprint{124145588.0}} & \numprint{2.43} & &\numprint{123677864.0} & \textbf{\numprint{0.54}} & \numprint{1.5} & &\numprint{124141373.0} & \numprint{0.79}\\ 
asia.osm & &\textbf{\numprint{703469054.0}} & \numprint{5.72} & &\numprint{703325680.0} & \textbf{\numprint{0.42}} & \numprint{4.5} & &\numprint{703455351.0} & \numprint{1.90}\\ 
buddha & &\textbf{\numprint{57544080.0}} & \numprint{1.61} & &\numprint{57254592.0} & \numprint{0.48} & \numprint{0.7} & &\numprint{57508556.0} & \textbf{\numprint{0.36}}\\ 
circuit5M & &\textbf{\numprint{522198078.0}} & \numprint{6.01} & &\numprint{522186264.0} & \textbf{\numprint{0.54}} & \numprint{4.3} & &\numprint{522196403.0} & \numprint{2.32}\\ 
com-youtube & &\textbf{\numprint{90295285.0}} & \numprint{0.49} & &\numprint{90271270.0} & \textbf{\numprint{0.11}} & \numprint{2.0} & &\textbf{\numprint{90295285.0}} & \numprint{0.22}\\ 
delaunay\textunderscore n. & &\textbf{\numprint{650103112.0}} & \numprint{161.52} & &\numprint{642176344.0} & \textbf{\numprint{0.82}} & \numprint{20.4} & &\numprint{648575756.0} & \numprint{16.71}\\ 
dewiki-2013 & &\textbf{\numprint{76968052.0}} & \numprint{42.06} & &\numprint{76517735.0} & \textbf{\numprint{1.52}} & \numprint{7.9} & &\numprint{74720376.0} & \numprint{11.96}\\ 
dielFilter. & &\textbf{\numprint{9386363.0}} & \numprint{3.11} & &\numprint{8517051.0} & \textbf{\numprint{0.14}} & \numprint{43.0} & &\textbf{\numprint{9386363.0}} & \numprint{5.90}\\ 
enwiki-2013 & &\textbf{\numprint{235386036.0}} & \numprint{68.29} & &\numprint{234445023.0} & \textbf{\numprint{2.58}} & \numprint{8.3} & &\numprint{235008063.0} & \numprint{21.42}\\ 
enwiki-2018 & &\textbf{\numprint{322698398.0}} & \numprint{72.67} & &\numprint{321782538.0} & \textbf{\numprint{3.19}} & \numprint{6.7} & &\numprint{321843964.0} & \numprint{21.45}\\ 
enwiki-2022 & &\textbf{\numprint{366989081.0}} & \numprint{112.53} & &\numprint{365827362.0} & \textbf{\numprint{3.91}} & \numprint{9.2} & &\numprint{366202711.0} & \numprint{35.88}\\ 
frwiki-2013 & &\textbf{\numprint{72589474.0}} & \numprint{25.79} & &\numprint{72238887.0} & \textbf{\numprint{1.44}} & \numprint{6.0} & &\numprint{72468246.0} & \numprint{8.63}\\ 
germany.osm & &\textbf{\numprint{683336494.0}} & \numprint{5.73} & &\numprint{683241943.0} & \textbf{\numprint{0.24}} & \numprint{9.2} & &\numprint{683326451.0} & \numprint{2.17}\\ 
itwiki-2013 & &\textbf{\numprint{58309171.0}} & \numprint{14.28} & &\numprint{58111245.0} & \textbf{\numprint{1.20}} & \numprint{4.0} & &\numprint{58284373.0} & \numprint{4.81}\\ 
ljournal-2. & &\numprint{320046313.0} & \numprint{14.70} & &\numprint{319882745.0} & \textbf{\numprint{1.24}} & \numprint{4.9} & &\textbf{\numprint{320046546.0}} & \numprint{6.06}\\ 
roadNet-CA & &\textbf{\numprint{111345705.0}} & \numprint{1.65} & &\numprint{111221611.0} & \textbf{\numprint{0.29}} & \numprint{1.8} & &\numprint{111325524.0} & \numprint{0.52}\\ 
roadNet-PA & &\textbf{\numprint{61698498.0}} & \numprint{0.88} & &\numprint{61612393.0} & \textbf{\numprint{0.24}} & \numprint{1.1} & &\numprint{61688549.0} & \numprint{0.26}\\ 
roadNet-PA* & &\textbf{\numprint{61723052.0}} & \numprint{0.89} & &\numprint{61640135.0} & \textbf{\numprint{0.24}} & \numprint{1.0} & &\numprint{61710606.0} & \numprint{0.25}\\ 
roadNet-TX & &\textbf{\numprint{78587576.0}} & \numprint{1.08} & &\numprint{78504664.0} & \textbf{\numprint{0.24}} & \numprint{1.4} & &\numprint{78575460.0} & \numprint{0.33}\\ 
soc-LiveJo. & &\textbf{\numprint{284022934.0}} & \numprint{13.34} & &\numprint{283065322.0} & \textbf{\numprint{1.50}} & \numprint{4.3} & &\numprint{283922214.0} & \numprint{6.38}\\ 
soc-flickr. & &\textbf{\numprint{129805768.0}} & \numprint{2.75} & &\numprint{129736468.0} & \textbf{\numprint{0.48}} & \numprint{1.5} & &\numprint{129805546.0} & \numprint{0.73}\\ 
soc-pokec-. & &\textbf{\numprint{83963760.0}} & \numprint{59.57} & &\numprint{83459341.0} & \textbf{\numprint{1.49}} & \numprint{14.8} & &\numprint{83920370.0} & \numprint{22.04}\\ 
wiki-Talk & &\textbf{\numprint{235837346.0}} & \numprint{0.76} & &\numprint{235829379.0} & \textbf{\numprint{0.07}} & \numprint{4.9} & &\textbf{\numprint{235837346.0}} & \numprint{0.34}\\ 
wiki-topca. & &\textbf{\numprint{106674375.0}} & \numprint{12.39} & &\numprint{106292095.0} & \textbf{\numprint{1.36}} & \numprint{2.7} & &\numprint{106654405.0} & \numprint{3.73}\\ 
\cmidrule{1-1}  \cmidrule{3-4}  \cmidrule{6-8}  \cmidrule{10-11} 
\textbf{geo. mean} & &\numprint{87643763.3} & \numprint{14.61} & &\numprint{85802295.0} & \textbf{\numprint{0.49}} & - & &\textbf{\numprint{87687741.9}} & \numprint{4.51}\\ 
\cmidrule{1-1}  \cmidrule{3-4}  \cmidrule{6-8}  \cmidrule{10-11} 
com-Friend. & &\textbf{\numprint{3832609030.0}} & \numprint{2777.90} & &\numprint{3821723789.0} & \textbf{\numprint{21.05}} & - & &- & -\\ 
europe.osm & &\textbf{\numprint{3003671463.0}} & \numprint{25.36} & &\numprint{3003307414.0} & \textbf{\numprint{0.53}} & - & &- & -\\ 
hollywood-. & &\textbf{\numprint{64876723.0}} & \numprint{68.55} & &\numprint{64822962.0} & \textbf{\numprint{2.83}} & \numprint{5.7} & &\numprint{45215004.0} & \numprint{16.02}\\ 
imdb-2021 & &\textbf{\numprint{242337085.0}} & \numprint{1.00} & &\numprint{242320746.0} & \textbf{\numprint{0.32}} & \numprint{1.0} & &\numprint{64162924.0} & \textbf{\numprint{0.32}}\\ 
indochina-. & &- & - & &\numprint{498050460.0} & \textbf{\numprint{4.25}} & \numprint{24.1} & &\textbf{\numprint{498132749.0}} & \numprint{102.51}\\ 
it-2004 & &\textbf{\numprint{2714543237.0}} & \numprint{3071.69} & &\numprint{2713308072.0} & \textbf{\numprint{109.52}} & - & &- & -\\ 
sk-2005 & &- & - & &\textbf{\numprint{3226591750.0}} & \textbf{\numprint{304.24}} & - & &- & -\\ 
soc-sinawe. & &\textbf{\numprint{5872160022.0}} & \numprint{50.07} & &\numprint{5872157313.0} & \textbf{\numprint{1.35}} & - & &- & -\\ 
twitter-20. & &\textbf{\numprint{3025416086.0}} & \numprint{495.77} & &\numprint{3024686549.0} & \textbf{\numprint{7.60}} & - & &- & -\\ 
uk-2005 & &\textbf{\numprint{2509963974.0}} & \numprint{544.34} & &\numprint{2509446931.0} & \textbf{\numprint{37.85}} & - & &- & -\\ 
webbase-20. & &\textbf{\numprint{8432593760.0}} & \numprint{653.93} & &\numprint{8431738031.0} & \textbf{\numprint{25.20}} & - & &- & -\\ 

    \end{tabular}
  }}
  \end{center}
\end{table*}

\fi

\newpage

\end{document}